\documentclass[aps,pre,reprint,amsmath,amssymb,longbibliography,floatfix]{revtex4-2}
\usepackage[T1]{fontenc}
\usepackage{graphicx,bm,booktabs,amsthm,xcolor,mathrsfs}
\usepackage{tikz}
\usetikzlibrary{arrows.meta,positioning}
\usepackage[hidelinks]{hyperref}

\newtheorem{proposition}{Proposition}
\newcommand{\E}{\mathbb E}
\newcommand{\dd}{\,\mathrm d}
\newcommand{\KL}{D_{\mathrm{KL}}}
\newcommand{\Id}{\mathbb I}

\newcommand{\Var}{\operatorname{Var}}
\newcommand{\Cov}{\operatorname{Cov}}
\hypersetup{pdftitle={Hidden kinetic correlations control collective phases of entropy-conditioned histories},pdfauthor={Abhishek Chowdhury}}
\begin{document}
\title{Hidden kinetic correlations control collective phases of entropy-conditioned histories}
\author{Abhishek Chowdhury}
\email{achowdhury@iitbbs.ac.in}
\affiliation{School of Basic Sciences, Indian Institute of Technology Bhubaneswar, India}

\begin{abstract}
Visible dynamics and mean dissipation need not determine the collective
phases of rare histories. We construct local Markov models with identical
ordinary visible path distributions, stationary states, mean entropy
production, mean jump activity and driving affinities, yet different phases
when histories are weighted by their full entropy production. Visible
configurations modulate hidden driven cycles, whose exact elimination turns
kinetic correlations into interactions. A positive spectral reduction
separates hidden relaxation from collective modes. Periodic binary models
retain opposite phases with entropy variances matched at every observation
time. For conserved particles, a pair matches the first two joint cumulants
of entropy and hidden activity at every observation time, yet selects
segregation or a fluctuating state with suppressed long-wavelength density
noise. Bound-state transport and fermionic evolution determine their
contrasting dynamics. At a clean Ising critical midpoint, the entropy source
couples quadratically to the thermal perturbation, making the fourth
cumulant the first to detect the critical susceptibility. Its rate per site
grows logarithmically, with an exact cutoff set by system size and
observation time. For two quenched disorder ensembles, exact sample-response
formulas and numerical evidence support a finite fourth-cumulant density
despite very slow relaxation; a weak-disorder continuum response admits a
controlled Borel sum. A common positive hidden eigenvector characterizes
autonomous selected dynamics. When this condition fails, a finite network
develops visible memory while preserving its ordinary visible process and
mean diagnostics. These results identify kinetic information missed by
aggregate thermodynamic observations.
\end{abstract}
\maketitle

\section{Introduction}
\label{sec:intro}
A visible trajectory can be reversible while internal degrees of freedom
consume fuel. This occurs naturally when a measured conformation, position
or population label omits a driven internal cycle. Inferring dissipation
from such observations is a central problem of stochastic thermodynamics.
Path irreversibility and waiting-time information constrain hidden
dissipation, but their resolving power depends on which events and
variables are observed~\cite{RoldanParrondo2012,SkinnerDunkel2021,DeguentherMeerSeifert2024}.
The collective behavior of rare histories raises a further question.
Even if the full mean dissipation and its noise are known separately,
do they determine how a system organizes when its histories are selected
for reduced irreversibility?

We answer this question by comparing microscopic models under explicit
observational constraints. In our principal matched comparisons, the complete ordinary visible process, full
stationary distribution, mean entropy production, mean number of jumps
and microscopic affinities are held fixed. The examples use several
visible update rules, from independent flips to conserved exchanges;
the observational constraints are imposed within each such family. Changing the dependence of
hidden turnover on visible configurations can nevertheless cross a
collective phase boundary in the entropy-weighted ensemble. Two stronger
comparisons test whether additional aggregate measurements remove this
ambiguity. Periodic models can also have identical ordinary entropy
variance at every observation time and occupy opposite selected phases.
For conserved visible particles, even the first two joint cumulants of
entropy and hidden activity can agree exactly while the selected
organization changes from macroscopic segregation to suppressed
long-wavelength density fluctuations.

The distinction between driving and kinetics is established in
nonequilibrium physics. Symmetric transition prefactors affect response
and state selection at fixed thermodynamic
forces~\cite{Maes2016Nondissipative}. Autonomous sensory networks similarly
separate an upstream stochastic signal from a driven downstream
system~\cite{BaratoHartichSeifert2013}. Kinetically constrained models with
autonomous configurational motion and mobility-gated dissipative cycles
have also been studied through activity-selected histories and suppressed
reaction~\cite{Kaneko2026KineticLocking}. Here the kinetic comparison
preserves the stated ordinary observations and thermodynamic diagnostics.
Its consequence appears after selecting histories by the \emph{full}
entropy, including hidden reactions. The additional information resides
in correlations between visible configurations and hidden event times,
which we quantify through a mixed observable and its path Fisher
information.

The mechanism is accessible in a chain of independently flipping visible
variables coupled to hidden three-state cycles. Alignment slows both
directions of a cycle by the same factor, leaving its affinity unchanged.
For a given visible history the hidden events are conditional Poisson
processes. Their entropy weight penalizes integrated turnover, so
persistent alignment can become favorable even though the ordinary
visible variables do not interact. The resulting spin operator is exact
at finite positive rates. Biasing the interaction energy of independent
spins is already known to produce an Ising trajectory
transition~\cite{JackSollich2010,VasiloiuOakesCarolloGarrahan2020}; entropy
and activity can also resolve different history
phases~\cite{GingrichVaikuntanathanGeissler2014}. The present question is
which of these phases is compatible with the same measured ordinary
process and thermodynamic information. The effective interaction is
fixed by the microscopic entropy histories and their kinetic gates.

The mathematical organization makes the comparison exact and
extendable. Hidden-group averaging identifies a positive visible sector;
a spectral bound establishes when its critical modes govern the full
microscopic relaxation. Symmetric-polynomial identities construct matched
gate moments on opposite sides of the phase boundary. With particle
conservation, the same elimination gives an anisotropic Heisenberg (XXZ) operator. Its bound
states determine the mobility and current fluctuations of a moving
domain, while a determinant of one-particle propagators determines the
hyperuniform partner's evolution. These structures turn the matched
static diagnostics into precise comparisons of selected transport.
For more general hidden clocks, a common positive eigenvector identifies
exactly when selected visible motion stays autonomous. When that
condition fails, a finite network develops visible memory while
preserving the ordinary visible process and mean thermodynamic diagnostics. Eliminating hidden states therefore
has a testable boundary, as well as an exact solvable regime.

A general path-space variational principle gives entropy bias $s=1/2$
a physical meaning. For mutually absolutely continuous forward and
reversed histories, its selected measure is the reversal-invariant
history distribution of smallest relative-entropy cost. Reversible
information geometry and stochastic control organize this principle
\cite{WolferWatanabe2021,Andrieux2025,ChetriteTouchetteControl2015};
Appendix~\ref{app:information} derives the finite-time and stationary-rate
versions.
Tuning that midpoint to criticality makes the entropy source tangent to
the thermal perturbation. The Ising variance density remains finite,
whereas the fourth-cumulant rate per site grows logarithmically with
system size or observation time. Conditional counting explains this
unusual diagnostic and fixes the finite-time crossover. A three-state
Potts realization changes the logarithm into power growth. Quenched gates
pose a different question, since very slow modes need not couple strongly
to the measured entropy. We derive the exact energy-response matrix
elements and find numerical support for a finite fourth-cumulant density
despite activated relaxation. A weak-disorder continuum limit explains
how a response can be smooth and nonanalytic, with a divergent expansion
whose Borel sum and truncation error are controlled. The critical
calculation thus depends on the observable's coupling to slow modes, not
on relaxation times alone.

Sections~\ref{sec:model}--\ref{sec:phases} develop the microscopic
construction, exact reduction and uniform-chain comparison.
Section~\ref{sec:periodic-profiles} strengthens the observational
constraints with spatially varying gates.
Sections~\ref{sec:critical} and \ref{sec:potts} derive the entropy
signatures, followed by the conserved comparison in
Sec.~\ref{sec:conserved}. Section~\ref{sec:extensions} connects duality,
general hidden clocks and further local dynamics.
Section~\ref{sec:selected-memory} determines when selection creates
visible memory. Appendix~\ref{app:information} develops the path-measure and rate
variational principles. Appendices~\ref{app:ising}--\ref{app:gate-inference}
derive the binary spectrum, duality, local extensions, counting response
and inference formulas. Appendices~\ref{app:conserved} and
\ref{app:droplet-transport} treat conserved states and transport;
Appendices~\ref{app:general-hidden}--\ref{app:selected-memory} treat
general hidden clocks, quenched response and selected memory.

\section{A driven network with autonomous visible dynamics}
\label{sec:model}
Consider binary variables $z_i=\pm1$ on the vertices of a finite simple graph
with $N$ vertices and edge set $\mathcal E$. Each variable flips independently
at rate $\gamma>0$. A hidden cycle sits on each edge $e=(i,j)$. Its state is
$h_e\in\mathbb Z_3=\{0,1,2\}$, with addition modulo three.
The rates $u_e$ for $h_e\to h_e+1$ and $v_e$ for $h_e\to h_e-1$
depend on the visible alignment $P_e=z_i z_j$,
\begin{equation}
 u_e(z)=r(1-\eta P_e)e^{A/2},\quad
 v_e(z)=r(1-\eta P_e)e^{-A/2},
 \label{eq:rates}
\end{equation}
where $r>0$, $A>0$ and $0\leq\eta<1$. All rates are finite and positive.
The dimensionless parameter $A$ is the entropy increment of one clockwise
jump; the affinity around a complete three-step cycle is $3A$.
Rates $r$ and $\gamma$ have units of inverse time. Entropies are in units
with Boltzmann's constant equal to one.

\begin{figure*}[t]
\centering
\includegraphics[width=\textwidth]{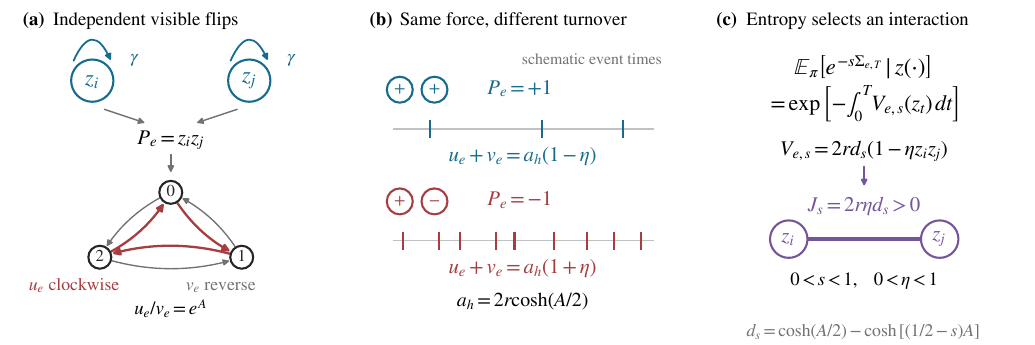}
\caption{From hidden kinetic coupling to an interaction between selected
histories. (a)   The visible variables flip independently at rate $\gamma$;
their product gates both directions of a hidden three-state cycle while
the ratio $u/v=e^A$ stays fixed. (b) Alignment reduces the conditional
hidden activity to $a_h(1-\eta)$ and opposition increases it to
$a_h(1+\eta)$, where $a_h=2r\cosh(A/2)$. The event marks are schematic.
(c) Summing the hidden entropy $\Sigma_{e,T}$ on the illustrated edge
exactly multiplies the visible path weight by
$\exp\left[-\int_0^T V_{e,s}\,dt\right]$. Here $s$ is the entropy bias and
$d_s=\cosh(A/2)-\cosh\left[(1/2-s)A\right]$. For $0<s<1$ and $\eta>0$, the
alignment part of this potential is the ferromagnetic interaction
$-J_s z_i z_j$, with $J_s=2r\eta d_s>0$. The scalar cost is retained. Forward arrows in (a) follow
$0\to1\to2\to0$, and reverse arrows follow the opposite orientation.
The solid visible bond belongs to the selected weight; the ordinary
visible rates remain independent.}

\label{fig:model}
\end{figure*}

We call $1-\eta P_e$ a kinetic gate, meaning the configuration-dependent
factor multiplying both directional rates. It changes the frequency of
hidden events while preserving their forward-to-reverse ratio. Aligned visible variables slow the hidden cycle; opposite
variables accelerate it. A physical interpretation is a fuel-driven
internal cycle whose transition barriers depend on neighboring
conformations. Local detailed balance can assign the same chemical entropy
$A$ to every forward step while the conformations remain energetically
degenerate~\cite{Maes2016Nondissipative}. The model isolates how this conformation-dependent turnover affects rare histories. The labels $z_i$ represent configurations, such as conformations.
Reversal reads a history in the opposite temporal order and reverses
its jumps while leaving each configuration label unchanged.
The operation $z_i\mapsto-z_i$ is a separate spin-inversion symmetry.

The visible process is autonomous. Hidden jumps never change $z$, and the
visible rates never depend on $h$. Under stationary preparation, its complete path
distribution is the same product of independent telegraph processes for
every $(r,A,\eta)$. At fixed $z$, a cycle with spatially uniform clockwise
and counterclockwise rates has a uniform stationary distribution, even
when its net clockwise circulation  is nonzero. The full stationary distribution is therefore
\begin{equation}
 \pi(z,h)=2^{-N}3^{-|\mathcal E|}.
 \label{eq:stationary}
\end{equation}
 Here $|\mathcal E|$ is the number of graph edges, hence the number of
hidden cycles. Equation~\eqref{eq:stationary} fixes the configuration
probabilities, while Eq.~\eqref{eq:rates} still changes full-state temporal correlations.

Let $K_{e,T}^{+}$ and $K_{e,T}^{-}$ count clockwise $(+)$ and counterclockwise $(-)$ hidden jumps on edge
$e$ in a time interval of length $T$, and let
$K_{e,T}=K_{e,T}^{+}+K_{e,T}^{-}$. For stationary histories, the logarithm
of the ratio of forward and reversed path probabilities is exactly
\begin{equation}
 \Sigma_T=A\sum_{e\in\mathcal E}(K_{e,T}^{+}-K_{e,T}^{-}).
 \label{eq:entropy}
\end{equation}
Visible flips and the stationary endpoint term contribute zero.
The mean entropy production and hidden jump activity per edge are
\begin{equation}
 \sigma=2rA\sinh(A/2),\qquad a_h=2r\cosh(A/2).
 \label{eq:means}
\end{equation}
Neither depends on $\eta$, since $\E_\pi\left[P_e\right]=0$.
The mean total jump rate is $N\gamma+|\mathcal E|a_h$.
Thus changing $\eta$ preserves the full stationary state, every visible
multi-time statistic, the mean full entropy production, the mean full
activity and every cycle affinity. We will determine whether its
entropy-conditioned phases are fixed by these ordinary diagnostics.

\section{Eliminating hidden histories exactly}
\label{sec:reduction}

All underlying microscopic processes are classical continuous-time Markov
processes. For microscopic states $x,y$, we use the backward generator
\begin{equation}
 (Lf)(x)=\sum_{y\ne x}k_{xy}\left[f(y)-f(x)\right],\qquad L\mathbf1=0.
 \label{eq:backward-convention}
\end{equation}
It evolves observables $f$. A column of state probabilities instead
obeys $\dot p=L^{\mathsf T}p$, with $p_x\ge0$ and
$\|p\|_1=\sum_xp_x=1$. The spin Hamiltonians below represent tilted
classical path weights. Their $L^2$ eigenvectors organize this evolution;
physical probabilities will follow from a positive similarity transform.

The appropriate statistical ensemble weights a full history by
$e^{-s\Sigma_T}$, with real entropy bias $s$. Its normalization and
long-time scaled cumulant generating function (SCGF) are
\begin{equation}
 Z_{N,T}(s)=\E_\pi\left[e^{-s\Sigma_T}\right],\enspace
 \psi_N(s)=\lim_{T\to\infty}T^{-1}\log Z_{N,T}(s).
 \label{eq:scgf}
\end{equation}
This is a statistical selection of histories. A Markov process realizing
the selected histories in their long-time bulk is  obtained below.
Throughout, entropy conditioning means this normalized exponential
weight. It differs from imposing an exact value of $\Sigma_T$, which
would remove the fluctuations of that same integrated observable.

\subsection{Conditional Poisson processes and a local interaction}
For a specified visible history $z(t)$, the two directions of every hidden
cycle are independent inhomogeneous Poisson processes. Their intensities
depend on $z(t)$ but not on the hidden cycle position. Summing these hidden
histories gives, at every finite $T$,
\begin{align}
 \E\left[e^{-s\Sigma_T}\mid z(\cdot)\right]
 &=\exp\left[-\int_0^T V_s(z(t))\dd t\right],\label{eq:conditional}\\
 V_s(z)&=2r d_s\sum_{e\in\mathcal E}(1-\eta P_e),\label{eq:potential}\\
 d_s&=\cosh(A/2)-\cosh\left[(1/2-s)A\right].\label{eq:ds}
\end{align}
The expectation in this identity is under the ordinary process, initially
in Eq.~\eqref{eq:stationary}. Later expectations with subscript $s$
refer to the stationary driven process defined below; disorder averages
will carry the separate subscript \emph{dis}.
For $0<s<1$, $d_s>0$. The entropy weight penalizes histories according to
their hidden turnover, which depends on visible alignment. This is the
mechanism through which an interaction enters a process with independent
ordinary visible dynamics.

For $0\le s\le1$, $V_s\ge0$, so the same weight can be read as survival
under removal at a state-dependent rate $V_s$. This is the killing
interpretation of the Feynman--Kac formula. For general real $s$ the
potential may have either sign.

Write $X_i$ for the operator that flips $z_i$ and $Z_i$ for multiplication
by $z_i$. The visible Feynman--Kac operator, whose potential accumulates
along a path, is
\begin{equation}
 H_s=\gamma\sum_i(\Id-X_i)+2r d_s\sum_{(i,j)\in\mathcal E}
 (\Id-\eta Z_i Z_j).
 \label{eq:H}
\end{equation}
The matrix $H_s$ is real symmetric in the counting inner product
$\langle f,g\rangle=\sum_z f(z)^*g(z)$, and $\mathbf1_z=1$.
With the ordinary visible probability $2^{-N}$,
Eq.~\eqref{eq:conditional} becomes
\begin{equation}
 Z_{N,T}(s)=2^{-N}\mathbf1^{\mathsf T}e^{-T H_s}\mathbf1.
 \label{eq:finiteT}
\end{equation}
No separation of time scales or approximation has entered this reduction.
It removes $3^{|\mathcal E|}$ hidden states while retaining their full
contribution to the generating function.

\subsection{A positive projection and the full spectrum}
The finite-time argument has an algebraic counterpart that also
determines whether hidden modes contribute to the slow critical spectrum.
Let $U_e$ act on a backward observable by
$(U_e f)(z,h)=f(z,h+\boldsymbol e_e)$, where $\boldsymbol e_e$ advances
cycle $e$ by one step modulo three. Its transpose acts on probability columns. These operators generate
the finite Abelian group $G=(\mathbb Z_3)^{|\mathcal E|}$. Averaging over
this group,
\begin{equation}
 \mathcal P=3^{-|\mathcal E|}\sum_{g\in G}U_g,
 \label{eq:projector}
\end{equation}
averages hidden positions at fixed visible configuration and preserves
nonnegative functions. Visible functions are closed under pointwise products and include the
constant function, so the range of this projection is an algebra. That algebra is invariant under
the ordinary generator, so projection defines an exact Markov process,
rather than just closing a selected set of means. The entropy-tilted operator commutes with every
$U_e$ and leaves this algebra invariant, now with the Feynman--Kac potential  
in Eq.~\eqref{eq:potential}.

An entropy tilt multiplies each off-diagonal rate by the entropy weight
of that jump and keeps its original escape diagonal. It therefore need
not conserve probability. In the convention of Eq.~\eqref{eq:backward-convention},
the negative full tilted generator is
\begin{equation}
 \begin{split}
 \mathscr H_s={}&\gamma\sum_i(\Id-X_i)\\
 &+\sum_e r(\Id-\eta Z_i Z_j)\left[\vphantom{e^{-(1/2-s)A}U_e^{-1}}2\cosh(A/2)\Id\right.\\
 &\hspace{18mm}\left.-e^{(1/2-s)A}U_e-e^{-(1/2-s)A}U_e^{-1}\right].
 \end{split}
 \label{eq:fullH}
\end{equation}
The invariant sector $U_e=1$ is precisely Eq.~\eqref{eq:H}.
 In the Doi--Peliti representation, let $b_i|1\rangle=|0\rangle$ and
$b_i|0\rangle=0$ on occupations $n_i=0,1$. These hard-core operators
enforce single occupancy through $b_i^2=0$ and
$\{b_i,b_i^\dagger\}=1$; operators on distinct sites commute. Then $X_i=b_i+b_i^\dagger$ and
$Z_i=1-2b_i^\dagger b_i$, on the two-dimensional local occupation space.
The hidden variables use a three-dimensional clock representation.
Equation~\eqref{eq:fullH} is an exact finite local-state stochastic
Hamiltonian; introducing unrestricted bosons would require projecting
back to these local spaces.

The lowest eigenvector $\phi_s(z)$ of the real symmetric $H_s$ is strictly
positive. Its lift, constant in $h$, is therefore the positive Perron
eigenvector of the full irreducible tilted generator. Thus
\begin{equation}
 \psi_N(s)=-E_N(s),\qquad E_N(s)=\min\operatorname{spec}H_s.
 \label{eq:perron}
\end{equation}
Positivity identifies this reduced eigenpair as the one controlling the
full entropy-generating function.

There is also quantitative control of every discarded sector. The
characters of $G$, its discrete Fourier modes, assign $U_e=e^{\mathrm i\theta_e}$ with
$\theta_e=0,\pm2\pi/3$. In such a sector, $H_s$ acquires the diagonal term
\begin{equation}
 \begin{split}
 \sum_e2r(1-\eta P_e)\Bigl\{&\cosh\left[(1/2-s)A\right](1-\cos\theta_e)\\
 &-\mathrm i\sinh\left[(1/2-s)A\right]\sin\theta_e\Bigr\}.
 \end{split}
 \label{eq:sectors}
\end{equation}
Taking the real part of the Rayleigh quotient of any sector eigenvector
shows that a sector with $m$ nonzero characters obeys
\begin{equation}
 \operatorname{Re} E\geq E_N(s)+
 3m r(1-\eta)\cosh\left[(1/2-s)A\right].
 \label{eq:hiddengap}
\end{equation}
The bound is independent of $N$. Hidden modes therefore stay separated
from critical visible modes at fixed positive rates and $\eta<1$.
This Fourier organization makes the exact reduction useful for dynamics
as well as stationary correlations.

\subsection{The dynamics of selected histories}

At finite observation time, the selected probability of a full history is
$dQ_{s,T}=Z_{N,T}(s)^{-1}e^{-s\Sigma_T}dP_T$, where $P_T$ denotes its
ordinary stationary probability. Far from the endpoints of a long history,
$Q_{s,T}$ approaches a time-homogeneous driven process at each fixed $N$.
Its generator and stationary distribution are distinct objects from the
nonconservative tilt and its eigenvector.

Normalize $\sum_z\phi_s(z)^2=1$. The generalized Doob transform uses this
positive eigenvector to turn the tilted evolution into a conservative
Markov process~\cite{ChetriteTouchetteControl2015}. Its rates are
\begin{align}
 w_s(z\to z^i)&=\gamma\,\frac{\phi_s(z^i)}{\phi_s(z)},\label{eq:doobvis}\\
 u_{e,s}(z)&=r(1-\eta P_e)e^{(1/2-s)A},\nonumber\\
 v_{e,s}(z)&=r(1-\eta P_e)e^{-(1/2-s)A}.
 \label{eq:doobhidden}
\end{align}
Here $z^i$ is obtained by flipping the $i$th variable. The stationary
distribution is $\phi_s(z)^2/3^{|\mathcal E|}$.
The driven visible process is reversible for every real $s$, whereas the
full process is reversible at $s=1/2$, when the hidden cycle affinities
vanish~\cite{PriorityBonancaJarzynski2016}. The collective visible rates in Eq.~\eqref{eq:doobvis} describe
the selected long histories and replace the constant visible flip rate
$\gamma$. At finite $T$ the exact selected measure also has endpoint effects.

For visible functions, the conservative backward generator is
\begin{equation}
 L_s^{\rm D}f=-\phi_s^{-1}H_s(\phi_s f)+E_N(s)f,
 \qquad L_s^{\rm D}\mathbf1=0.
 \label{eq:doob-generator-convention}
\end{equation}
Its stationary visible probability is $p_s^{\rm D}(z)=\phi_s(z)^2$.
Indeed $p_s^{\rm D}(z)w_s(z,z^i)=\gamma\phi_s(z)\phi_s(z^i)$ is
symmetric in its endpoints. This detailed-balance identity explains the
square and its normalization as a classical probability. For a nonsymmetric tilted matrix the
corresponding stationary probability is the product of its positive
left and right Perron vectors, normalized to sum to one.

The same transform relates stationary visible correlations to the
spin spectrum. Here $\langle\cdot\rangle_s$ denotes a stationary driven expectation.
For a diagonal observable $F$ and physical time separation $t\geq0$,
\begin{equation}
 \langle F(t)F(0)\rangle_s
 =\langle\phi_s|F e^{-t(H_s-E_N)}F|\phi_s\rangle.
 \label{eq:random-stochastic-correlation}
\end{equation}
Subtracting $\langle F\rangle_s^2$ gives the connected correlation.
Thus $e^{-t(H_s-E_N)}$, often called imaginary-time evolution in spin
language, describes correlations at physical time $t$ of the driven
Markov process. Excitation energies determine relaxation rates,
while observable matrix elements determine which modes a measurement sees.

Appendix~\ref{app:finite-duration-bulk} gives the finite-duration rates
and marginal explicitly. It shows how the $L^2$ lowest mode becomes an
$L^1$ stationary probability in the long-time interior.

\section{Identical ordinary diagnostics, different collective phases}
\label{sec:phases}
\begin{figure*}[t]
\includegraphics[width=\textwidth]{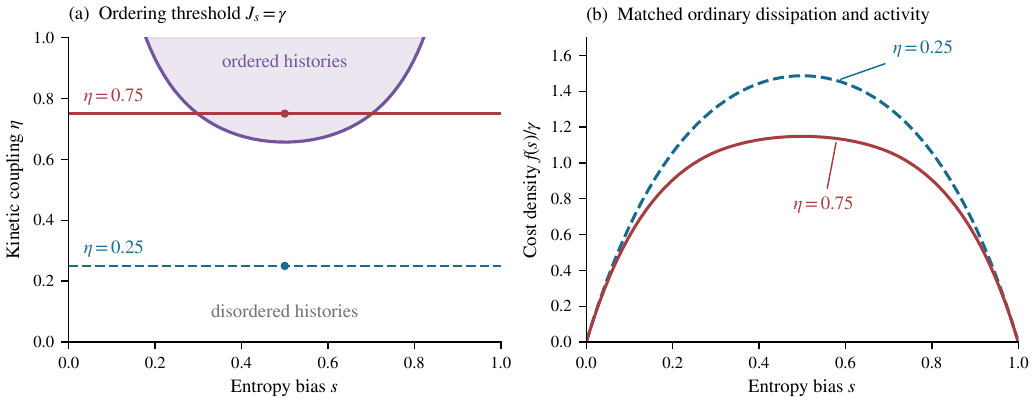}
\caption{Hidden kinetic coupling changes the phase while the ordinary
diagnostics stay fixed. (a) Entropy-conditioned phase diagram at
$\gamma=1$, $A=4$, $r=1/\sinh2$. The shaded region is $J_s>1$; its
tip is at $\eta_c=1/(2\tanh1)$. Both horizontal examples have
$\sigma=8$ and $a_h=2\coth2$. (b) Their exact cost densities from
Eq.~\eqref{eq:freeenergy}. The scalar term is retained, so both curves
have the same slope $\sigma$ at $s=0$, although their midpoint phases
differ.}
\label{fig:phases}
\end{figure*}

Take a periodic nearest-neighbor chain with $N\geq3$ and $|\mathcal E|=N$. Equation~\eqref{eq:H}
is a transverse-field Ising operator with a specified scalar term,
\begin{align}
 H_s&=N\left[\gamma+2r d_s\right]\Id-\gamma\sum_iX_i-J_s\sum_iZ_iZ_{i+1},
 \label{eq:ising}\\
 J_s&=2r\eta d_s.\label{eq:J}
\end{align}
The scalar is essential to the absolute trajectory probability and its
entropy derivatives. It does not affect $\phi_s$.
The Ising solution~\cite{CriticalPfeuty1970} determines both this cost and
the selected correlations. Its trajectory interpretation for independent
spins~\cite{JackSollich2010,VasiloiuOakesCarolloGarrahan2020} becomes a
microscopic entropy comparison through the coupling $J_s$ in
Eq.~\eqref{eq:J}.

We first take $T\to\infty$ at finite $N$, then $N\to\infty$ at fixed
$\gamma,r,A,\eta$. The cost density in this limit is
\begin{equation}
 \begin{split}
 f(s)&=-\lim_{N\to\infty}\frac{\psi_N(s)}N\\
 &=\gamma+2r d_s-\frac1{2\pi}\int_0^\pi\varepsilon_s(k)\dd k,\\
 \varepsilon_s(k)&=2\sqrt{\gamma^2+J_s^2-2\gamma J_s\cos k}.
 \end{split}
 \label{eq:freeenergy}
\end{equation}
Here $\varepsilon_s(k)$ is the Ising fermion excitation energy.
For $0<s<1$, the selected histories are ordered when $J_s>\gamma$.
In the bulk stationary distribution, this means
\begin{equation}
 \lim_{\ell\to\infty}\langle z_i z_{i+\ell}\rangle_s=m_s^2,
 \qquad m_s=(1-\gamma^2/J_s^2)^{1/8}.
 \label{eq:order}
\end{equation}
Finite rings  retain spin-flip symmetry, so the two-point definition
avoids assigning them a spontaneous nonzero one-point mean.
The transition at $J_s=\gamma$ is a genuine infinite-system singularity;
the finite-size SCGF remains analytic in the parameters considered.

For $0\le J_s<\gamma$, the long-distance correlation vanishes.
At $J_s=\gamma$ the correlation length diverges and the order parameter
turns on continuously with exponent $1/8$~\cite{CriticalPfeuty1970}.
These phases characterize configurations at bulk times of the
entropy-selected classical histories.

\begin{proposition}[Matched diagnostics and opposite phases]
\label{prop:matched}
Fix $\gamma,r,A>0$, and vary $0\leq\eta<1$ in Eq.~\eqref{eq:rates}.
All members have identical visible path distributions, mean full entropy
production, mean full activity and cycle affinities. Their entropy-midpoint
ensembles include both an ordered and a disordered member if and only if
\begin{equation}
 J_{\max}=2r\left[\cosh(A/2)-1\right]
          =a_h\left[1-\operatorname{sech}(A/2)\right]>\gamma.
 \label{eq:criterion}
\end{equation}
The separating value is $\eta_c=\gamma/J_{\max}$.
\end{proposition}
\begin{proof}
The matched statistics follow from Eqs.~\eqref{eq:stationary} and
\eqref{eq:means} and autonomy of the visible generator. At the midpoint,
$J_{1/2}=\eta J_{\max}$. The Ising criterion then proves both directions.
Strict inequality in Eq.~\eqref{eq:criterion} is required because all
hidden rates are positive and $\eta<1$. The member $\eta=0$ is disordered.
\end{proof}

Figure~\ref{fig:phases} illustrates an explicit pair. Choose units $\gamma=1$, set $A=4$ and
$r=1/\sinh2$. Both members have $\sigma=8$ and $a_h=2\coth2$.
With $\eta=1/4$ the midpoint coupling is $\tfrac12\tanh1<1$;
with $\eta=3/4$ it is $\tfrac32\tanh1>1$.
This phase distinction survives despite complete agreement of the
ordinary diagnostics stated in Proposition~\ref{prop:matched}.

The reversible midpoint has a direct information-theoretic meaning.
Let $P_T$ be the full stationary path measure and $P_T^R$ its time
reversal. For a trial history measure $Q$, the Kullback--Leibler divergence
$\KL(Q\Vert P_T)=\int dQ\log(dQ/dP_T)$ measures its logarithmic
probability cost. Among all reversal-invariant $Q$, the minimizer is
\begin{equation}
 \begin{gathered}
 Q_{1/2,T}=\frac{\sqrt{P_T P_T^R}}{Z_{N,T}(1/2)},\\
 \min_{Q=Q^R}\KL(Q\Vert P_T)=-\log Z_{N,T}(1/2).
 \end{gathered}
 \label{eq:KLprojection}
\end{equation}
This finite-time optimization ranges over history measures, including
ones with time-dependent endpoint effects. The corresponding rate minimum over
stationary reversible Markov dynamics is $E_N(1/2)$, attained by
Eqs.~\eqref{eq:doobvis} and \eqref{eq:doobhidden} at the midpoint.
Reversible information projections organize Markov kernels
\cite{WolferWatanabe2021} and quantify distance from equilibrium
\cite{Andrieux2025}; the driven-process variational principle relates
path selection to optimal control~\cite{ChetriteTouchetteControl2015};
Appendix~\ref{app:information} gives the needed continuous-time derivation.
Proposition~\ref{prop:matched} thus concerns the collective organization
of the least costly reversible histories.

The origin of the effective interaction can be stated directly on path
space. Let $y$ be a visible history and $h$ its hidden complement, with
observation commuting with time reversal. Products of measures below
mean products of densities with respect to a common reference measure.
For $0\leq s\leq1$, define
the overlap of the conditional hidden history distributions by
\begin{equation}
 \mathcal A_s(y)=\int
 \left[dP_T(h\mid y)\right]^{1-s}
 \left[dP_T^{\rm R}(h\mid y)\right]^s.
 \label{eq:app-conditional-affinity}
\end{equation}
The visible marginal $Q_s^{\rm vis}$ of the normalized full selected
measure is therefore
\begin{equation}
 dQ_s^{\rm vis}(y)
 =\frac{
 \left[dP_T^{\rm vis}(y)\right]^{1-s}
 \left[d(P_T^{\rm vis})^{\rm R}(y)\right]^s
 \mathcal A_s(y)}{Z_T(s)}.
 \label{eq:app-visible-conditional-tilt}
\end{equation}
Here $Z_T(s)$ is the normalizing path integral, equal to $Z_{N,T}(s)$
for the present model. Selection and observation commute precisely when $\mathcal A_s(y)$
is independent of the visible history on the selected support.
In our model the ordinary visible process is reversible, so selecting
its own entropy leaves it unchanged. Selecting the full entropy instead
weights it by the history-dependent overlap in Eq.~\eqref{eq:conditional}.
Appendix~\ref{app:conditional-overlap} relates this obstruction to the
relative-entropy chain rule.

\subsection{Dissipation and activity boundaries}
The criterion can be expressed through measured mean quantities.
The ratio $\sigma/a_h=A\tanh(A/2)$ fixes $A>0$ uniquely. It does not
fix $\eta$. At specified $a_h$ and $\sigma$, an ordered member exists
precisely when $a_h>\gamma$ and
\begin{equation}
 \sigma>\sigma_c(a_h),\quad
 \sigma_c(a_h)=2\sqrt{\gamma(2a_h-\gamma)}
 \operatorname{arcosh}\frac{a_h}{a_h-\gamma}.
 \label{eq:activityfrontier}
\end{equation}
This follows by solving $a_h\left[1-\operatorname{sech}(A/2)\right]=\gamma$
and using Eq.~\eqref{eq:means}. At equality the limiting value
$\eta=1$ would be required, outside the positive-rate domain.

If activity is unconstrained, the same family has an ordered member
at a specified entropy rate if and only if $\sigma>4\gamma$, since
\begin{equation}
 J_{1/2}=\eta\sigma\frac{\tanh(A/4)}A<\frac\sigma4.
 \label{eq:entropybudget}
\end{equation}
The function $\tanh(A/4)/A$ decreases from $1/4$ to zero. Choosing
sufficiently small positive $A$ and $\eta$ sufficiently close to one
proves sufficiency when $\sigma>4\gamma$.
Approaching this lower boundary requires growing hidden activity.
The factor $1/4$ is also familiar in entropy-fluctuation
bounds~\cite{PietzonkaBaratoSeifert2016,GingrichHorowitzPerunovEngland2016}.
Here the specific result is its sharp role in the phase criterion of
this kinetic family, with the finite-activity boundary in
Eq.~\eqref{eq:activityfrontier}.

When $J_{1/2}>\gamma$, the two critical entropy biases are
\begin{equation}
s_\pm=\frac12\pm\frac1A\operatorname{arcosh}
 \left[\cosh(A/2)-\frac{\gamma}{2r\eta}\right].
 \label{eq:edges}
\end{equation}
The selected histories are ordered for $s_-<s<s_+$. The symmetry $s\leftrightarrow1-s$ is
time-reversal symmetry of the entropy-generating function, and the
edges meet at $s=1/2$ when $\eta=\eta_c$.

\subsection{Which additional observation distinguishes them?}
The parameter omitted by the mean diagnostics is a correlation between
hidden turnover and visible alignment. At ordinary stationarity,
\begin{equation}
 \lim_{T\to\infty}\frac1T\E_\pi\left[\int_0^T P_e(t)\dd K_{e,t}\right]
 =-\eta a_h.
 \label{eq:crossreadout}
\end{equation}
The value of $P_e$ at a hidden jump is unambiguous because that jump
leaves $z$ unchanged. Together with $a_h$, this mixed first moment
identifies $\eta$ and therefore the phase criterion. It requires
joint access to visible configurations and hidden event times.

Even an entropy fluctuation measurement without state-resolved hidden
activity distinguishes members with different $\eta\geq0$.
For any finite simple graph in the binary model,
\begin{equation}
 \lim_{T\to\infty}\frac{\Var_\pi\Sigma_T}{T}
 =|\mathcal E|\left[A^2a_h+\frac{\eta^2\sigma^2}{2\gamma}\right].
 \label{eq:ordinaryvariance}
\end{equation}
Conditional Poisson noise gives the first term. The second is the
fluctuation of the integrated conditional entropy rate
$\sigma\sum_e(1-\eta P_e)$. Independence of the visible spins implies
$\Cov(P_e(t),P_f(0))=\delta_{ef}e^{-4\gamma|t|}$, including for edges
that share a vertex, which proves Eq.~\eqref{eq:ordinaryvariance}.
Purely visible histories carry no information about $\eta$, but these
additional observations measure the kinetic correlation they omit.

For event-resolved observations this gain can be quantified by the Fisher
information, the mean squared derivative of the log likelihood. For the
parameters $(\eta,a_h,A)$, the ordinary stationary path information is
\begin{equation}
 \frac{\mathcal I_T}{|\mathcal E|T}
 =\operatorname{diag}\left(
 \frac{a_h}{1-\eta^2},\frac1{a_h},\frac{a_h}{4\cosh^2(A/2)}\right).
 \label{eq:gate-fisher}
\end{equation}
The gate  can be inferred locally from the joint observations even when $a_h$ and
$A$ must also be estimated. These latter unknowns are nuisance parameters
for inference of $\eta$. The score, the derivative of the path log
likelihood, has zero covariance with their scores. Visible paths alone have exactly zero
information about $\eta$. Appendix~\ref{app:gate-inference} derives this
specialization of path-space information theory~\cite{Pantazis2013} and
distinguishes event-resolved information from aggregate moment measurements.

\section{Spatial organization beyond matched moments}
\label{sec:periodic-profiles}

The preceding identification used a single uniform gate. For a spatial profile on the same binary ring, a few aggregate
measurements can leave the phase-controlling product undetermined. Allow a fixed periodic gate profile $\eta_i\in(0,1)$,
while keeping $\gamma,r,A$ common. Every edge still has the same entropy
and activity means in Eq.~\eqref{eq:means}. Write
$J_0(s)=2rd_s$ for the interaction scale before multiplication by the
gate. The exact Ising bonds are $J_i(s)=J_0(s)\eta_i$.
For $0<s<1$, repeat a fixed period of length $n$ and take the number of periods to infinity. The phase depends on the geometric mean,
\begin{equation}
 J_0(s)\left(\prod_{i=1}^{n}\eta_i\right)^{1/n}
 \gtrless\gamma
 \quad\text{for order/disorder}.
 \label{eq:periodic-criterion}
\end{equation}
The product arises from a boundary mode of the Ising operator on an
auxiliary semi-infinite open chain with the same periodic bonds.
A Majorana mode is a self-adjoint linear combination of fermion creation
and annihilation operators. At zero energy, its coefficient on successive
sites is multiplied by $\gamma/J_i$. Normalizability after successive periods
gives Eq.~\eqref{eq:periodic-criterion}, consistently with the exact
surface-magnetization formula~\cite{IgloiRieger1998RandomWalks}.

To keep the first two gate moments fixed while changing their product,
encode three gates as roots of a monic cubic. Its first two coefficients
fix those moments, while its constant coefficient controls the product.
We therefore vary the constant term and choose $\eta_i(c)$ to be the roots of 
\begin{equation}
 p_c(x)=(x-\tfrac12)^3-\frac{x-\tfrac12}{16}+c,
 \qquad c_\pm=\pm\frac1{256}.
 \label{eq:periodic-cubic}
\end{equation}
Their discriminant is $37/65536>0$. The two stationary points lie
between zero and one, and $p_c(0)<0<p_c(1)$, so all three roots belong
to $(0,1)$. Newton's identities, which relate polynomial coefficients
to sums of powers of their roots, give
\begin{equation}
 \sum_i\eta_i=\frac32,\quad
 \sum_i\eta_i^2=\frac78,\quad
 \prod_i\eta_i(c_\pm)=\frac{24\mp1}{256}.
 \label{eq:periodic-moments}
\end{equation}
Repeat each triplet on a ring with $N$ divisible by three and choose
$J_0(1/2)/\gamma=(32/3)^{1/3}$.
The cubed ratios in Eq.~\eqref{eq:periodic-criterion} are $23/24$
and $25/24$. The two midpoint ensembles are therefore in opposite
phases in the infinite-chain limit at finite positive microscopic rates.

In addition to the matched ordinary diagnostics, both profiles have the
same aggregate turnover--alignment rate
$T^{-1}\E_\pi\left[\sum_i\int_0^T P_i\dd K_i\right]=-a_h\sum_i\eta_i$. They also have identical ordinary entropy variance
at every finite observation time,
\begin{align}
 \Var_\pi\Sigma_T={}&NA^2a_hT
 +\sigma^2\sum_i\eta_i^2
 \left[\frac{T}{2\gamma}
 -\frac{1-e^{-4\gamma T}}{8\gamma^2}\right].
 \label{eq:periodic-finite-variance}
\end{align}
Conditional Poisson variance and the covariance
$\Cov(P_i(t),P_j(0))=\delta_{ij}e^{-4\gamma|t|}$ prove this identity.
Edge-resolved mixed observations can still distinguish the profiles.

The algebraic mechanism extends to any finite list of gate moments.
Start with $n$ distinct roots in $(0,1)$ whose product is critical and
change only the constant coefficient of their monic polynomial.
For a sufficiently small change the roots remain real and positive.
Newton's identities preserve their power sums through order $n-1$,
whereas their product crosses the phase boundary. This statement concerns
gate moments; higher trajectory cumulants can also depend on spatial
multi-edge correlations.

\section{Critical entropy fluctuations and relaxation}
\label{sec:critical}
We return to the binary ring with a spatially uniform gate. The critical
predictions below concern stationary driven histories, first at fixed $N$
and then as $N$ grows; the ordinary visible dynamics still consists of
independent flips. First, ordinary spectral relaxation remains fast. At $s=0$
the visible generator has gap $2\gamma$, while Eq.~\eqref{eq:hiddengap}
bounds the remaining sectors. The real-part spectral gap of the full
ordinary generator is at least
\begin{equation}
 \min\{2\gamma,3r(1-\eta)\cosh(A/2)\},
 \label{eq:ordinarygap}
\end{equation}
uniformly in $N$. At the critical midpoint, the visible Ising gap is
$2\gamma\tan\left[\pi/(4N)\right]$. For all sufficiently large $N$ it is smaller
than the hidden-sector bound and hence equals the full driven gap.
Critical slowing down with dynamical exponent one therefore occurs in
the full entropy-conditioned system despite the size-independent
ordinary gap.

\subsection{A quadratic entropy source}
The meeting of the two critical bias edges has a less familiar consequence
for entropy measurements. Set $\delta=s-1/2$ and tune $J_{1/2}=\gamma$.
Then
\begin{equation}
 J_s-\gamma=-r\eta A^2\delta^2+O(\delta^4).
 \label{eq:tangent}
\end{equation}
The thermal direction here is the energy perturbation $J-\gamma$,
rather than a change of bath temperature. The entropy source approaches
the Ising critical surface quadratically,
rather than linearly. The Ising singularity
$(J-\gamma)^2\log|J-\gamma|$ consequently becomes
\begin{equation}
 f_{\mathrm{sing}}(s)=\frac{r^2\eta^2A^4}{\pi\gamma}
 \delta^4\log|\delta|,
 \label{eq:fourthlog}
\end{equation}
up to terms analytic in $\delta$ and higher singular orders.
Thus the variance density stays finite at this critical point, while
the fourth entropy cumulant detects a logarithmic singularity.
An ordinary crossing of either edge away from the midpoint instead
couples linearly to the thermal direction.

The counting mechanism gives a complementary explanation. We measure
the original entropy observable in the stationary midpoint-driven
process, which is itself reversible. The entropy production computed from the driven rates vanishes, but
we continue to count the original jump marks $\pm A$. Their fluctuations
are the observable studied here. Let $\kappa_n$ denote their entropy
cumulants. With entropy counting field $\zeta$, their long-time generating rate is
$\psi_N(1/2-\zeta)-\psi_N(1/2)$, so for even $n$,
$\lim_{T\to\infty}\kappa_n/T=\psi_N^{(n)}(1/2)$.
The counting identity below also holds at finite $T$ under this stationary
preparation. For the directly counted hidden activity, let
$K_T=\sum_e K_{e,T}$ count hidden jumps, and define
$\mathscr C_T(\zeta,\omega)=\log\E_{1/2}\left[e^{\zeta\Sigma_T+\omega K_T}\right]$. Conditional on their event times, the
directions at the midpoint are independent fair signs. Hence
\begin{equation}
 \mathscr C_T(\zeta,\omega)
 =\mathscr C_T\bigl(0,\omega+\log\cosh(A\zeta)\bigr).
 \label{eq:all-count-composition}
\end{equation}
In particular, at every finite $T$,
\begin{equation}
 \begin{aligned}
 \kappa_2&=A^2\E_{1/2}\left[K_T\right],\\
 \kappa_4&=A^4\left[3\Var K_T-2\E_{1/2}\left[K_T\right]\right].
 \end{aligned}
 \label{eq:entropy-activity-cumulants}
\end{equation}
The random number of fair entropy increments has two distinct roles.
Its mean fixes the entropy variance; its temporal correlations enter
the fourth cumulant through the hidden-activity susceptibility.

Odd entropy cumulants vanish by reversal symmetry. The fourth is the
first even cumulant to contain the variance of the number of hidden
events, hence an integrated alignment response. Higher even cumulants
probe higher counting responses, as Eqs.~\eqref{eq:count-sixth} and
\eqref{eq:count-eighth} make explicit.

 This counting identity holds for the common-affinity
uniform cycles on any graph and with any autonomous reversible visible generator;
it does not assume Ising criticality.

\subsection{Exact finite-size amplitudes}
For this uniform periodic critical chain, with $m=0,\ldots,N-1$, define
\begin{equation}
 \begin{gathered}
 R_N=\csc\frac{\pi}{2N},\qquad
 \theta_m=\frac{(m+1/2)\pi}{N},\\
 S_N=\frac12\sum_{m=0}^{N-1}\frac{\cos^2\theta_m}{\sin\theta_m}.
 \end{gathered}
 \label{eq:RS}
\end{equation}
The exact finite-size results, illustrated in Fig.~\ref{fig:cumulants}, are
\begin{align}
 \lim_{T\to\infty}\frac{\kappa_2}{T}
 &=2rA^2(N-\eta R_N),\label{eq:k2}\\
 \lim_{T\to\infty}\frac{\kappa_4}{T}
 &=2rA^4(N-\eta R_N)
   +\frac{12r^2\eta^2 A^4}{\gamma}S_N.\label{eq:k4}
\end{align}
Appendix~\ref{app:ising} derives these expressions from the finite-ring
spectrum and the integrated alignment covariance. In particular,
\begin{equation}
 \frac{S_N}{N}=\frac{\log(8N/\pi)+\gamma_{\mathrm E}-1}{\pi}
 +O(N^{-2}),
 \label{eq:Sasym}
\end{equation}
where $\gamma_{\mathrm E}$ is Euler's constant. Equations~\eqref{eq:k2}
and \eqref{eq:k4} fix both the regular background and the divergent
coefficient.

The cost of making the full histories reversible contains a universal
finite-size term. At $J_{1/2}=\gamma$,
\begin{align}
 E_N(1/2)&=N\left[\gamma+2r d_{1/2}\right]-2\gamma R_N\nonumber\\
 &=N f(1/2)-\frac{\pi\gamma}{6N}+O(N^{-3}).
 \label{eq:casimir}
\end{align}
This is the conformal finite-size energy
$-\pi c v/(6N)$ with central charge $c=1/2$ and velocity
$v=2\gamma$~\cite{Affleck1986,BloteCardyNightingale1986}.
In the present setting it is a subextensive correction to a minimum
relative-entropy rate. The full microscopic interpretation is exact
because the hidden modes remain separated by Eq.~\eqref{eq:hiddengap}.

\begin{figure*}[t]
\includegraphics[width=\textwidth]{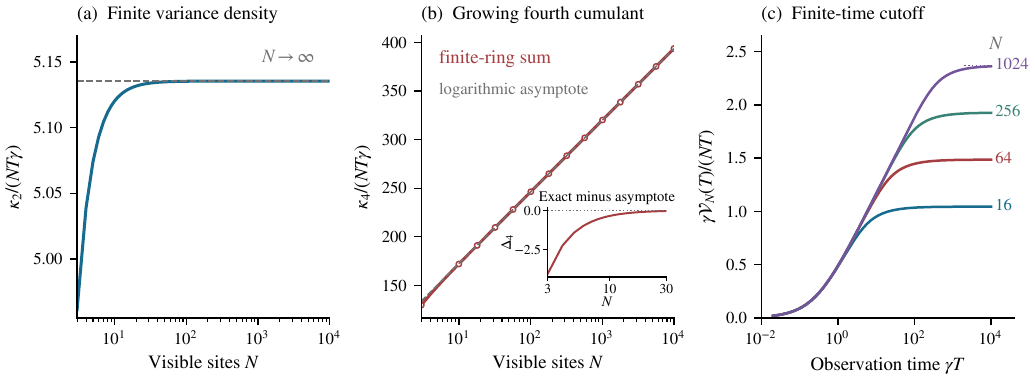}
\caption{Which entropy fluctuation detects criticality, and when.
The force and rate scale are those of Fig.~\ref{fig:phases}, with
$\eta=\eta_c$. (a) The long-time variance density approaches a finite
limit. (b)  The fourth-cumulant density grows logarithmically with $N$; the inset
shows $\Delta_4$, the exact $\kappa_4/(NT\gamma)$ minus its plotted
large-$N$ expression. Solid curves evaluate Eqs.~\eqref{eq:k2} and \eqref{eq:k4}; dashed
curves use their asymptotes. (c) The exact integrated alignment variance
$\gamma\mathcal V_N(T)/(NT)$ grows with observation time and saturates
at $S_N/N$ (dotted levels). Multiplication by $12r^2\eta^2A^4/\gamma$
gives its contribution to $\kappa_4/(NT)$. The common $1/t$ correlation
tail is cut off by observation time or ring size. All entropy cumulants
refer to the original entropy observable in the stationary reversible
midpoint process.}

\label{fig:cumulants}\label{fig:time-crossover}
\end{figure*}

\subsection{Finite observation time}
The logarithm has a temporal cutoff that can be determined exactly.
At criticality, let $\mathcal B=\sum_iZ_iZ_{i+1}$ and write
$C_{\mathcal B}(t)=\langle\mathcal B(t)\mathcal B(0)\rangle_{1/2}
-\langle\mathcal B\rangle_{1/2}^2$. The spectral identity in
Eq.~\eqref{eq:random-stochastic-correlation} has a particularly instructive
consequence. With $\mathcal V_N(T)=\Var\int_0^T\mathcal B(t)\dd t$,
\begin{equation}
 \begin{gathered}
 \mathcal V_N(T)=2\int_0^T(T-t)C_{\mathcal B}(t)\dd t,\\
 \frac{C_{\mathcal B}(t)}{N}\sim\frac{1}{2\pi\gamma t},
 \qquad 1\ll\gamma t\ll N.
 \end{gathered}
 \label{eq:critical-green-kubo-main}
\end{equation}
The conditional hidden intensity is $2r\left[N-\eta\mathcal B(t)\right]$, so
$\Var K_T=\E_{1/2}\left[K_T\right]+4r^2\eta^2\mathcal V_N(T)$. Together with
Eq.~\eqref{eq:entropy-activity-cumulants}, this connects the $1/t$ tail
to the fourth entropy cumulant. Appendix~\ref{app:finite-time-counts} evaluates
the finite spectral sum at every $N\ge3$ and $T$. For
$1\ll\gamma T\ll N$, the growing fourth-cumulant rate is
\begin{equation}
 \frac{\kappa_4}{NT}
 =\frac{12r^2\eta^2A^4}{\pi\gamma}\log(\gamma T)+O(1).
 \label{eq:temporal-log}
\end{equation}
For $\gamma T\gg N$ the logarithm saturates at $\log N$, recovering
Eq.~\eqref{eq:k4}. The variance rate in Eq.~\eqref{eq:k2} is already exact
at finite $T$ under stationary midpoint preparation. These formulas
separate a long-time cumulant rate from a cumulant measured before the
critical relaxation time.

\subsection{Random criticality beyond a divergent cumulant}
\label{sec:random-response}
We now keep the independent binary visible flips but replace the
periodic gate profile by spatial disorder. Each realization is fixed
throughout a trajectory; this is quenched disorder. For nondegenerate
independent identically distributed gates in a fixed interval
$0<\eta_{\min}\leq\eta_i\leq\eta_{\max}<1$, the critical criterion is
$J_0(s)\exp(\E_{\rm dis}\left[\log\eta\right])=\gamma$. Here $\E_{\rm dis}$ averages over the gate distribution after computing
a trajectory observable for each fixed realization. It differs from
$\E_s$, which averages stationary driven histories at fixed gates.
The established random-Ising theory has activated critical relaxation,
$\log\tau_{\rm rel}\sim N^{1/2}$ for typical scales, in place of a
finite dynamical power~\cite{Fisher1995,IgloiRieger1998RandomWalks}.
The uniform hidden-sector bound transfers these low-energy scales to
the full driven process. Near a tuned critical midpoint, the quadratic
entropy source in Eq.~\eqref{eq:tangent} makes the neighboring
Griffiths dynamical exponent diverge as $(s-1/2)^{-2}$. In this regime,
rare strongly coupled intervals cause slow power-law relaxation. Appendix~\ref{app:random-gates} derives the
coefficient and specifies the distinction between quenched, typical
and averaged quantities.

The fourth cumulant  probes the response to the weighted alignment
$\mathcal B_\eta=\sum_i\eta_iZ_iZ_{i+1}$. For one fixed sample, let
$\mathscr E_N(x)$ be the ground energy of
$-\gamma\sum_iX_i-x\mathcal B_\eta$, and define
$\chi_\eta(N)=-\mathscr E_N''(x_*)$ at $x_*=J_0(1/2)$.
The bond scale $x$ has units of inverse time. The conditional-count
identity, Eq.~\eqref{eq:entropy-activity-cumulants},  
relates the variance of hidden event counts to the alignment
susceptibility and gives the long-time relation
\begin{equation}
 \lim_{T\to\infty}\frac{\kappa_4-A^2\kappa_2}{T}
 =12r^2A^4\chi_\eta(N).
 \label{eq:random-cumulant-main}
\end{equation}
Thus a slow mode matters for this measurement through its alignment
matrix element as well as its lifetime.

Free fermions reduce the $2^N$-state response to an $N$-dimensional
matrix. Set $M(x)=\gamma I-x\operatorname{diag}(\eta_i)S_{\rm AP}$,
where $S_{\rm AP}$ shifts a site and changes sign on crossing the
periodic boundary. The unique positive spin ground state is even under global spin
inversion; its Jordan--Wigner fermions therefore obey this antiperiodic
boundary condition. If $M=U\operatorname{diag}(\varsigma_a)V^{\mathsf T}$
is its singular-value decomposition, with orthogonal $U,V$ and
$\varsigma_a>0$, the ground energy is
$\mathscr E_N(x)=-\sum_a\varsigma_a$. The two sets of real Majorana
modes are coupled by $M$, so its singular values are half the fermion
excitation energies. Put $\mathsf F=U^{\mathsf T}M'V$, with $M'=\dd M/\dd x$.
Then
\begin{equation}
 \chi_\eta(N)=\sum_{a<b}
 \frac{(\mathsf F_{ab}-\mathsf F_{ba})^2}{\varsigma_a+\varsigma_b}.
 \label{eq:random-curvature-main}
\end{equation}
The organizing step is to embed $M$ in the symmetric matrix
$\bigl(\begin{smallmatrix}0&M\\M^{\mathsf T}&0\end{smallmatrix}\bigr)$,
whose eigenvalues are $\pm\varsigma_a$. In the curvature of the sum
of positive eigenvalues, transitions within that branch cancel in
pairs. The remaining positive--negative transitions have denominators
$\varsigma_a+\varsigma_b$ and matrix elements proportional to the
antisymmetric part of $\mathsf F$. This proves the structure of
Eq.~\eqref{eq:random-curvature-main} and its regularity at coincident
positive singular values. Appendix~\ref{app:random-energy-response}
gives the differentiation and finite-time response. Physically,
alignment preserves spin-inversion parity and therefore creates
fermion pairs. In
particular, the much smaller odd-parity gap cannot be substituted for
these denominators. The formula also avoids subtracting nearby
extensive energies in a numerical derivative.

The corresponding bounded lattice problem can be tested with
Eq.~\eqref{eq:random-curvature-main}. We choose
gates whose logarithms $\log(4\eta_i)$ are independent and uniform on $\left[-w,w\right]$,
$\gamma=1$ and $x_*=4$. For $w=0.3$ or $0.8$ all rates stay positive,
and the ensemble satisfies the geometric-mean critical condition
exactly. Each sample retains its own finite-size drift.
Figure~\ref{fig:disorder-response} compares the disorder-averaged
curvature density with clean critical growth. The numerical test,
its sampling errors and its precision checks are detailed in
Appendix~\ref{app:disorder-numerics}. For $w=0.8$, the mean curvature
density stays close to $0.0502$ from $N=64$ to $1024$, while the median
smallest pair energy decreases from $0.1329$ to $1.379\times10^{-4}$.
The weak-disorder case reaches an apparent plateau near $0.0865$
after a longer crossover.

\begin{figure*}[t]
\includegraphics[width=\textwidth]{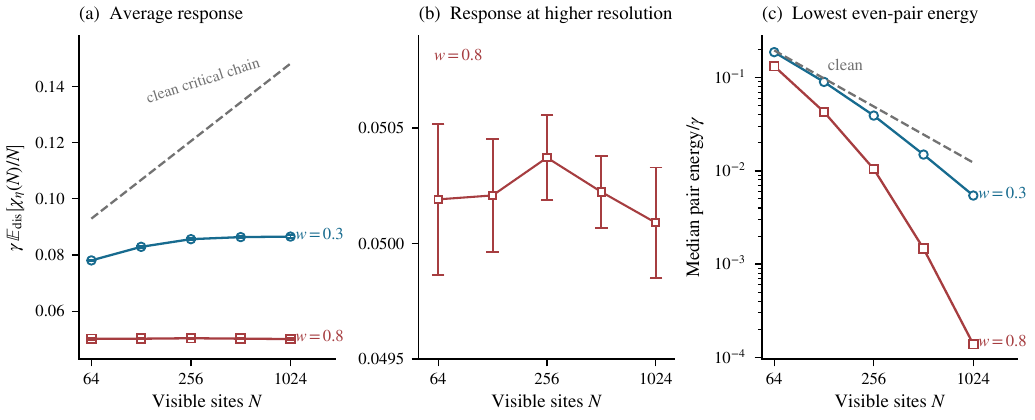}
\caption{Quenched critical response without clean logarithmic growth.
(a) Mean curvature density from Eq.~\eqref{eq:random-curvature-main}
for independent log-uniform gates. Error bars are one sampling standard
error, using 512 samples at each of $N=64,128,256$, 256 at $512$, and
64 at $1024$, for each width. The dashed curve is the exact clean
critical result. (b) The same $w=0.8$ response and error bars on an expanded vertical
scale. The variation across sizes is not resolved beyond sampling error.
(c) Median smallest even-pair energy in the same samples, compared with
the lowest clean even-parity pair excitation $8\gamma\sin\left[\pi/(2N)\right]$.
The falling energy scale does not force the weighted curvature to grow.
Connecting lines guide the eye; no limiting value or exponent is fitted.}

\label{fig:disorder-response}
\end{figure*}

These results motivate the lattice conjecture that
$\lim_{N\to\infty}\E_{\rm dis}\left[\chi_\eta(N)/N\right]$ is finite for these two bounded
distributions. The continuum derivation and the lattice evidence have
different logical roles. A proof for general bounded disorder would
also have to control rare samples before interchanging derivatives,
disorder averages and the thermodynamic limit.

The long-wavelength fermion description of a weakly disordered Ising
chain has a spatially random mass, the coefficient measuring local
distance from its critical point. In this continuum theory the averaged energy
is infinitely differentiable but nonanalytic at the critical
point~\cite{BunderMcKenzie1999}. Up to analytic backgrounds and an
overall dimensional factor, its relevant contribution is the
Borel--Laplace integral~\cite{BunderMcKenzie1999,DLMF2026}
\begin{equation}
 \Phi_{\rm dis}(\upsilon)=-\int_0^\infty e^{-t/(2\upsilon)}
 \left[\tfrac12\coth(t/2)-\tfrac1t\right]\dd t.
 \label{eq:disorder-borel}
\end{equation}
The integration variable $t$ is dimensionless, and $\upsilon>0$
measures the dimensionless distance from criticality.
At a tuned midpoint its leading weak-disorder identification is
$\upsilon\simeq\mu_s/v_\eta\propto(s-1/2)^2$, where
$\mu_s=\log\left[\gamma/J_0(s)\right]-\E_{\rm dis}\left[\log\eta\right]$ and
$v_\eta=\Var_{\rm dis}(\log\eta)$. This relation keeps the leading weak-disorder scaling field.
At finite lattice disorder, its nonlinear corrections and analytic
energy background must also be determined.

Termwise Laplace integration of the kernel expansion gives a
factorially divergent Bernoulli series.
The Borel kernel's nearest poles, $t=\pm2\pi i$, lie off the positive integration ray,
so the integral defines an unambiguous real sum. Every fixed derivative
at the midpoint remains finite, while optimal truncation has an error
$O(\sqrt{\upsilon}\,e^{-\pi/\upsilon})$, smaller than every power of the
entropy bias. Appendix~\ref{app:disorder-borel} derives a signed
remainder bounded by the first omitted term. Thus asymptotic summation
captures critical information beyond any finite cumulant order within
the continuum theory. The lattice plateau remains a separate conjecture;
analytic backgrounds and nonlinear scaling fields matter for its amplitudes.

\section{Beyond binary order}
\label{sec:potts}
The physical comparison and reduction do not depend on free fermions.
A non-Ising example illustrates how the same entropy source can test a
different universality class. Replace each visible bit by a $q$-state
variable $z_i\in\{0,\ldots,q-1\}$ for integer $q\geq2$, jumping to every different value
at rate $\gamma$. On each edge use Eq.~\eqref{eq:rates} with
\begin{equation}
 P_e^{(q)}=\frac{q\,\delta_{z_i,z_j}-1}{q-1}.
 \label{eq:Pq}
\end{equation}
It equals one on an equal-color pair, $-1/(q-1)$ otherwise, and has
zero uniform mean. Positivity and the matched ordinary visible histories, stationary
probabilities, mean entropy and activity, and affinities survive,
with visible activity $(q-1)\gamma$ per site.

On a periodic ring of $N\geq3$ sites, let $\mathsf X_i$ cyclically shift
the color and let $\mathsf Z_i$ have diagonal entries
$e^{2\pi\mathrm i z_i/q}$. These $q$-dimensional clock matrices generalize
the two-dimensional Pauli matrices. The exact reduced operator is
\begin{align}
 H_s^{(q)}={}&N\left[(q-1)\gamma+2r d_s\right]\Id
 -\gamma\sum_i\sum_{k=1}^{q-1}\mathsf X_i^k\nonumber\\
 &-J_s^{(q)}\sum_i\sum_{k=1}^{q-1}\mathsf Z_i^k \mathsf Z_{i+1}^{-k},\nonumber\\
 &J_s^{(q)}=\frac{2r\eta d_s}{q-1}.
 \label{eq:pottsH}
\end{align}
For $q=3$ and $0<s<1$, this is the ferromagnetic quantum Potts operator, with a continuous
self-dual transition at $J^{(3)}=\gamma$
~\cite{KarraschSchuricht2017,ZhangSierra2025}. The matched-family phase
criterion becomes $J_{\max}>2\gamma$; if hidden activity is unrestricted,
an ordered member exists precisely when $\sigma>8\gamma$.

The three-state thermal correlation-length exponent is $\nu=5/6$ and
the critical theory has $c=4/5$~\cite{KarraschSchuricht2017}.
Its singular ground-energy density scales as
$|J^{(3)}-\gamma|^{(1+1)\nu}=|J^{(3)}-\gamma|^{5/3}$.
At the critical entropy midpoint,
$J_s^{(q)}-\gamma=-r\eta A^2(s-1/2)^2/(q-1)+O((s-1/2)^4)$.
For $q=3$ this quadratic dependence gives
\begin{equation}
 f_{\mathrm{sing}}^{(3)}(s)\propto|s-1/2|^{10/3},\qquad
 \lim_{T\to\infty}\frac{\kappa_4}{NT}\propto N^{2/5}.
 \label{eq:pottscritical}
\end{equation}
The variance per site remains finite. The power $2/5$ follows from
thermal susceptibility scaling $N^{2/\nu-2}$, since the fourth entropy
derivative is the first derivative that contains that susceptibility.
These are scaling predictions using the established Potts critical
theory; the exact model reduction itself holds for every finite $N$.
The distinction between logarithmic Ising growth and the Potts power
gives the non-Ising extension an observable consequence.

More generally, in spatial dimension $d$ with dynamical exponent
$z_{\mathrm{dyn}}$, a quadratic entropy source tangent to a single
thermal scaling field gives a singular fourth-cumulant density
proportional to $L^{2/\nu-(d+z_{\mathrm{dyn}})}$, when hyperscaling
applies and the exponent is positive. A marginal susceptibility may
instead have logarithmic growth, as in the Ising chain. This statement
concerns a finite-size critical regime with $T$ longer than its relaxation
time. Its hypotheses must be reconsidered for  systems at or above an upper critical dimension or transitions with
additional relevant scaling fields.

\section{Conserved particles and indistinguishable thermodynamic noise}
\label{sec:conserved}
\begin{figure*}[t]
\includegraphics[width=\textwidth]{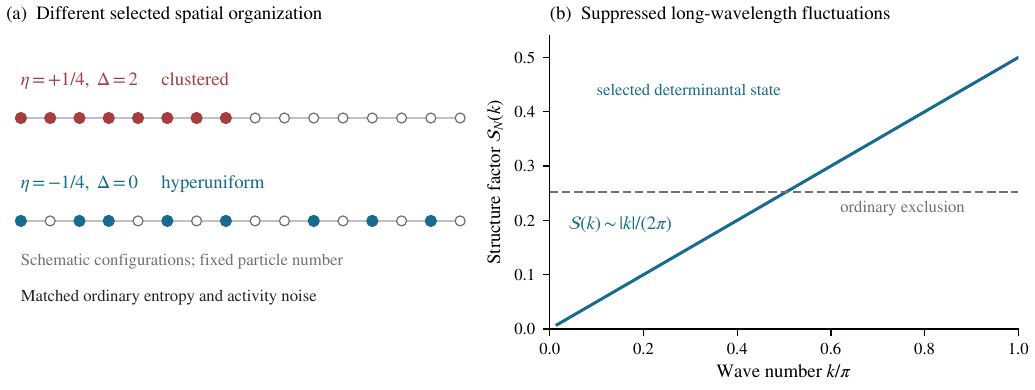}
\caption{Collective organization in the conserved matched pair.
The schematic particle strips illustrate segregation and a fluctuating
state with repulsion between occupied positions.
The graph compares the exact nonzero-wave-number structure factors
for $N=128$, $Q=64$. The selected determinantal state suppresses
long-wavelength fluctuations, whereas ordinary exclusion has a constant
structure factor. The two signs of the hidden gate have identical
ordinary entropy--activity means and covariance matrices at every
observation time.}

\label{fig:conserved}
\end{figure*}

We return to binary occupations and change the visible dynamics from
independent flips to particle-conserving exchanges. Consider a periodic
symmetric exclusion process with $N\geq3$ sites and $Q$ particles,
$0<Q<N$. Here $Q$ is the conserved particle number; the Potts color
number $q$ no longer enters. Each adjacent
particle--hole pair exchanges at rate $\gamma$. Write $n_i\in\{0,1\}$,
$z_i=1-2n_i$ and $P_i=z_i z_{i+1}$. The ordinary process is uniform within
the fixed-$Q$ sector. Its mean bond alignment is
\begin{equation}
 m_{N,Q}=\frac{(N-2Q)^2-N}{N(N-1)}.
 \label{eq:conserved-centering}
\end{equation}
Couple each bond to a hidden cycle with rates
\begin{equation}
 \begin{gathered}
 u_i=r f_i e^{A/2},\qquad v_i=r f_i e^{-A/2},\\
 f_i=1-\eta(P_i-m_{N,Q}).
 \end{gathered}
 \label{eq:conserved-rates}
\end{equation}
Both signs of $\eta$ are allowed; $|\eta|<1/2$ ensures positive rates
uniformly in size. The centering keeps $\langle f_i\rangle_0=1$ exactly,
including the finite-size canonical correlations. Changing $\eta$ leaves
every ordinary visible history, the full stationary distribution, all
cycle affinities, and the means $\sigma=2rA\sinh(A/2)$ and
$a_h=2r\cosh(A/2)$ unchanged.

The two signs have a stronger observational equivalence. Let
$\mathcal Y_T=\int_0^T\sum_i(P_i-m_{N,Q})\dd t$ under stationary ordinary
exclusion. Given that history, the directional hidden counts are Poisson
with means $r e^{\pm A/2}(NT-\eta\mathcal Y_T)$. Their conditional
variance and covariance determine the covariance matrix $\mathcal C_T$ of the vector $(\Sigma_T,K_T)^{\mathsf T}$,
\begin{equation}
 \begin{aligned}
 \mathcal C_T={}&NT\begin{pmatrix}A^2a_h&\sigma\\\sigma&a_h\end{pmatrix}\\
 &+\eta^2\Var_0\mathcal Y_T
 \begin{pmatrix}\sigma^2&\sigma a_h\\\sigma a_h&a_h^2\end{pmatrix},
 \end{aligned}
 \label{eq:conserved-matched-noise}
\end{equation}
where $K_T$ counts hidden jumps. Hence $+\eta$ and $-\eta$ have the same
first and second joint entropy--activity cumulants at every finite
observation time. Correlations with the visible configuration can
distinguish them, since
$\operatorname{Cov}_0(\mathcal Y_T,K_T)
=-a_h\eta\operatorname{Var}_0\mathcal Y_T$.

Exact hidden elimination now gives a number-conserving anisotropic
Heisenberg (XXZ) operator,
where $X_i,Y_i,Z_i$ are again Pauli operators on the binary occupation
space $n_i=0,1$,
\begin{align}
 H_s-NC_s
 &=\frac\gamma2\sum_i\Delta_s(\Id-Z_iZ_{i+1})\nonumber\\
 &\quad-\frac\gamma2\sum_i(X_iX_{i+1}+Y_iY_{i+1}),\nonumber\\
 J_s&=2r\eta d_s,\qquad \Delta_s=1+2J_s/\gamma,\nonumber\\
 C_s&=2rd_s\left[1-\eta(1-m_{N,Q})\right],
 \label{eq:conserved-XXZ}
\end{align}
with $d_s=\cosh(A/2)-\cosh\left[(1/2-s)A\right]$. The anisotropy $\Delta_s$ is the ratio of longitudinal to transverse
couplings. The full dominant eigenvector
is again independent of hidden positions. Its positive visible part $\phi_s$, normalized to unit squared norm
over the fixed-$Q$ configurations, determines the selected stationary distribution
$p_s^{\rm D}(n)=\phi_s(n)^2$ and the legal exchange rate
$\gamma\phi_s(n')/\phi_s(n)$. Particle number
remains conserved in the driven process.

Choose finite parameters satisfying
\begin{equation}
 r\left[\cosh(A/2)-1\right]=\gamma,\qquad \eta_+=\tfrac14,
 \qquad\eta_-=-\tfrac14.
 \label{eq:conserved-pair}
\end{equation}
At the midpoint, the two XXZ anisotropies are respectively $2$ and $0$.
For the positive member, positivity of the unbiased exclusion operator
and a contiguous-block trial state imply
\begin{equation}
 \begin{gathered}
 NC_++4J_+\leq E_+\leq NC_++4J_++2\gamma,
 \\\langle N_{\mathrm{wall}}\rangle_+\leq4.
 \end{gathered}
 \label{eq:conserved-wall-main}
\end{equation}
Here the subscripts $+$ and $-$ denote the two gates evaluated at
$s=1/2$, and $N_{\mathrm{wall}}$ counts particle--hole interfaces. At fixed
density $Q/N\to\rho\in(0,1)$, every finite spatial window converges
to an occupied window with probability $\rho$ or an empty window with
probability $1-\rho$. Finite rings remain translation invariant. The bounded interface count establishes macroscopic segregation.
A particle-rich domain in the vacancy background is called a droplet;
its  position remains delocalized.

For the negative member, the interaction cancels the diagonal exclusion
escape term, leaving free fermions. The occupied orbitals form a Slater
determinant, which becomes the positive sine product
\begin{equation}
 \phi_-(x_1,\ldots,x_Q)
 =\frac{2^{Q(Q-1)/2}}{N^{Q/2}}
 \prod_{a<b}\sin\frac{\pi(x_b-x_a)}N,
 \label{eq:conserved-vandermonde-main}
\end{equation}
for ordered occupied sites $0\leq x_1<\cdots<x_Q<N$.
Thus a local hidden turnover choice selects a circular determinantal
distribution with repulsion between occupied positions. Its static
structure factor, with $\rho=Q/N$ at finite size and defined for
$k_\ell=2\pi\ell/N$, $\ell=1,\ldots,N-1$, by
$\mathcal S_N(k_\ell)=N^{-1}\langle|\sum_j e^{ik_\ell j}(n_j-\rho)|^2\rangle_-$,
is exactly
\begin{equation}
 \mathcal S_N(k_\ell)=\frac{\min\{\ell,N-\ell,Q,N-Q\}}N.
 \label{eq:conserved-structure-main}
\end{equation}
As $N\to\infty$ at fixed density, it approaches $|k|/(2\pi)$ at small
nonzero wave number. This vanishing structure factor defines
hyperuniformity, the suppression of density fluctuations on long length
scales. Ordinary exclusion instead retains $\rho(1-\rho)$ in this limit.

The fermionic structure also determines the dynamics. Antisymmetrizing
independent one-particle evolution gives its exterior power, represented
by a determinant. The positive ground-state transform then turns it
into the exact transition probability between ordered configurations,
\begin{equation}
 p_t(x,y)=e^{E_{\mathrm{hop}}t}
 \frac{\phi_-(y)}{\phi_-(x)}
 \det\left[g_t(x_a,y_b)\right]_{a,b=1}^Q.
 \label{eq:conserved-temporal-kernel}
\end{equation}
Here $g_t$ is the one-particle hopping propagator on the ring with
$e^{ikN}=(-1)^{Q-1}$, and
$E_{\mathrm{hop}}=-2\gamma\sin(\pi Q/N)/\sin(\pi/N)$ is the sum of
occupied hopping energies. Appendix~\ref{app:conserved} gives the
momentum kernel and normalization. This propagator describes the
homogeneous driven process, whose finite-duration entropy-weighted
bridge can still have endpoint dependence~\cite{PopkovSchutz2011}.
Its first particle--hole excitation fixes the visible relaxation gap,
\begin{equation}
 g_N^{\mathrm{vis}}
 =4\gamma\sin(\pi Q/N)\sin(\pi/N).
 \label{eq:conserved-gap}
\end{equation}
At fixed $0<\rho=Q/N<1$, this is
$4\pi\gamma\sin(\pi\rho)/N+o(N^{-1})$. Once below the uniform
hidden-sector bound, it is also the full microscopic gap. The same
fermionic organization therefore controls spatial fluctuations and
physical relaxation, not just the stationary density.

Activity-conditioned exclusion already connects XXZ chains with
segregation and hyperuniformity
\cite{LecomteGarrahanVanWijland2012,JackThompsonSollich2015}.
The sine-product state and its effective repulsion also arise in an
extreme-current or extreme-activity limit
\cite{PopkovSchutzSimon2010,PopkovSchutz2011}.
Here that determinantal state occurs at finite microscopic rates and a
finite entropy bias, paired with a segregated member whose ordinary
thermodynamic noise agrees exactly. The comparison identifies a
limitation of aggregate diagnostics even when conservation controls the
collective dynamics. Appendix~\ref{app:conserved} derives the finite-ring
normalization, interface bound and time-dependent selected dynamics.

\subsection{A mobile bound domain and exponential relaxation}
\label{sec:droplet-transport}

The clustered member has a further dynamical distinction from its
hyperuniform partner. Its boundaries can fluctuate while displacement of
the entire domain becomes exponentially slow. Two limits make this
statement precise, one at fixed particle number on the infinite line
and one at fixed density on a ring.
The interface estimate establishes local coexistence. The mobile
bound cluster below is obtained from a separate spectral construction.

First fix $Q$ and take the infinite-line limit of the selected particle
dynamics. The scalar $NC_s$ cancels from its Doob generator before this
limit. Absolute position has no normalized stationary distribution on
$\mathbb Z$; stationarity here refers to the internal particle separations. Write
$\Delta=\cosh\vartheta>1$ and order their positions as
$x_1<\cdots<x_Q$. The XXZ bound-state dispersion is
\cite{FischbacherStolz2014,NachtergaeleSpitzerStarr2007}
\begin{align}
 E_Q(k)&=2\gamma\sinh\vartheta\,
 \frac{\cosh(Q\vartheta)-\cos k}{\sinh(Q\vartheta)},
 \label{eq:droplet-band}\\
 D_Q&=\frac{\gamma\sinh\vartheta}{\sinh(Q\vartheta)}.
 \label{eq:droplet-mobility}
\end{align}
The momentum $k$ labels translation of the entire bound domain in the
operator representation. The band curvature $D_Q$ is also an exact
transport coefficient of the conditioned classical particle process. This probabilistic meaning can be established directly,
without identifying a moving domain with a random walker. Define
\begin{equation}
 b_j=\frac{\cosh\left[(j-Q/2)\vartheta\right]}{\cosh(Q\vartheta/2)},
 \qquad
 w_j=\frac{D_Q}{\gamma b_{j-1}b_j},
 \label{eq:droplet-weights}
\end{equation}
where $b_0=b_Q=1$. Here the $b_j$ are scalar bound-state factors,
with gap index $j$, distinct from the local hard-core operators. Set
$X_{\mathrm{eff}}=\sum_{j=1}^{Q}w_jx_j$. The weights are positive and sum
to unity. The cancellation is local in each gap. Particle $j$ hops
right at rate $\gamma b_{j-1}/b_j$ and left at rate
$\gamma b_j/b_{j-1}$ when the destination is empty. For the two
particles bordering a nonempty internal gap $j$, the products of jump rate and
$X_{\mathrm{eff}}$ increment have equal magnitude $D_Q/b_j^2$ and
opposite sign. Their drift contributions cancel, as do those of the
two outer moves. Hence
$X_{\mathrm{eff}}(T)-X_{\mathrm{eff}}(0)$ is an exact martingale,
a coordinate with zero conditional drift. Appendix~\ref{app:droplet-transport}
evaluates its quadratic variation. With the
internal gaps initially stationary,
\begin{equation}
 \operatorname{Var}\left[X_{\mathrm{eff}}(T)-X_{\mathrm{eff}}(0)\right]
 =2D_QT .
 \label{eq:droplet-variance}
\end{equation}
The arithmetic center of mass $\overline x=Q^{-1}\sum_jx_j$ has the same long-time diffusivity and
Brownian scaling limit. The internal shape remains fluctuating and has a
product-geometric distribution. Such stable clouds belong to the broader
theory of exclusion with particle-dependent rates
\cite{MalyshevMenshikovPopovWade2023}; the entropy-selected rates here make
the harmonic coordinate and diffusivity explicit. For the matched member
$\Delta=2$, Eq.~\eqref{eq:droplet-mobility} gives
$D_Q\sim2\sqrt3\gamma(2-\sqrt3)^Q$, whereas the mean visible jump rate
tends to $2\gamma/\sqrt3$ as $Q$ grows. Frequent boundary rearrangements
therefore coexist with slow collective transport.

The same bound-state structure fixes non-Gaussian collective-current
fluctuations. In the entropy-selected infinite-line particle process,
let $C_T=\overline x(T)-\overline x(0)$ be the arithmetic center-of-mass
displacement of a fixed-$Q$ cloud, with $Q\ge2$ and stationary internal
gaps. Appendix~\ref{app:droplet-current}
establishes
\begin{equation}
 \lim_{T\to\infty}\frac1T\log\E\left[e^{\theta C_T}\right]
 =2D_Q(\cosh\theta-1),\quad |\theta|<Q\vartheta .
 \label{eq:droplet-current-main}
\end{equation}
The proof controls the endpoint factors in the tilted path measure,
which are essential on an infinite gap space. The real field $\theta$ biases displacement after the entropy selection
at $s=1/2$. The resulting cloud drifts at $2D_Q\sinh\theta$ and remains
bound throughout this interval. The mean span $\ell_Q=x_Q-x_1$ grows by the exact factor
$\left[1-(\cosh\theta-1)/(\cosh(Q\vartheta)-1)\right]^{-1}$.
The internal gaps lose a normalizable stationary distribution at the
strip boundary. The collective current has a Poisson form at
long times even though the center of mass is not a finite-time Markov
projection.

For two particles, the complete SCGF follows the bound level until it
meets a continuous band. With $c_\theta=\cosh(\theta/2)$,
\begin{equation}
 \Lambda_2(\theta)=
 \begin{cases}
 \dfrac{2\gamma}{\Delta}(c_\theta^2-1),&c_\theta\le\Delta,\\[2pt]
 4\gamma c_\theta-2\gamma(\Delta+\Delta^{-1}),&c_\theta\ge\Delta.
 \end{cases}\label{eq:two-particle-full-SCGF}
\end{equation}
At $|\theta|=2\vartheta$, the first derivative is continuous but the
second drops by $\gamma(\Delta-\Delta^{-1})$. Beyond this threshold the
internal gap has no normalizable stationary distribution. Figure~\ref{fig:current-binding}
shows this current-induced unbinding transition. It is an infinite-line
result with stationary initial gaps; a fixed finite ring has an analytic
SCGF. The same bound-level and continuous-band mechanism appears in
current-conditioned zero-range processes and their two-particle
exclusion interpretation ~\cite{RakosHarris2008}. Appendix~\ref{app:two-particle-current}
establishes spectral dominance and the complete velocity rate function
for the present entropy-selected process. For general $Q$, the exterior
SCGF remains undetermined by the bound-state construction.

\begin{figure*}[t]
\includegraphics[width=\textwidth]{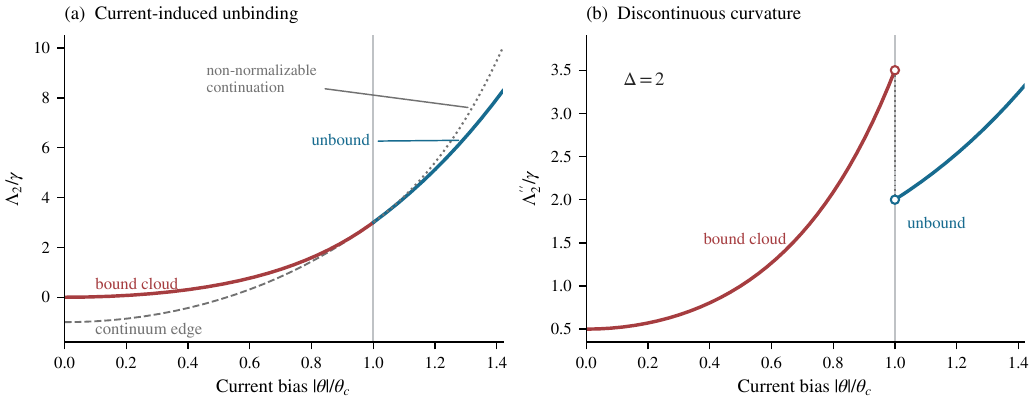}
\caption{Exact two-particle current transition on the infinite line at
$\Delta=2$. (a) The bound SCGF meets the continuous-band edge at
$\theta_c=2\operatorname{arcosh}\Delta$. The dashed curve is the
continuous-band edge below threshold. Beyond threshold, the dotted
continuation of the bound level is not a normalizable eigenstate.
(b) The first derivative is continuous, while the second derivative
has unequal one-sided limits. The plot concerns a further current bias
of the entropy-selected process with stationary initial gaps.}
\label{fig:current-binding}
\end{figure*}

Next keep a nonzero particle density on a ring. Theorem~6.1 of Nachtergaele and Starr for the periodic XXZ droplet
band bounds the first $N$ energies within
$O(e^{-Q\vartheta}+e^{-(N-Q)\vartheta})$ of
$2\gamma\sinh\vartheta$ \cite{NachtergaeleStarr2001}. Consequently the
visible conditioned relaxation gap satisfies
\begin{equation}
 g_{N,Q}\le C_{\Delta,\gamma}
 \left(e^{-Q\vartheta}+e^{-(N-Q)\vartheta}\right).
 \label{eq:droplet-ring-gap-bound}
\end{equation}
At fixed density this proves that the inverse-gap relaxation time grows
at least exponentially with $N$. The finite ring retains a unique translation-invariant
stationary state. The bound concerns its slow approach to stationarity;
it does not posit a chosen droplet position or determine an exact
finite-density gap prefactor. High-precision finite-ring calculations
test the more detailed mobility estimate in
Appendix~\ref{app:droplet-ring-numerics}. They support a specific
asymptotic conjecture while resolving sizable corrections on small rings.

\section{Symmetry and the scope of the construction}
\label{sec:extensions}
We return to the independent-flip binary ring of Sec.~\ref{sec:critical},
with uniform gates tuned to $J_{1/2}=\gamma$ and no conserved particle
number. Its spin representation has a noninvertible
Kramers--Wannier operator $\mathcal D$ that
interchanges $X_i$ and $Z_{i-1}Z_i$, with
\begin{equation}
 \mathcal D^2=(\Id+\mathcal F)\mathcal T^{-1},\qquad
 \mathcal F=\prod_i X_i,
 \label{eq:KW}
\end{equation}
where $\mathcal T$ translates one lattice site. The projection onto
the even spin-flip sector is why this is not an invertible change of
variables~\cite{SeibergSeifnashriShao2024,ZhangSierra2025}.

There is an operational consequence for the selected dynamics. At the
critical midpoint, measure time in $\tau=\gamma t$. Let $F_\tau$ count
visible flips and let $B_\tau=\int_0^\tau\sum_i z_i z_{i+1}\dd\tau'$
be the integrated alignment. Define their joint long-time generating
function by
\begin{equation}
 \Theta_N(u,v)=\lim_{\tau\to\infty}\tau^{-1}
 \log\E_{1/2}\left[e^{uF_\tau+vB_\tau}\right].
 \label{eq:Theta}
\end{equation}
Counting visible flips multiplies the flip coefficient by $e^u$;
biasing integrated alignment changes the bond coefficient to $1+v$.
Kramers--Wannier duality exchanges these two coefficients. For the
positive finite-ring ground state, this gives
\begin{equation}
 \Theta_N(u,v)=\Theta_N(\log(1+v),e^u-1),\qquad v>-1.
 \label{eq:stochasticKW}
\end{equation}
Combining this identity with the finite-ring Ising spectrum gives
\begin{equation}
 \lim_{\tau\to\infty}\frac1\tau
 \operatorname{Cov}\begin{pmatrix}F_\tau\\B_\tau\end{pmatrix}
 =\begin{pmatrix}R_N+S_N&-S_N\\-S_N&S_N\end{pmatrix}.
 \label{eq:KWcov}
\end{equation}
The critical logarithms cancel in the variance of $F_\tau+B_\tau$.
Appendix~\ref{app:duality} derives both the spectral identity and the
associated compensated counting martingale. This is a testable
fluctuation relation inherited from a noninvertible symmetry. The duality acts as a signed linear map between tilted operators;
its probabilistic consequence is the relation between their fluctuation
generators.

The second structure is freedom in the visible algebra on which the
hidden cycles act. Replace the product-flip generator by any finite
irreducible reversible visible generator $L_0$, with stationary measure
$\pi_0$. Let the two directions of each hidden cycle have positive
rates $u_e(z),v_e(z)$ that depend on the visible configuration but
remain independent of hidden position. These generalize Eq.~\eqref{eq:rates}. The stationary full
measure is $\pi(z,h)=\pi_0(z)3^{-|\mathcal E|}$. Visible jump entropy cancels the
$\pi_0$ endpoint contribution, and the conditional-history derivation
again gives exactly $H_s=-L_0+V_s$, now self-adjoint in $L^2(\pi_0)$.

The hidden entropy marks are $\pm\log\left[u_e(z)/v_e(z)\right]$, so
\begin{equation}
 \begin{aligned}
 V_s(z)=\sum_e\left[\begin{aligned}&u_e(z)+v_e(z)\\
 &-u_e(z)^{1-s}v_e(z)^s-v_e(z)^{1-s}u_e(z)^s\end{aligned}\right].
 \end{aligned}
 \label{eq:general-visible-potential}
\end{equation}

For every real $s$, let $\varphi_s$ be the positive lowest eigenfunction
of $H_s$, normalized by $\sum_z\pi_0(z)\varphi_s(z)^2=1$.
Its lift, constant in the hidden positions, is the full Perron vector.
Its visible driven rates are
$k_0(z,z')\varphi_s(z')/\varphi_s(z)$ and its stationary probability is
$\pi_0(z)\varphi_s(z)^2$. The counting-inner-product amplitude is
$\phi_s=\sqrt{\pi_0}\varphi_s$, so the same probability is $\phi_s^2$.
The counting-space Hamiltonian is
$\widetilde H_s=D_0^{1/2}H_sD_0^{-1/2}$, where
$D_0=\operatorname{diag}\pi_0$. Equation~\eqref{eq:random-stochastic-correlation}
then uses $\phi_s$ and $\widetilde H_s$ together.
Appendix~\ref{app:conditional-overlap} derives this similarity.
Thus the positive projection organizes more than one spin model.

Hidden translation symmetry can also be relaxed without losing the
long-time solution. Suppose that every hidden generator has the form
$g_e(z)\mathsf K_e$, where $g_e>0$ and $\mathsf K_e$ is a fixed finite
irreducible generator. If $\lambda_{e,0}(s)$ is the Perron eigenvalue of
its entropy tilt, the full SCGF follows exactly from
\begin{equation}
 \overline L_s=L_0+\sum_e g_e(z)\lambda_{e,0}(s).
 \label{eq:general-hidden-main}
\end{equation}
Products of positive hidden eigenvectors embed visible observables
into the full state space, replacing the constant hidden eigenvectors
used by the group average. This embedding also makes the driven visible process autonomous.
Appendix~\ref{app:general-hidden} proves a spectral bound using
triangularization, which remains valid for defective hidden matrices,
and identifies the low-energy range without Jordan blocks.
The hidden endpoint vectors now enter finite-time generating functions,
as Eq.~\eqref{eq:hidden-general-finite-time} specifies; uniform cycles make this contraction a single exponential through
their constant hidden eigenvector.
An explicit three-state clock shows how hidden residence times affect
the midpoint interaction. At fixed nonuniform stationary probabilities,
varying the visible gates retains the full stationary distribution and
the matched means. Varying the residence probabilities themselves
retains the mean dissipation, activity and affinity but changes the
hidden stationary distribution. Relative to the uniform clock, it
reduces the effective interaction, as Appendix~\ref{app:hidden-noncirculant}
derives.

Other bounded local interactions can be realized in the same way,
while preserving the visible transition graph. Appendix~\ref{app:extensions}
constructs plaquette and dipole-conserving examples and specifies their
symmetry and sector constraints. We now ask when the exact autonomous
reduction itself fails.

\section{Memory selected by hidden turnover}
\label{sec:selected-memory}
\begin{figure*}[t]
\includegraphics[width=\textwidth]{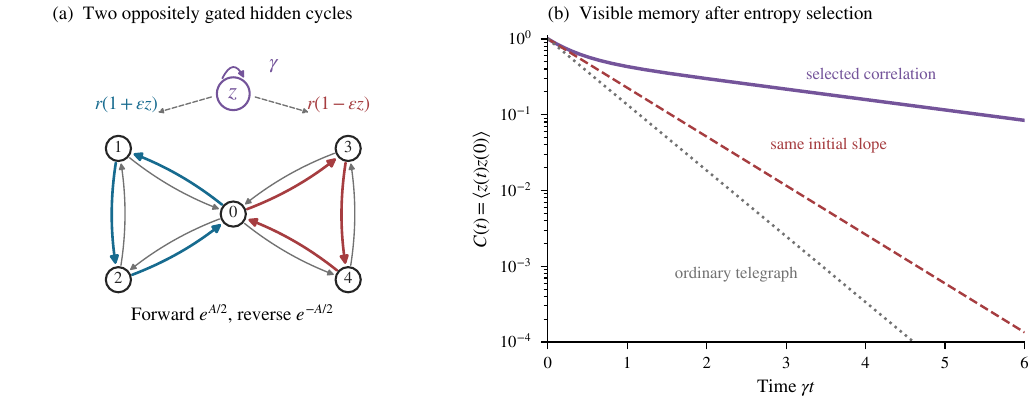}
\caption{Entropy selection creates visible memory. (a) Two hidden cycles
share state zero. Labels are the kinetic prefactors. Colored arrows
indicate forward jumps with factor $e^{A/2}$; gray arrows indicate
reverse jumps with factor $e^{-A/2}$. The visible bit flips independently in the ordinary process.
(b) Exact selected correlation for $r=\gamma$, $A=2\operatorname{arcosh}2$
and $\epsilon=1/2$, plotted against dimensionless time $\gamma t$. Its persistence exceeds both the ordinary telegraph
correlation and an exponential with the same selected initial slope.
The curves follow the three-dimensional spectral expression, with no fit.}
\label{fig:selected-memory}
\end{figure*}

The positive lift in Eq.~\eqref{eq:general-hidden-main} suggests the
precise boundary of autonomous selected dynamics. Here we allow the hidden tilted matrix itself, rather than a scalar
prefactor alone, to depend on $z$. Assume a finite tilted operator
$L_s=L_0+\mathsf K_s(z)$, where $L_0$ is irreducible
and its visible jumps leave the hidden state unchanged. Each hidden
matrix $\mathsf K_s(z)$ is irreducible, with nonnegative off-diagonal
entries. The condition is a common positive eigenline,
\begin{equation}
 \mathsf K_s(z)b_s=\lambda_s(z)b_s,\qquad b_s>0
 \quad\hbox{independent of }z.
 \label{eq:common-line-main}
\end{equation}
Indeed the driven visible rate is
$k_0(z,z')\Phi_s(z',h)/\Phi_s(z,h)$. Independence of $h$, together
with connectivity of visible transitions, forces the full positive
Perron vector to factorize as $\Phi_s(z,h)=a_s(z)b_s(h)$.
Substitution in the eigenvalue equation gives
Eq.~\eqref{eq:common-line-main} and the reduced operator
$L_0+\lambda_s(z)$. Conversely, their positive eigenvectors lift to
the full Perron vector and autonomous rates.

At the reversible entropy midpoint this criterion also determines
whether stationary visible histories are Markovian, meaning that their
future depends on the present visible state without further history.
To see the role of reversibility, let $\mathcal P$ be conditional
expectation onto visible functions in the driven stationary measure,
$\mathcal Q=I-\mathcal P$, and $L^{\rm D}$ its self-adjoint generator.
Then
\begin{equation}
 \mathcal P(L^{\rm D})^2\mathcal P-(\mathcal P L^{\rm D}\mathcal P)^2
 = (\mathcal QL^{\rm D}\mathcal P)^*(\mathcal QL^{\rm D}\mathcal P).
 \label{eq:common-line-compression-main}
\end{equation}
Differentiating a Markov compressed semigroup twice at $t=0$ makes
the left side zero. The positive square then forces invariance of visible functions, hence autonomy.
This stationary weighted projection differs from the uniform hidden
group average used earlier. Appendix~\ref{app:common-perron-line}
gives the finite-state assumptions and converse, and distinguishes
the general-$s$ statement~\cite{BurkeRosenblatt1958}.

Hidden-state projections can carry memory
\cite{RoldanParrondo2012}, and conditioning on empirical occupations can
generate self-interacting dynamics~\cite{CoghiGarrahan2025}.
The comparison below matches ordinary visible histories, stationary
probabilities, mean dissipation, mean activity and driving forces.
Its selected visible process nevertheless develops memory. These are
the mean constraints of Sec.~\ref{sec:model}. The conserved pair also
matches the first two joint entropy--activity cumulants at every
observation time, as Eq.~\eqref{eq:conserved-matched-noise} shows.

Consider a visible bit $z=\pm1$ that flips at rate $\gamma>0$,
independently of a hidden variable $h\in\{0,1,2,3,4\}$.
The hidden graph consists of the two oriented cycles
$0\to1\to2\to0$ and $0\to3\to4\to0$, which share their central state.
The forward and reverse rates on the left cycle are
$r(1+\epsilon z)e^{\pm A/2}$, and those on the right are
$r(1-\epsilon z)e^{\pm A/2}$, where $r,A>0$ and $|\epsilon|<1$.
Each undirected hidden edge represents a single pair of reverse
transitions. The parameter $\epsilon$ redistributes turnover between
the cycles without changing their affinity $3A$.
The ordinary stationary measure is uniform on all ten states. The total
hidden activity and entropy-production rates of this five-state unit are
\begin{equation}
 a_{\rm bt}=\frac{12r}{5}\cosh(A/2),\qquad
 \sigma_{\rm bt}=\frac{12rA}{5}\sinh(A/2),
 \label{eq:memory-matched-means}
\end{equation}
independently of $\epsilon$, even conditional on the current visible
state. The full mean activity is $\gamma+a_{\rm bt}$. Every microscopic
rate ratio remains fixed, and the visible process is the same telegraph
process throughout the family.

At the entropy midpoint, let $\Phi_\epsilon(z,h)>0$ be the normalized
ground vector of the ten-state symmetric tilted Hamiltonian. Its driven
visible flip rate is
\begin{equation}
 w_\epsilon(z,h)=\gamma
 \frac{\Phi_\epsilon(-z,h)}{\Phi_\epsilon(z,h)}.
 \label{eq:memory-driven-flip}
\end{equation}
Interchanging the two hidden cycles together with $z\mapsto-z$ leaves
the model invariant, so the selected visible stationary measure remains
uniform. The selected temporal correlations nevertheless change.
The stationary correlation
$C_\epsilon(t)=\langle z(t)z(0)\rangle_{1/2}$ tests that change directly.
Since $L^{\rm D}z=-2w_\epsilon z$, reversibility gives
$C_\epsilon''(0^+)=\langle(L^{\rm D}z)^2\rangle$. Hence
\begin{align}
 -C_\epsilon'(0^+)&=2\langle w_\epsilon\rangle,\\
 C_\epsilon''(0^+)-\left[C_\epsilon'(0^+)\right]^2
 &=4\Var_{\Phi_\epsilon^2}(w_\epsilon)>0
 \quad(\epsilon\ne0).
 \label{eq:memory-witness}
\end{align}
Appendix~\ref{app:selected-memory} proves strict positivity for every
allowed nonzero $\epsilon$ by reducing the problem to two
three-dimensional symmetry sectors. A stationary two-state Markov process with uniform measure has
a single exponential correlation and makes the left-hand side vanish.
Thus Eq.~\eqref{eq:memory-witness} detects memory in the stationary
selected visible process, even though its equal-time measure has not
changed. The complete driven process is still reversible and Markovian.
The exact identity $\langle w_\epsilon^2\rangle=\gamma^2$ and
nonzero rate variance imply
$a_\epsilon=2\langle w_\epsilon\rangle<2\gamma$. Strict convexity
of the positive spectral mixture then gives
$C_\epsilon(t)>e^{-a_\epsilon t}>e^{-2\gamma t}$ for every $t>0$.
The additional persistence cannot be absorbed into a single fitted flip
rate. Both inequalities are strict for $\epsilon\ne0$.

The symmetry reduction also determines the memory at all times.
It expresses the correlation as
$C_\epsilon(t)=\boldsymbol u^{\mathsf T}e^{-t\mathsf G_\epsilon}\boldsymbol u$. Reflection within each triangle reduces the relevant hidden states
to three amplitudes. Here $\boldsymbol u$ is the normalized ground
vector in the resulting block even under the combined visible flip
and cycle interchange.
Multiplication by $z$ takes it into the odd block, whose Hamiltonian
minus the ground energy is the positive-definite relaxation matrix
$\mathsf G_\epsilon$. Let
$\mathsf Q=I-\boldsymbol u\boldsymbol u^{\mathsf T}$,
$\boldsymbol b=\mathsf Q\mathsf G_\epsilon\boldsymbol u$, and let
$\mathsf G_\perp$ restrict $\mathsf Q\mathsf G_\epsilon\mathsf Q$ to the
two orthogonal amplitudes. Eliminating them gives
\begin{equation}
 \dot C_\epsilon(t)=-a_\epsilon C_\epsilon(t)
 +\int_0^t\mathfrak M_\epsilon(t-\tau)C_\epsilon(\tau)\,d\tau,
 \quad C_\epsilon(0)=1,
 \label{eq:memory-main-kernel}
\end{equation}
with the constructive expression
\begin{equation}
 \mathfrak M_\epsilon(t)=
 \boldsymbol b^{\mathsf T}e^{-t\mathsf G_\perp}\boldsymbol b.
 \label{eq:memory-kernel-main}
\end{equation}
Here $\boldsymbol b$ is understood in an orthonormal basis of that
complement. Positivity of $\mathsf G_\perp$ makes the kernel a sum of at
most two decaying exponentials with nonnegative weights. This finite-dimensional elimination uses the projection geometry of
Mori's memory construction~\cite{Mori1965}, applied here to dissipative
Markov evolution. In particular,
$\mathfrak M_\epsilon(0)$ equals the positive quantity in
Eq.~\eqref{eq:memory-witness}. The physical source of memory is the
change of the positive hidden eigenvector when the visible bit flips.
The two frozen hidden Hamiltonians have the same spectrum but different
ground vectors. Their mismatch couples visible relaxation to hidden
relaxation, although no such feedback acts on the ordinary visible
process. No separation of time scales is required.

\section{Discussion}
\label{sec:discussion}

Hidden kinetic correlations can change collective organization while
preserving the measured visible process and specified thermodynamic
diagnostics. Proposition~\ref{prop:matched} establishes this ambiguity
for mean dissipation, activity and affinities. The examples here locate the missing information
in the dependence of hidden turnover on visible configurations. The periodic binary construction in Sec.~\ref{sec:periodic-profiles}
retains opposite phases after matching finite-time entropy variance.
The conserved pair in Sec.~\ref{sec:conserved} strengthens the comparison
to the first two joint cumulants of entropy and hidden activity. These are explicit comparisons between local
finite-rate models, with the thermodynamic limit taken after the
long-time selection at each finite size.

The hierarchy of comparisons also identifies useful measurements. A
mixed event--alignment observation and its Fisher information determine
a uniform gate, including when activity and affinity are unknown.
For the periodic binary ring of Sec.~\ref{sec:periodic-profiles}, aggregate
arithmetic moments need not determine the geometric mean controlling
its Ising phase. The constant-term
deformation makes this obstruction exact while keeping the gates real
and positive. Edge-resolved joint observations recover information that
these aggregate constraints discard. In internal-state models, this
distinguishes measuring conformational traces and reaction counts
separately from observing their correlations.

For conserved particles, Eq.~\eqref{eq:conserved-matched-noise} pairs
macroscopic segregation with a determinantal state whose density
fluctuations are suppressed at long wavelengths. Activity-biased exclusion develops phase separation and hyperuniformity
\cite{LecomteGarrahanVanWijland2012,JackThompsonSollich2015}, while
large-current exclusion selects a determinantal state
\cite{PopkovSchutzSimon2010}.
Their finite-rate entropy realization here retains the exact agreement
of ordinary entropy--activity means and noise. The determinant fixes
transition probabilities and density correlations, whereas the bound
state fixes a different transport mechanism. A domain can have frequent
boundary rearrangements and exponentially small translational diffusion.
Further current selection stretches its stationary shape; for two
particles the SCGF on the whole real bias axis locates unbinding at a
continuous-band threshold. The bound-level/continuous-band mechanism also occurs in the
current-conditioned zero-range problem and its two-particle exclusion
interpretation~\cite{RakosHarris2008}. Here it describes the further
current fluctuations of the entropy-selected clustered member,
Eq.~\eqref{eq:two-particle-full-SCGF}. On a ring at fixed density,
the rigorous exponential gap bound and the numerically tested prefactor
separate the collective time scale from its finer asymptotics.

For independent hidden clocks with scalar gates,
Eq.~\eqref{eq:hidden-general-reduction} identifies the information needed
to close a visible evolution equation. A common positive hidden eigenvector
preserves autonomous selected motion, and a gap separates the remaining
hidden modes when the stated bound holds. Jordan blocks away from the
reversible midpoint are compatible with this separation. If the hidden
positive direction instead depends on the visible state, selection can
make a visible jump sensitive to an unobserved internal configuration.
The two overlapping cycles of Sec.~\ref{sec:selected-memory} isolate
this effect in ten states, outside the common-eigenvector family. The network's nonexponential
visible correlation has a positive memory kernel even though the
ordinary visible process is a telegraph process. This finite example
links the spectral obstruction to a measurable delay, without assuming
a separation of time scales. The full selected process remains Markovian;
memory concerns its observed projection.

Criticality and its thermodynamic diagnostic need not appear at the same
cumulant order. At a tuned midpoint the entropy source couples
quadratically to the thermal direction. Equation~\eqref{eq:entropy-activity-cumulants} relates the fourth entropy
cumulant to hidden-activity fluctuations in the uniform binary ring;
its temporal cutoff follows by integrating the marginal $1/t$
alignment correlation until the observation time or ring size cuts it off.
Disorder further separates a small spectral gap from a large response.
The pair-response formula, Eq.~\eqref{eq:random-curvature-main}, weights
each relaxation mode by its coupling to the energy operator. For the two disorder ensembles in
Eq.~\eqref{eq:disorder-numerical-ensemble}, numerical response densities
approach apparent plateaus while the lowest pair energies decrease. In the weak-disorder
continuum theory, every fixed response derivative can remain finite
although their large-order growth prevents convergence of the power
series. The Borel integral resolves that distinction and bounds optimal
truncation. Relating this continuum response to the lattice amplitudes requires
the scaling-field and background terms identified in
Appendix~\ref{app:disorder-borel}.

The gap-prefactor conjecture in Eq.~\eqref{eq:droplet-gap-conjecture}
and the response plateau in Fig.~\ref{fig:disorder-response} pose two
specific analytical questions.
The finite-density relaxation-gap prefactor requires relative control of an
exponentially narrow band, and a general lattice response limit requires
uniform control of rare disorder realizations. For more than two
particles, the exterior current SCGF and the possible dominance of
fragmented configurations remain to be determined. A many-body extension of the selected-memory example would ask how
visible-dependent hidden eigenvectors modify critical relaxation;
the finite network settles the local mechanism but leaves its collective
scaling open. The plaquette and dipole-conserving realizations of
Appendix~\ref{app:extensions} lead to phase and hydrodynamic problems
within their distinct symmetry and configuration sectors. The explicit operators and observables derived here make these extensions
accessible to analytical and computational study. Joint measurements of
visible configurations and hidden events distinguish the solved models
where their separate diagnostics agree.

\begin{acknowledgments}
A.C. acknowledges the ICTS, Bengaluru program, \emph{Monopole Moduli 2026}, for their hospitality and for creating stimulating environments in which part of the work was completed.
\end{acknowledgments}

\appendix
\section{Path-space and rate variational principles}
\label{app:information}

Consider a finite irreducible jump process with rates $k_{xy}$ on a
bidirected support and stationary distribution $\pi_x>0$. All
configuration variables are even under time reversal, and the rates
have no external time dependence. Let $P_T$ be its stationary
history measure and $P_T^{\rm R}$ its reversal. Reciprocal
support ensures mutual absolute continuity. The total entropy is
\begin{equation}
 \Sigma_T=\log\frac{dP_T}{dP_T^{\rm R}}
 =\log\frac{\pi_{x_0}}{\pi_{x_T}}
   +\sum_{\text{jumps }x\to y}\log\frac{k_{xy}}{k_{yx}}.
 \label{eq:app-total-entropy}
\end{equation}
This path definition distinguishes configuration irreversibility from
a heat interpretation requiring additional reservoir assignments
\cite{RoldanParrondo2012}.

\subsection{The reversible midpoint measure}
\label{app:path-midpoint}

For $0\le s\le1$, define
\begin{equation}
 \begin{gathered}
 Z_T(s)=\int(dP_T)^{1-s}(dP_T^{\rm R})^s,\\
 dQ_s=Z_T(s)^{-1}e^{-s\Sigma_T}\,dP_T.
 \end{gathered}
 \label{eq:app-geometric-path}
\end{equation}
Products of measures denote products of densities with respect to
any common dominating measure. For every history measure $Q$
with finite relative entropies, substitution gives
\begin{align}
 &(1-s)\KL(Q\Vert P_T)
       +s\KL(Q\Vert P_T^{\rm R})
 \nonumber\\
 &\hspace{12mm}
 =\KL(Q\Vert Q_s)-\log Z_T(s).
 \label{eq:app-kl-completion}
\end{align}
Nonnegativity of relative entropy proves the variational minimum and
its unique optimizer $Q_s$.

At $s=1/2$, $Z_T$ is the Hellinger affinity, the overlap of the square-root path densities. The R\'enyi divergence of order $1/2$ is
$D_{1/2}(P_T\Vert P_T^{\rm R})=-2\log Z_T(1/2)$,
and $Q_{1/2}^{\rm R}=Q_{1/2}$. If
$Q^{\rm R}=Q$, reversal invariance gives
$\KL(Q\Vert P_T)
 =\KL(Q\Vert P_T^{\rm R})$. Thus
\begin{equation}
 -\log Z_T(1/2)
 =\min_{Q^{\rm R}=Q}
      \KL(Q\Vert P_T).
 \label{eq:app-reversible-projection}
\end{equation}
The minimization includes all reversal-invariant history measures.
The finite-time optimizer is generally not a stationary homogeneous
Markov process. The stationary driven process describes its long-time
bulk, as in the general variational construction
\cite{ChetriteTouchetteControl2015}.

\subsection{Finite duration and stationary bulk}
\label{app:finite-duration-bulk}
For the binary uniform-cycle model and remaining physical duration $\tau$, set
$\mathfrak h_{s,\tau}=e^{-\tau H_s}\mathbf1>0$.
The exact finite-duration selected visible rates at time $t$ and its
initial distribution are
\begin{align}
 k_{s,T}(t;z,z')&=k_0(z,z')
 \frac{\mathfrak h_{s,T-t}(z')}{\mathfrak h_{s,T-t}(z)},\nonumber\\
 q_{s,T}(0,z)&=\frac{2^{-N}\mathfrak h_{s,T}(z)}{Z_{N,T}(s)}.
 \label{eq:finite-duration-rates}
\end{align}
The future path weight is the function $\mathfrak h$, so conditioning a
jump on the remaining duration gives the first ratio. Symmetry of $H_s$
then gives the time-$t$ marginal
\begin{equation}
 q_{s,T}(t,z)=\frac{\mathfrak h_{s,t}(z)
             \mathfrak h_{s,T-t}(z)}{2^N Z_{N,T}(s)}.
 \label{eq:finite-duration-marginal}
\end{equation}
Its mass is one by the semigroup property. At fixed $N$, taking $t\to\infty$ and $T-t\to\infty$ selects the
positive lowest mode in both factors, giving $q_{s,T}(t,z)\to\phi_s(z)^2$. The initial endpoint instead
approaches a probability proportional to $\phi_s(z)$. This distinction
separates stationary driven correlation formulas at finite observation
time from the endpoint-dependent finite-duration entropy ensemble.

\subsection{A reversible-rate variational problem}
\label{app:rate-variational}

Write $\lambda_x=\sum_{y\ne x}k_{xy}$. The negative tilted backward
generator for bulk jump entropy has entries
\begin{equation}
 (\widehat H_s)_{xx}=\lambda_x,\quad
 (\widehat H_s)_{xy}=-k_{xy}^{1-s}k_{yx}^{s}\quad(k_{xy}>0).
 \label{eq:app-tilted-matrix}
\end{equation}
The remaining off-diagonal entries vanish. The endpoint term in
Eq.~\eqref{eq:app-total-entropy} gives the exact formula
\begin{equation}
 Z_T(s)=(\pi^{1-s})^{\mathsf T}e^{-T\widehat H_s}\pi^s.
 \label{eq:app-endpoint-vectors}
\end{equation}
Powers of $\pi$ are componentwise. In particular $\widehat H_{1/2}$ is symmetric,
and its lowest eigenvalue is
$E_0=-\psi_N(1/2)$, with
$\psi_N(s)=\lim_{T\to\infty}T^{-1}\log Z_T(s)$ at fixed system size $N$.

Let $\widetilde k$ be trial rates reversible with respect to a
normalized $\rho>0$. Their stationary relative-entropy rate against
the original process is
\begin{equation}
 \mathcal R(\rho,\widetilde k)
 =\sum_{x,y\ne x}\rho_x
 \left[
 \widetilde k_{xy}\log\frac{\widetilde k_{xy}}{k_{xy}}
 -\widetilde k_{xy}+k_{xy}
 \right].
 \label{eq:app-relative-entropy-rate}
\end{equation}
All rate sums run over the original bidirected support. Trial processes
use that same support; adding a transition where the original rate vanishes
would have infinite relative-entropy cost. The initial relative entropy is
$O(1)$ and does not enter this rate.
Introduce the symmetric trial flux
$a_{xy}=\rho_x\widetilde k_{xy}
       =\rho_y\widetilde k_{yx}$.
For each unordered edge, its contribution is
\begin{equation}
 2a_{xy}\log\frac{a_{xy}}
 {\sqrt{\rho_x\rho_y k_{xy}k_{yx}}}
 -2a_{xy}+\rho_xk_{xy}+\rho_yk_{yx}.
 \label{eq:app-flux-cost}
\end{equation}
Strict convexity fixes
$a_{xy}=\sqrt{\rho_x\rho_y k_{xy}k_{yx}}$.
After this edgewise minimization,
\begin{equation}
 \min_{\widetilde k\,{\rm reversible\ at}\,\rho}
       \mathcal R(\rho,\widetilde k)
 =(\sqrt\rho)^{\mathsf T}\widehat H_{1/2}\sqrt\rho.
 \label{eq:app-flux-rayleigh}
\end{equation}
The Rayleigh principle and positivity of the lowest eigenvector
$\phi$ therefore prove
\begin{equation}
 \begin{gathered}
 \min_{\rho,\widetilde k\,{\rm reversible}}\mathcal R(\rho,\widetilde k)
 =-\psi_N(1/2),\\
 \rho_x^*=\frac{\phi_x^2}{\sum_y\phi_y^2},\qquad
 \widetilde k_{xy}^*=\sqrt{k_{xy}k_{yx}}\frac{\phi_y}{\phi_x}.
 \end{gathered}
 \label{eq:app-optimal-reversible-rates}
\end{equation}
These are precisely the midpoint driven rates. This derivation
establishes the reversible-rate minimum directly, without assuming
that the finite-time optimizer is already stationary.

\subsection{Conditional hidden overlap and projection}
\label{app:conditional-overlap}

Disintegrating $P_T$ and $P_T^{\rm R}$ into their visible marginals
and conditional hidden measures separates the integrand
$(dP_T)^{1-s}(dP_T^{\rm R})^s$. Its integral over $h$ is precisely
Eq.~\eqref{eq:app-visible-conditional-tilt}, with the conditional overlap
in Eq.~\eqref{eq:app-conditional-affinity}. H\"older's inequality gives
$0<\mathcal A_s(y)\leq1$ for $0\leq s\leq1$ on the common support.
This also proves the stated commutation criterion.
At the midpoint,
$-\log \mathcal A_{1/2}(y)=\tfrac12
D_{1/2}(P_T(h\mid y)\Vert
        P_T^{\rm R}(h\mid y))$.
The conditional overlap therefore acts as a visible path interaction,
even when the ordinary visible process is reversible.

For comparison, the entropy chain rule reads
\begin{align}
 \KL(P_T\Vert P_T^{\rm R})
 ={}&\KL(P_T^{\rm vis}\Vert
                  (P_T^{\rm vis})^{\rm R})
 \label{eq:app-kl-chain}\\
 &+\E_{P_T^{\rm vis}}\left[\KL(P_T(h\mid y)\Vert P_T^{\rm R}(h\mid y))\right]. \notag
\end{align}
Mean hidden irreversibility and its history-dependent overlap are
different inputs. Fixing the former need not fix the latter.

For the ring construction with a reversible autonomous visible
generator $L_0$, the full stationary distribution is
$\pi_0(z)3^{-|\mathcal E|}$. Detailed balance makes the visible jump
sum telescope to $\log\left[\pi_0(z_T)/\pi_0(z_0)\right]$, cancelling its endpoint
term exactly. Total entropy is then the hidden jump sum, whose
direct tilted representative is $\overline L_s=L_0-V_s$.
Restricted to hidden-constant observables, the bulk tilted generator
corresponding to Eq.~\eqref{eq:app-tilted-matrix} is
$D_0^s\overline L_sD_0^{-s}$, where
$D_0=\operatorname{diag}\pi_0$ on visible states. The symmetric visible
operator is
$-D_0^{1/2}\overline L_sD_0^{-1/2}$.
Normalize the backward eigenfunction by
$\sum_z\pi_0(z)\varphi_s(z)^2=1$. Then the driven probability is
$\pi_0\varphi_s^2=\phi_s^2$, with
$\phi_s=D_0^{1/2}\varphi_s$ normalized in the counting inner product.
The generator similarity acts on observables, while these probabilities
remain normalized in $L^1$. These factors are essential when the visible reference
process is generalized beyond independent fair spins.

\section{Finite-chain spectrum and entropy fluctuations}
\label{app:ising}

Consider a periodic chain of $N\geq3$ visible spins. Throughout this appendix,
\begin{equation}
 \mathcal X=\sum_{i=1}^N X_i,\quad
 \mathcal B=\sum_{i=1}^N Z_iZ_{i+1},\quad Z_{N+1}=Z_1.
 \label{eq:appendix-spin-sums}
\end{equation}
The ground energy of $\mathcal K(h,J)=-h\mathcal X-J\mathcal B$ is denoted by
$\mathcal E_N(h,J)$, for $h,J>0$. Scalar terms in the entropy operator
are kept separately. This convention distinguishes the Ising ground energy
from the full entropy SCGF.

\subsection{Parity sectors and the critical spectrum}

The Jordan--Wigner fermions are
\begin{equation}
 c_i=\left(\prod_{j<i}X_j\right)\frac{Z_i-\mathrm iY_i}{2},
 \qquad X_i=1-2c_i^\dagger c_i,
\end{equation}
where $Y_i$ is the third Pauli operator.

Here the local quantization axis is $X$, so $c_i^\dagger c_i=(1-X_i)/2$
is the fermion occupation in the Ising solution, not the visible
occupation $(1-Z_i)/2$.

 The string of $X_j$ changes the single-occupancy, or hard-core,  algebra into fermionic anticommutation relations. The conserved spin inversion
$\mathcal F=\prod_iX_i$ is their number parity. In the even sector,
the fermions are antiperiodic, with momenta
$k_m=(2m+1)\pi/N$, $m=0,\ldots,N-1$.
Diagonalizing each paired momentum block gives the excitation energy
\begin{equation}
 \varepsilon(k)=2\sqrt{h^2+J^2-2hJ\cos k}.
 \label{eq:appendix-dispersion}
\end{equation}
The ground state belongs to this even sector. Indeed, the spin-basis
matrix has strictly negative spin-flip entries and a connected transition
graph, so its ground vector is unique and strictly positive. Spin inversion
therefore acts on it with eigenvalue $+1$. Its energy is
\begin{equation}
 \mathcal E_N(h,J)
 =-\sum_{m=0}^{N-1}\sqrt{h^2+J^2-2hJ\cos k_m}.
 \label{eq:appendix-even-energy}
\end{equation}
The expression is symmetric under $h\leftrightarrow J$, including at
finite $N$ in the stated periodic, positive-coupling problem.

At $h=J=\gamma$, put
\begin{equation}
 \begin{gathered}
 \theta_m=\frac{(2m+1)\pi}{2N},\qquad R_N=\csc\frac{\pi}{2N},\\
 S_N=\frac12\sum_{m=0}^{N-1}\frac{\cos^2\theta_m}{\sin\theta_m}.
 \end{gathered}
 \label{eq:appendix-critical-sums}
\end{equation}
The finite trigonometric sum $\sum_m\sin\theta_m=R_N$ gives
\begin{equation}
 \begin{gathered}
 \mathcal E_N(\gamma,\gamma)=-2\gamma R_N,\\
 \partial_J\mathcal E_N=-R_N,\qquad
 \partial_J^2\mathcal E_N=-S_N/\gamma.
 \end{gathered}
 \label{eq:appendix-critical-derivatives}
\end{equation}
The derivatives keep $h$ fixed before setting $h=J=\gamma$.
They follow directly by differentiating each positive square root in
Eq.~\eqref{eq:appendix-even-energy}. In particular,
$\langle\mathcal X\rangle=\langle\mathcal B\rangle=R_N$.

Odd parity instead gives periodic momenta. At criticality its zero-momentum
mode has zero excitation energy, permitting the required odd occupation.
The lowest odd energy is $-2\gamma\cot\left[\pi/(2N)\right]$. Subtracting the even
ground energy determines the first spectral gap,
\begin{equation}
 g_N=2\gamma\left[\csc\frac{\pi}{2N}
                 -\cot\frac{\pi}{2N}\right]
     =2\gamma\tan\frac{\pi}{4N}.
 \label{eq:appendix-critical-gap}
\end{equation}
The allowed excitations within either parity sector lie above this level.
This is the spin-sector gap; the hidden-character bound in the main text
identifies it with the full microscopic driven gap whenever
$g_N<3r(1-\eta)$ at the midpoint.

The large-$N$ behavior relevant to the entropy fluctuations is
\begin{equation}
 R_N=\frac{2N}{\pi}+\frac{\pi}{12N}+O(N^{-3}),\quad
 \frac{S_N}{N}=\frac1\pi\log N+O(1).
 \label{eq:appendix-critical-asymptotics}
\end{equation}
For the second relation, use
$\cos^2\theta/\sin\theta=\csc\theta-\sin\theta$ and separate the two
endpoint singularities of the cosecant sum. Each endpoint contributes
half of the logarithmic coefficient. The $N^{-1}$ term in $R_N$ is
responsible for the finite-size correction $-\pi\gamma/(6N)$ to the
critical ground energy.

\subsection{Entropy derivatives and the critical susceptibility}
The full entropy SCGF is
\begin{equation}
 \psi_N(s)=-N\gamma-2rNd_s-\mathcal E_N(\gamma,2r\eta d_s).
 \label{eq:appendix-full-scgf}
\end{equation}
In the midpoint-driven stationary process, the long-time generating
function of the original entropy is
$\psi_N(1/2-\zeta)-\psi_N(1/2)$. Endpoint factors do not affect this rate.
At the midpoint $d_s'=d_s'''=0$, $d_s''=-A^2$ and $d_s''''=-A^4$.
The chain rule and Eq.~\eqref{eq:appendix-critical-derivatives} then
give Eqs.~\eqref{eq:k2} and \eqref{eq:k4}, including the scalar background.

The same coefficient follows from the occupation-time response.
For $W_T$ defined in Eq.~\eqref{eq:W}, the Feynman--Kac response to
a field coupled to $\mathcal B$ gives
\begin{equation}
 \begin{gathered}
 \lim_{T\to\infty}\frac{\E_{1/2}\left[W_T\right]}{T}=r(N-\eta R_N),\\
 \lim_{T\to\infty}\frac{\Var W_T}{T}
 =\frac{r^2\eta^2S_N}{\gamma}.
 \end{gathered}
 \label{eq:appendix-dwell-response}
\end{equation}
Together with the conditional Poisson identity in
Eq.~\eqref{eq:cumulantbridge}, this independently identifies the
fourth-cumulant enhancement as an integrated alignment susceptibility.

To extract the constant in Eq.~\eqref{eq:Sasym}, subtract both endpoint
poles from $\csc(\pi x)$,
\begin{equation}
 g(x)=\csc(\pi x)-\frac1{\pi x}-\frac1{\pi(1-x)}.
\end{equation}
The resulting function is smooth at the endpoints, and
$\int_0^1g(x)\dd x=(2/\pi)\log(2/\pi)$.
The midpoint sum of the pole terms is
$(2N/\pi)\left[\psi_{\rm dg}(N+1/2)-\psi_{\rm dg}(1/2)\right]$, where $\psi_{\rm dg}$ is the logarithmic
derivative of the gamma function. Using
$\psi_{\rm dg}(1/2)=-\gamma_{\mathrm E}-2\log2$ and the midpoint integration
error gives
$N^{-1}\sum_m\csc\theta_m=(2/\pi)\left[\log(8N/\pi)+\gamma_{\mathrm E}\right]+O(N^{-2})$.
Subtracting $R_N/(2N)$ proves the stated expression for $S_N/N$.

The reduction, phase boundaries and finite-ring expressions were
derived analytically. Their symbolic derivatives were checked with the
open-source computer algebra system SymPy~\cite{SymPySoftware}.
Independently assembled microscopic generators checked the finite-time
contraction and all hidden Fourier sectors for binary and three-state
fixtures; spin matrices checked the critical energies, gaps and
susceptibilities for $3\leq N\leq7$. These finite numerical comparisons
corroborate the analytical identities. The Potts powers follow from the thermal exponent $\nu=5/6$ of the
three-state critical chain~\cite{KarraschSchuricht2017} and the quadratic
source in Eq.~\eqref{eq:pottscritical}; their amplitudes require a separate Potts calculation.

\section{Duality between activity and alignment}
\label{app:duality}

Use dimensionless time $\tau=\gamma t$ for the critical visible driven
process. Let $\phi>0$ be the critical ground vector and
$\Phi=\operatorname{diag}\phi$. Its backward generator is
\begin{equation}
 G=\Phi^{-1}(\mathcal X+\mathcal B)\Phi-2R_NI.
 \label{eq:appendix-critical-doob}
\end{equation}
Let $F_\tau$ count visible flips during dimensionless duration
$\tau$, and let
$B_\tau=\int_0^\tau\mathcal B(\tau')\dd\tau'$.
Biasing their joint distribution by $e^{uF_\tau+vB_\tau}$
changes the tilted backward matrix to
\begin{equation}
 G_{u,v}=\Phi^{-1}\left[e^u\mathcal X+(1+v)\mathcal B\right]\Phi-2R_NI.
 \label{eq:appendix-joint-tilt}
\end{equation}
Off-diagonal flip rates acquire $e^u$, while the occupation integral
adds the diagonal field $v\mathcal B$. Its Perron eigenvalue is the
joint SCGF $\Theta_N(u,v)$. For real $u$ and $v>-1$, the two couplings
are positive; the symmetry of Eq.~\eqref{eq:appendix-even-energy} therefore
proves Eq.~\eqref{eq:stochasticKW}.
This is a finite-$N$, long-time identity. It does not assume that duality
is an invertible map on every boundary sector or that the finite-time
generating functions coincide. Write $\lambda(h,J)=-\mathcal E_N(h,J)$. At $h=J=1$,
$\lambda_h=\lambda_J=R_N$ and
$\lambda_{hh}=\lambda_{JJ}=S_N$. Homogeneity of degree one implies
$\lambda_{hJ}=-S_N$. Substituting $h=e^u$ and $J=1+v$ then proves
Eq.~\eqref{eq:KWcov}; the extra $R_N$ in the flip variance comes from
the exponential count field.
The fluctuating sum has variance rate $R_N$, despite the logarithmically
enhanced entries separately. Indeed, the visible escape rate is
$2R_N-\mathcal B(z)$, so
$F_\tau+B_\tau-2R_N\tau$ is the compensated
counting martingale. Under stationary preparation its finite-time
variance is $R_N\tau$.

\section{Additional marks and local interactions}
\label{app:extensions}
The same conditional Poisson contraction controls more detailed
observations and more general visible generators. These consequences
state the reach of the construction without assuming integrability of
the resulting interacting model.

\subsection{Joint entropy and activity}
If $K_T=\sum_eK_{e,T}$ counts hidden jumps, a joint weight
$e^{-s\Sigma_T-\xi K_T}$ replaces the local potential by
\begin{multline}
 V_{s,\xi}(z)=2r\sum_e(1-\eta P_e)\\
 \times\bigl\{\cosh(A/2)-e^{-\xi}\cosh\left[(1/2-s)A\right]\bigr\}.
 \label{eq:joint-hidden-marks}
\end{multline}
The same Fourier projection remains exact. A mark coupled instead to
$\int P_e\dd K_e$ puts $e^{-\xi P_e}$ in the corresponding term.
Differentiating at zero gives Eq.~\eqref{eq:crossreadout}. Thus the
missing kinetic correlation is accessible within the same generating
operator, with no new elimination prescription.

Write $d_{s,\xi}=\cosh(A/2)-e^{-\xi}\cosh\left[A(s-1/2)\right]$. On an even
periodic binary chain the thermodynamic two-field phase diagram follows from
$J_{s,\xi}=2r\eta d_{s,\xi}$. For $\eta>0$, define
$b_\pm=\cosh(A/2)\pm\gamma/(2r\eta)$. The antiferromagnetic phase lies
below
\begin{equation}
 \xi_{\rm AF}(s)=\log\cosh\left[A(s-1/2)\right]-\log b_+.
 \label{eq:joint-AF}
\end{equation}
If $b_->0$, the ferromagnetic phase lies above
$\xi_{\rm F}(s)=\log\cosh\left[A(s-1/2)\right]-\log b_-$.
Between the two boundaries the selected histories are disordered.
If $b_-\le0$, the ferromagnetic region is absent at all finite biases.
These statements follow from $|J|=\gamma$; alternating spin inversion
maps $J<0$ to $J>0$ on an even ring. Odd rings have a frustrated boundary
bond, so the finite-ring identity must not be transferred to them.

At the critical entropy midpoint, $b_-=1$ and the local ferromagnetic
boundary is
\begin{equation}
 \begin{aligned}
 \xi_{\rm F}(s)&=\log\cosh\left[A(s-1/2)\right]\\
 &=\tfrac12A^2(s-1/2)^2+O((s-1/2)^4).
 \end{aligned}
 \label{eq:curved-critical-boundary}
\end{equation}
Activity bias crosses the thermal direction linearly; entropy bias is
tangent to this curve. Thus the critical activity variance and fourth
entropy cumulant are two derivatives of the same singularity, with
different source coordinates. At a reversible midpoint weighted also
by $e^{-\xi K_T}$, the path variational problem becomes
\begin{equation}
 -\log Z_T(1/2,\xi)
 =\min_{Q^{\rm R}=Q}\{\KL(Q\Vert P_T)+\xi\E_Q\left[K_T\right]\}.
 \label{eq:activity-reversible-cost}
\end{equation}
It minimizes a penalized cost rather than the unpenalized reversible
cost of Eq.~\eqref{eq:app-reversible-projection}.

\subsection{Local potentials at a prescribed entropy bias}
Let $L_0$ be an autonomous reversible visible generator, and let $W_C(z)$
be a bounded real function on a finite local support $C$. Choose a
constant $b_C$ with $W_C+b_C>0$, and a common positive affinity $A$.
Attach one hidden cycle to each support and set
\begin{equation}
 \begin{gathered}
 u_C(z)=r_C(z)e^{A/2},\qquad v_C(z)=r_C(z)e^{-A/2},\\
 r_C(z)=\frac{W_C(z)+b_C}{2\left[\cosh(A/2)-1\right]}.
 \end{gathered}
 \label{eq:local-embedding}
\end{equation}
At the midpoint the exact projected operator is
$-L_0+\sum_C(W_C+b_C)$. The shifts affect the absolute entropy cost
but not the selected visible distribution. This realizes arbitrary
bounded local diagonal interactions on the given visible transition
graph. A single microscopic model fixes their entire $s$ dependence;
it cannot assign unrelated potentials independently at every bias.

For independent visible flips, a cycle coupled through
$P_\square=\prod_{i\in\square}z_i$ generates a square-plaquette
interaction. Flipping a complete row or column leaves every plaquette
product unchanged. The full microscopic generator and its entropy tilts
therefore have subsystem symmetries of the kind studied in
Ref.~\cite{YouDevakulBurnellSondhi2018}. Their finite positive ground
state is invariant. A product of $Z_i$ has zero expectation unless its
support intersects every row and every column an even number of times.
A rectangle's four corners can satisfy this condition while a separated
pair cannot. These exact selection rules do not determine a plaquette
phase diagram or imply subsystem topological order.

Alternatively, symmetric moves $0110\leftrightarrow1001$ on an open
chain conserve particle number and its first spatial moment. Choosing
$L_0$ from these moves restricts the construction to each connected
configuration sector, relevant to multipole-conserving
systems~\cite{HanLakeRo2024}. The positive transform preserves every
allowed transition and every disconnected sector. Equal number and
first moment need not place two configurations in the same sector,
which is the additional obstruction associated with kinetic
fragmentation. A hydrodynamic theory or collective transition requires
an analysis of this visible generator.

\subsection{An entropy contribution invisible to complex zeros}
Complex zeros of a finite-time generating function can characterize
trajectory transitions~\cite{FlindtGarrahan2013}. Their diagnostic
scope can be checked explicitly here. Add an independent hidden
three-state cycle with constant positive rates $u,v$. Its factor in
the entropy generating function is
\begin{equation}
 Z_{\mathrm{extra},T}(s)
 =\exp\{T\left[u^{1-s}v^s+v^{1-s}u^s-u-v\right]\}.
 \label{eq:zero-free}
\end{equation}
This function is entire and has no zeros in complex $s$, using the
real logarithms of the positive rates to define the powers. It leaves
all finite-time zeros of the original generating function unchanged,
as well as the normalized selected visible histories, while adding
$(u-v)\log(u/v)$ to the mean entropy rate. Thus complex zeros identify
the singular structure without determining every regular entropy
contribution. The distinction also explains why scalar terms must
be kept when computing absolute costs and cumulants.

\section{Finite-time counting and critical response}
\label{app:finite-time-counts}
The counting composition in the main text can be developed through
conditional intensities, then evaluated using the critical spectral sum.
This connects higher counted cumulants to the observation-time cutoff.

\subsection{Conditional intensity and higher counted cumulants}
In the driven midpoint process, define the integrated hidden intensity
\begin{equation}
 W_T=r\int_0^T\sum_e\left[1-\eta P_e(t)\right]\dd t.
 \label{eq:W}
\end{equation}
Conditional on a visible history, $W_T$ is the aggregate integrated
intensity of the hidden jumps in each direction.
For the \emph{original} entropy observable of Eq.~\eqref{eq:entropy},
\begin{equation}
 \log\E_{1/2}\left[e^{\zeta\Sigma_T}\mid z(\cdot)\right]
 =2\left[\cosh(A\zeta)-1\right]W_T.
 \label{eq:midpoisson}
\end{equation}
With the original driving force attached to the midpoint jumps,
the second and fourth entropy cumulants obey
\begin{equation}
 \begin{aligned}
 \kappa_2&=2A^2\E_{1/2}\left[W_T\right],\\
 \kappa_4&=2A^4\E_{1/2}\left[W_T\right]+12A^4\Var_{1/2}W_T.
 \end{aligned}
 \label{eq:cumulantbridge}
\end{equation}

The intensity representation connects the counted events of
Eq.~\eqref{eq:all-count-composition} to the time spent in each visible
configuration. Writing $c_j$ for the cumulants of $K_T$, its next two
entropy orders are
\begin{align}
 \kappa_6/A^6&=15c_3-30c_2+16c_1,\label{eq:count-sixth}\\
 \kappa_8/A^8&=105c_4-420c_3+588c_2-272c_1.\label{eq:count-eighth}
\end{align}
All odd entropy cumulants vanish. The statement concerns the original
entropy observable in the driven ensemble. Changing the hidden waiting
mechanism can invalidate the fair-sign reduction even when a spectral
reduction survives, as explained in Appendix~\ref{app:general-hidden}.

\subsection{Spectral response and the temporal cutoff}
In the critical binary chain, $\mathcal B$ preserves fermion parity and
creates pairs of opposite momenta. A pair has excitation gap
$\Omega_m=8\gamma\sin\theta_m$ and squared matrix element
$4\cos^2\theta_m$. Counting each opposite-momentum pair once gives
\begin{equation}
 C_{\mathcal B}(t)
 =2\sum_{m=0}^{N-1}\cos^2\theta_m\,
 e^{-8\gamma |t|\sin\theta_m}.
 \label{eq:critical-time-correlation}
\end{equation}
For odd $N$ the unpaired momentum $\pi$ has zero matrix element. The connected equal-time
variance is $C_{\mathcal B}(0)=N$.

Let $\mathcal V_N(T)$ be the variance of
$\int_0^T\mathcal B(t)\dd t$ in the stationary midpoint process. The
spectral sum gives
\begin{align}
 \mathcal V_N(T)
 &=4\sum_m\cos^2\theta_m\left[
 \frac{T}{\Omega_m}-\frac{1-e^{-\Omega_mT}}{\Omega_m^2}\right] \label{eq:finite-time-integrated-variance}
\\
 &=\frac{T S_N}{\gamma}
 -\frac1{16\gamma^2}\sum_m\cot^2\theta_m
       (1-e^{-8\gamma T\sin\theta_m}). \notag
 \end{align}
It follows that
\begin{equation}
 \kappa_4(T)=2rA^4T(N-\eta R_N)
       +12r^2\eta^2A^4\mathcal V_N(T).
 \label{eq:finite-time-fourth}
\end{equation}
For short times $\mathcal V_N(T)=NT^2+O(T^3)$; for long times its rate
is $S_N/\gamma$. The two endpoint regions of the momentum sum give
the $1/t$ tail in Eq.~\eqref{eq:critical-green-kubo-main}; its time
integral proves Eq.~\eqref{eq:temporal-log}. The odd-parity gap, asymptotically eight times smaller than the first
pair gap, controls full relaxation but does not appear in this even observable's
spectral sum.

\section{Information in state-resolved hidden events}
\label{app:gate-inference}
The gate deformation also has a direct path interpretation. Write $P_{\eta,T}$ for the full ordinary stationary history distribution. Relative to the member $\eta=0$,
\begin{multline}
 \log\frac{\dd P_{\eta,T}}{\dd P_{0,T}}
 =\sum_e\int_0^T\log(1-\eta P_e)\dd K_e\\
 +a_h\eta\int_0^T\sum_e P_e\dd t.
 \label{eq:gate-path-action}
\end{multline}
Both terms are invariant under time reversal. The entropy observable is the same function on path space for every gate, although its distribution changes. This is a precise instance of the time-symmetric kinetic contribution, often called frenesy~\cite{Maes2016Nondissipative}. Fixing mean turnover does not fix this complete path weight.

The binary uniform-cycle model has stationary preparation independent
of $(\eta,a_h,A)$. Its path score for $\eta$ is the martingale
\begin{equation}
 S_\eta=-\sum_e\int_0^T\frac{P_e}{1-\eta P_e}
 \{\dd K_e-a_h(1-\eta P_e)\dd t\}.
 \label{eq:gate-score}
\end{equation}
Compensated counts on distinct edges have zero quadratic covariation.
The mean squared score is therefore
\begin{equation}
 \mathcal I_{\eta\eta}
 =a_h|\mathcal E|T\E_\pi\left[\frac{P_e^2}{1-\eta P_e}\right]
 =\frac{a_h|\mathcal E|T}{1-\eta^2}.
 \label{eq:gate-score-information}
\end{equation}
Path-space information theory relates this mean squared score to
relative entropy ~\cite{Pantazis2013}. Averaging also over the two jump
directions gives Eq.~\eqref{eq:gate-fisher}. For interior parameter
values, the Cram\'er--Rao bound for an unbiased gate estimate is
$(1-\eta^2)/(a_h|\mathcal E|T)$, unchanged by the two nuisance parameters.
The bound is not a guarantee of finite-sample attainment.

Let $K_{\rm al}$ and $K_{\rm op}$ count jumps on aligned and opposite visible pairs,
respectively, and let $B_T=\int_0^T\sum_eP_e\dd t$.
The part of the log likelihood depending on the gate is
\begin{equation}
 \ell(\eta)=K_{\rm al}\log(1-\eta)+K_{\rm op}\log(1+\eta)+a_h\eta B_T.
 \label{eq:gate-sufficient}
\end{equation}
Thus unsigned hidden event times and their visible configurations
retain all gate information; resolving clockwise directions adds
information about $A$ but not $\eta$. The likelihood is concave, and
an interior maximizing gate solves
\begin{equation}
 a_h B_T\eta^2+(K_{\rm al}+K_{\rm op})\eta
 +(K_{\rm al}-K_{\rm op})-a_h B_T=0.
 \label{eq:gate-likelihood-equation}
\end{equation}
Small samples can instead put the maximum at a parameter boundary.
For distinct edge gates, the gate-information block is diagonal with
entries $a_hT/(1-\eta_e^2)$ when edge identities are observed. Aggregate
moments discard that resolution, which is why the matched profiles in
Sec.~\ref{sec:periodic-profiles} evade their inference.

\section{Exact conserved construction}
\label{app:conserved}
The fixed particle-number sector is needed both for the microscopic
centering and for the fermionic boundary condition. This appendix makes
those conventions explicit and derives the spatial and temporal
consequences of Sec.~\ref{sec:conserved}.

\subsection{Canonical centering and the exact operator}
Since $\sum_i z_i=N-2Q$ is fixed, expanding its square and using
exchangeability gives Eq.~\eqref{eq:conserved-centering}. The two values
of the local gate are $1-\eta(1-m_{N,Q})$ and
$1+\eta(1+m_{N,Q})$. Their positivity is equivalent to
\begin{equation}
 -\frac1{1+m_{N,Q}}<\eta<\frac1{1-m_{N,Q}}.
 \label{eq:conserved-positivity}
\end{equation}
The full stationary state is uniform on the fixed-$Q$ configurations
and hidden cycle positions. Visible exchanges have zero jump entropy,
and the stationary endpoint term vanishes. Conditional hidden Poisson
counting therefore gives
\begin{equation}
 \E\left[e^{-s\Sigma_T}\mid n(\cdot)\right]
 =\exp\left[-2rd_s\int_0^T\sum_i f_i(n(t))\dd t\right].
 \label{eq:conserved-conditional}
\end{equation}
Differentiating its two-mark version, or applying conditional variance
to the directional Poisson counts, proves
Eq.~\eqref{eq:conserved-matched-noise}. This equality refers to hidden
activity and entropy; it does not equate higher cumulants or joint
observations involving arbitrary visible path functionals.

The negative unbiased exclusion generator is
\begin{equation}
 H_0=\frac\gamma2\sum_i
 (\Id-X_iX_{i+1}-Y_iY_{i+1}-Z_iZ_{i+1}).
 \label{eq:conserved-H0}
\end{equation}
Each bond acts on $01,10$ with diagonal entries $\gamma$ and
off-diagonal entries $-\gamma$. Adding $2rd_s\sum_i f_i$ proves
Eq.~\eqref{eq:conserved-XXZ}, including its scalar.
In a hidden Fourier sector with $p$ nonzero characters, the Hermitian
part exceeds the constant-character operator by at least
$3rp f_{\min}\cosh\left[(1/2-s)A\right]$, where $f_{\min}=\min_{n,i}f_i(n)>0$.
Thus the lifted positive visible eigenvector is the full Perron vector,
and hidden sectors stay a finite spectral distance away uniformly in
size for the chosen matched pair.

\subsection{Interface bound and local coexistence}
Set $\mathcal W=\sum_i(1-Z_iZ_{i+1})=2N_{\mathrm{wall}}$. Every
nonempty, nonfull configuration has $\mathcal W\geq4$ and $H_0\geq0$.
A basis state with a contiguous particle block has $\mathcal W=4$ and
diagonal expectation $2\gamma$ for $H_0$. Hence, for any $J>0$,
\begin{equation}
 4J\leq E-NC\leq4J+2\gamma,
 \qquad\langle N_{\mathrm{wall}}\rangle\leq2+\gamma/J.
 \label{eq:conserved-wall-proof}
\end{equation}
Uniqueness and positivity make the finite-ring ground state translation
invariant. For a fixed separation $\ell$,
\begin{equation}
 \Pr(n_i\ne n_{i+\ell})
 \leq\frac\ell N(2+\gamma/J).
 \label{eq:conserved-local-coexistence}
\end{equation}
The same union bound applies to all sites of a fixed finite window.
Together with $\langle n_i\rangle=Q/N$, it proves the local mixture
stated in the main text. It does not determine an exact single-droplet
distribution. The thermodynamic energy density obeys
$E/N-C\to0$ for every fixed $J>0$.

The coexistence conclusion extends to a fixed local spacetime window.
The eigenvector equation gives the total driven visible escape rate
\begin{equation}
 R_{\rm D}(n)=(\gamma+2J)N_{\mathrm{wall}}(n)-(E-NC).
 \label{eq:conserved-escape}
\end{equation}
Equation~\eqref{eq:conserved-wall-proof} bounds its stationary mean by
$4\gamma+\gamma^2/J$, independent of size. Translation invariance then
bounds the probability of any exchange across a fixed set of $b$ bonds
during time $T$ by $bT(4\gamma+\gamma^2/J)/N$.
Together with Eq.~\eqref{eq:conserved-local-coexistence}, this proves that
each fixed finite spacetime window tends to a fully occupied or empty
constant history with probabilities $\rho$ and $1-\rho$. The statement
keeps the window and $T$ fixed as $N$ grows.
Appendix~\ref{app:droplet-transport} treats collective motion in the
infinite-line fixed-$Q$ limit and bounds the slow band at fixed density.

\subsection{Fermionic parity, determinant and correlations}
For $J=-\gamma/2$, introduce occupation fermions
$d_i=(\prod_{j<i}Z_j)(X_i+\mathrm iY_i)/2$, so $n_i=d_i^\dagger d_i$.
Their quantization axis differs from that of the Ising fermions $c_i$
in Appendix~\ref{app:ising}. The shifted operator is
$-\gamma\sum_i(d_i^\dagger d_{i+1}+d_{i+1}^\dagger d_i)$.
The Jordan--Wigner transformation in the fixed-$Q$ sector requires
$e^{ikN}=(-1)^{Q-1}$. Its $Q$ occupied ground-state momenta are
\begin{equation}
 k_a=\frac{2\pi}{N}\left(a-\frac{Q+1}{2}\right),
 \qquad a=1,\ldots,Q,
 \label{eq:conserved-momenta}
\end{equation}
with one-particle energy $\epsilon(k)=-2\gamma\cos k$. The determinant
of orbitals $N^{-1/2}e^{ik_a x_b}$ reduces to
Eq.~\eqref{eq:conserved-vandermonde-main} after removal of a constant
phase. Cauchy--Binet proves its normalization over ordered occupied
sites. These sites are not distinguishable particle labels; crossing
the periodic bond entails reordering and the parity factor above.
Summing the occupied energies gives
\begin{equation}
 E_-=NC_-+E_{\mathrm{hop}},\qquad
 E_{\mathrm{hop}}=-2\gamma\frac{\sin(\pi Q/N)}{\sin(\pi/N)}.
 \label{eq:conserved-free-energy}
\end{equation}
For a legal hop from $x_a$ to its neighboring unoccupied site
$x_a+\delta$, the driven rate is
\begin{equation}
 w_- =\gamma\prod_{b\ne a}
 \left|\frac{\sin\left[\pi(x_a+\delta-x_b)/N\right]}
 {\sin\left[\pi(x_a-x_b)/N\right]}\right|,\qquad\delta=\pm1.
 \label{eq:conserved-doob-product}
\end{equation}
The rate depends on all occupied sites, although the microscopic
transition rules were local.

The stationary projection kernel is
$\mathcal K_N(d)=N^{-1}\sum_{a=1}^Q e^{ik_a d}$. Wick contraction gives
\begin{equation}
 \operatorname{Cov}_-(n_i,n_j)
 =\rho\delta_{ij}-|\mathcal K_N(i-j)|^2.
 \label{eq:conserved-covariance}
\end{equation}
The Fourier transform counts occupied momenta carried outside the
occupied interval by $k_\ell$, proving
Eq.~\eqref{eq:conserved-structure-main}. At fixed
$\rho\in(0,1)$ this differs from the ordinary canonical value
$\rho(1-\rho)N/(N-1)$ at nonzero wave number. Vanishing variance at
$k=0$ follows from number conservation in both ensembles and does not
by itself imply hyperuniformity.

\subsection{Finite-time selected evolution}
Let $\mathcal P_Q$ be the $N$ momenta obeying the parity condition and
define the one-particle kernel
\begin{equation}
 g_t(x,y)=\frac1N\sum_{k\in\mathcal P_Q}
 e^{ik(x-y)+2\gamma t\cos k}.
 \label{eq:conserved-oneparticle-kernel}
\end{equation}
Taking the exterior power of this hopping kernel and applying the
positive Doob similarity proves Eq.~\eqref{eq:conserved-temporal-kernel}.
The parity twist compensates the sign of a particle crossing the
ordered configuration boundary. The one-particle hopping kernel need
not itself be a stochastic kernel; the complete determinant and
ground-state factors give the normalized transition probability.

Writing $\mathcal P_{\rm occ}$ for the occupied interval, its stationary dynamic
structure factor for $t\geq0$ and $k=2\pi\ell/N$, $\ell=1,\ldots,N-1$, is
\begin{equation}
 \mathcal S_N(k,t)=\frac1N
 \sum_{\substack{p\in\mathcal P_{\rm occ}\\p+k\notin\mathcal P_{\rm occ}}}
 e^{-t\left[\epsilon(p+k)-\epsilon(p)\right]}.
\label{eq:conserved-dynamic-structure}
\end{equation}
Momenta in this formula are understood modulo $2\pi$.
Moving a particle across a Fermi edge costs
$2\gamma\left[\cos(\pi(Q-1)/N)-\cos(\pi(Q+1)/N)\right]$, proving
Eq.~\eqref{eq:conserved-gap}. At fixed density it eventually falls below
the hidden-sector bound, so it is also the full driven gap for
sufficiently large $N$. A small finite ring can instead have a slower
hidden mode. The fixed-density qualification is essential; fixed
$Q=1$ gives an $N^{-2}$ scale.

\section{Bound-domain transport in the conserved conditioned process}
\label{app:droplet-transport}

The additive constant $NC_s$ in Eq.~\eqref{eq:conserved-XXZ} cancels from
the Doob generator. For $\Delta=\cosh\vartheta>1$, we use
\begin{equation}
 \begin{aligned}
 H_{\mathrm{drop}}&=\frac\gamma2\sum_i\Delta(1-Z_iZ_{i+1})\\
 &\quad-\frac\gamma2\sum_i(X_iX_{i+1}+Y_iY_{i+1}).
 \end{aligned}
 \label{eq:droplet-normalization}
\end{equation}
Each legal particle--hole exchange has off-diagonal entry $-\gamma$,
and the diagonal is $\gamma\Delta$ times the number of interfaces.
In particular, a single particle has energy
$2\gamma\Delta-2\gamma\cos k$.

\subsection{The bound band and stationary internal shape}

For $Q$ particles on $\mathbb Z$, put
$g_j=x_{j+1}-x_j-1$. The coordinate Bethe ansatz
$\Psi=\prod_j a_j^{x_j}$ obeys the contact conditions
$a_j+a_{j+1}^{-1}=2\Delta$. Writing $a_j=b_{j-1}/b_j$ reduces these
conditions to a linear recurrence. The solution with total momentum $k$
is
\begin{equation}
 b_j(k)=\frac{e^{ik}\sinh\left[(Q-j)\vartheta\right]+
                  \sinh(j\vartheta)}{\sinh(Q\vartheta)}.
 \label{eq:droplet-complex-b}
\end{equation}
Here $b_0(k)=e^{ik}$ and $b_Q(k)=1$. Up to a momentum-dependent normalization, the eigenfunction is
$e^{ikx_1}\prod_{j=1}^{Q-1}b_j(k)^{g_j}$. Substitution into the
free-separation equation gives Eq.~\eqref{eq:droplet-band}.
The inequality $|b_j(k)|\le b_j(0)<1$ ensures a normalizable internal
shape in each momentum fiber. If $Q$ is even and $k=\pi$, the central
factor vanishes; its limiting interpretation restricts that gap to zero.
The positive $k=0$ state is the infinite-line ground state in the
generalized sense \cite{FischbacherStolz2014,NachtergaeleSpitzerStarr2007}.

At $k=0$, the factors reduce to those in
Eq.~\eqref{eq:droplet-weights}. The ground-state transform gives the
right and left rates of the $j$th ordered particle as
\begin{equation}
 r_j^+=\gamma\frac{b_{j-1}}{b_j},\qquad
 r_j^-=\gamma\frac{b_j}{b_{j-1}},
 \label{eq:droplet-rank-rates}
\end{equation}
whenever the destination is vacant. Their reciprocal structure means
that detailed balance for the gap process is solved by
\begin{equation}
 \pi_Q(g)=\prod_{j=1}^{Q-1}(1-b_j^2)b_j^{2g_j},
 \qquad g_j\in\mathbb Z_{\ge0}.
 \label{eq:droplet-gap-measure}
\end{equation}
This finite-dimensional gap process is irreducible, has bounded total
rate, and is positive recurrent. Absolute position on the infinite line
has no normalizable stationary distribution. Stationarity below refers
to the internal gaps. The same separation between a stable shape and a
wandering position appears in the general particle-dependent exclusion
framework \cite{MalyshevMenshikovPopovWade2023}.

\subsection{An exact harmonic coordinate}

Set $a_j=b_{j-1}/b_j$ and
$c_Q=(\sum_{j=1}^{Q}(b_{j-1}b_j)^{-1})^{-1}$. The weights
$w_j=c_Q/(b_{j-1}b_j)$ make the drift of
$X_{\mathrm{eff}}=\sum_jw_jx_j$ vanish for each individual gap pattern.
Indeed the contribution from the two unblocked outer edges is
$\gamma(-w_1/a_1+w_Qa_Q)=0$, while each nonempty internal gap contributes
\begin{equation}
 \gamma\mathbf1_{\{g_j>0\}}
       (w_ja_j-w_{j+1}/a_{j+1})=0 .
 \label{eq:droplet-drift-cancellation}
\end{equation}
The process is therefore a martingale for every initial configuration.
Under Eq.~\eqref{eq:droplet-gap-measure},
$\Pr(g_j>0)=b_j^2$. The predictable quadratic variation $\langle X_{\mathrm{eff}}\rangle_T$
accumulates the squared martingale jumps weighted by their conditional
rates. Its mean rate is
\begin{align}
 \frac{\mathrm d}{\mathrm dT}\E\left[\langle X_{\mathrm{eff}}\rangle_T\right]
 &=\gamma\left(\frac{w_1^2}{a_1}+w_Q^2a_Q\right)\nonumber\\
 &\quad+\gamma\sum_{j=1}^{Q-1}b_j^2
 \left(w_j^2a_j+\frac{w_{j+1}^2}{a_{j+1}}\right)
 \nonumber\\
 &=2\gamma c_Q .
 \label{eq:droplet-bracket}
\end{align}
The hyperbolic identity
\begin{equation}
 \sum_{j=1}^{Q}\frac1{b_{j-1}b_j}
 =\frac{\sinh(Q\vartheta)}{\sinh\vartheta}
 \label{eq:droplet-telescope}
\end{equation}
follows by telescoping differences of $\tanh\left[(j-Q/2)\vartheta\right]$.
Hence $\gamma c_Q=D_Q$, proving
Eq.~\eqref{eq:droplet-variance}. Ergodicity of the gaps and bounded
martingale jumps imply the Brownian scaling limit with variance rate
$2D_Q$. The difference between $X_{\mathrm{eff}}$ and the arithmetic center
of mass is a linear combination of the gaps. Their stationary
exponential tails make this difference negligible on diffusive scales
at fixed $Q$, including uniformly on compact rescaled time intervals.
This argument identifies an exact martingale, rather than assuming that
the projected center is itself Markovian.

The mean number of interfaces is $2+2\sum_{j=1}^{Q-1}b_j^2$. It tends to
$2\coth\vartheta$ as $Q\to\infty$. The mean visible escape rate is
$\gamma\Delta\,\E\left[N_{\mathrm{wall}}\right]-E_Q(0)$ and therefore tends
to $2\gamma/\sinh\vartheta$. These rates describe boundary motion on the
infinite line with $Q$ fixed before the long-time limit. They do not
replace the finite-density ring analysis.

\subsection{Collective current and a stable-cloud strip}
\label{app:droplet-current}
Consider the entropy-selected infinite-line particle process with fixed
$Q\ge2$, and start the gaps in
Eq.~\eqref{eq:droplet-gap-measure}. For the arithmetic center-of-mass
displacement $C_T=Q^{-1}\sum_j\left[x_j(T)-x_j(0)\right]$, define
$M_Q(\theta,T)=\E\left[e^{\theta C_T}\right]$. The absolute initial position
is arbitrary. For $|\theta|<Q\vartheta$,
\begin{align}
 M_Q(\theta,T)&\sim\mathcal A_Q(\theta)e^{\Lambda_Q(\theta)T},\nonumber\\
 \Lambda_Q(\theta)&=2D_Q(\cosh\theta-1),\label{eq:droplet-current-SCGF}\\
 \mathcal A_Q(\theta)&=\prod_{j=1}^{Q-1}
 \frac{(1-b_j^2)(1-\beta_j^+\beta_j^-)}
 {(1-b_j\beta_j^+)(1-b_j\beta_j^-)}>0,\label{eq:droplet-current-amplitude}
\end{align}
where the two positive internal waves are
\begin{equation}
 \beta_j^\pm(\theta)=\frac{e^{\pm\theta j/Q}\sinh\left[(Q-j)\vartheta\right]
 +e^{\mp\theta(Q-j)/Q}\sinh(j\vartheta)}{\sinh(Q\vartheta)}.
 \label{eq:droplet-current-beta}
\end{equation}
The sign convention follows by continuing Eq.~\eqref{eq:droplet-complex-b}
to $k=-i\theta$ and factoring out $e^{\theta\overline x}$, where
$\overline x=Q^{-1}\sum_jx_j$.
The positive gap eigenfunction is
$h_\theta(g)=\prod_j(\beta_j^+/b_j)^{g_j}$, with eigenvalue
$E_Q(0)-E_Q(-i\theta)=\Lambda_Q(\theta)$. Its Doob transform moves
particle $j$ right at $\gamma B_{j-1}^+/B_j^+$ and left at the reciprocal
times $\gamma$, where
$B_j^+=\left[e^\theta\sinh((Q-j)\vartheta)+\sinh(j\vartheta)\right]/\sinh(Q\vartheta)$,
$B_0^+=e^\theta$, and $B_Q^+=1$. Write $B_j^-(\theta)=B_j^+(-\theta)$.

The gap traffic equations admit the product distribution
$\pi_{Q,\theta}(g)=\prod_j(1-\rho_j)\rho_j^{g_j}$ with
$\rho_j=\beta_j^+\beta_j^-$. The recurrence for $B_j^\pm$ verifies
interior and boundary balance, and
\begin{equation}
 1-\rho_j=\frac{2\sinh(j\vartheta)\sinh\left[(Q-j)\vartheta\right]
 \left[\cosh(Q\vartheta)-\cosh\theta\right]}{\sinh^2(Q\vartheta)}.
 \label{eq:droplet-current-stability}
\end{equation}
Thus the cloud is positive recurrent exactly in the stated strip.
The continued eigenvalue alone does not establish its dominance, because
the gap space is infinite. Let $P_T^{(\theta)}$ denote the gap semigroup,
$\mu_\theta=\pi_Qh_\theta$, and $f_\theta=h_\theta^{-1}$. Changing path
measure gives the exact identity
\begin{equation}
 e^{-\Lambda_QT}M_Q(\theta,T)
 =\langle\mu_\theta/\pi_{Q,\theta},
                 P_T^{(\theta)}f_\theta\rangle_{\pi_{Q,\theta}}.
 \label{eq:droplet-current-endpoints}
\end{equation}
Both endpoint functions lie in this weighted $L^2$ space throughout the
strip. Strict convexity and the endpoint values give $\beta_j^\pm<1$;
Cauchy--Schwarz gives $\beta_j^+\beta_j^-\ge b_j^2$. Hence their squared
norms are geometric sums with ratios
$b_j^2\beta_j^\pm/\beta_j^\mp\le(\beta_j^\pm)^2<1$.
Ergodicity and Markov contraction imply strong $L^2$ convergence to the
stationary projection. The resulting two geometric sums give
Eq.~\eqref{eq:droplet-current-amplitude}, establishing the SCGF.
The moving cloud's span $\ell_Q=x_Q-x_1=Q-1+\sum_jg_j$ also follows
exactly. Here $\E_\theta$ and $\operatorname{Var}_\theta$ refer to
$\pi_{Q,\theta}$; their subscript $0$ denotes $\pi_Q$, the already
entropy-selected bound cloud. With $\zeta_Q=(\cosh\theta-1)/(\cosh(Q\vartheta)-1)$,
$\rho_j=b_j^2+(1-b_j^2)\zeta_Q$ and independence of the gaps gives
\begin{equation}
 \begin{gathered}
 \E_\theta\left[\ell_Q\right]=\frac{\E_0\left[\ell_Q\right]}{1-\zeta_Q},\\
 \operatorname{Var}_\theta\ell_Q=
 \frac{\operatorname{Var}_0\ell_Q+\zeta_Q\E_0\left[\ell_Q\right]}
 {(1-\zeta_Q)^2}.
 \end{gathered}\label{eq:droplet-current-swelling}
\end{equation}
These relations sum the geometric-gap means and variances. The smallest
site interval containing the cloud has $\ell_Q+1$ sites.

For velocities $|v|<v_c=2\gamma\sinh\vartheta$, the probability that $C_T/T$ is near $v$ has exponential rate
$I_Q(v)$, the interior Legendre transform
\begin{equation}
 I_Q(v)=v\,\operatorname{arsinh}\frac{v}{2D_Q}
       -\sqrt{v^2+4D_Q^2}+2D_Q .\label{eq:droplet-current-rate}
\end{equation}
Indeed the transformed cloud has mean velocity
$v=\Lambda_Q'(\theta)=2D_Q\sinh\theta$. Its ergodic theorem and
geometric endpoint tails give the change-of-measure lower bound at this
velocity, while the SCGF gives the matching Chernoff upper bound.
Equation~\eqref{eq:droplet-current-SCGF} has Poisson-current form at long
times. The escape rate depends on the gaps, so their stationary mixture
gives a nonexponential center-of-mass waiting time. Neither a Poisson
finite-time distribution nor a Markov center follows from the SCGF.
All gap loads reach one at $|\theta|=Q\vartheta$.
For general $Q$, loss of a normalizable cloud there does not determine
the exterior SCGF or prove a singularity; the next subsection settles
$Q=2$. For $Q=1$, there are
no gaps and the Poisson expression is exact for all real $\theta$ and $T$.

\subsection{An exact current-induced unbinding transition for two particles}
\label{app:two-particle-current}
For $Q=2$ the exterior SCGF can also be settled. The underlying
bound-level and continuous-band mechanism occurs in current-conditioned
zero-range processes and their two-particle exclusion interpretation
\cite{RakosHarris2008}. Here the gap realization determines the complete
center-of-mass current statistics after entropy selection, with the
stationary initial gap distribution fixed above.

Put $c_\theta=\cosh(\theta/2)$. The gap has birth rate $2\gamma/\Delta$
and death rate $2\gamma\Delta$ away from zero. Its two physical moves
at each event carry opposite center increments $\pm1/2$, so the current
tilt multiplies both off-diagonal rates by $c_\theta$. Similarity by
$\sqrt{\pi_2(g)}=\sqrt{1-\Delta^{-2}}\,\Delta^{-g}$ gives a half-line
Jacobi operator $A_\theta$ with off-diagonal $2\gamma c_\theta$, bulk
diagonal $-2\gamma(\Delta+\Delta^{-1})$, and boundary diagonal
$-2\gamma/\Delta$. Its bound eigenvector $(c_\theta/\Delta)^g$ is
square summable precisely for $c_\theta<\Delta$. Its eigenvalue
and the upper continuous-band edge are the two branches in
Eq.~\eqref{eq:two-particle-full-SCGF}. To verify its
dominance, write $M_2=\langle\sqrt{\pi_2},e^{TA_\theta}\sqrt{\pi_2}\rangle$.
The spectral supremum bounds its exponential growth from above.
Entrywise positivity bounds it below by
$\pi_2(0)\langle0|e^{TA_\theta}|0\rangle$.
Successive applications of the Jacobi operator to $|0\rangle$ span all
finite-support vectors, so this boundary spectral measure reaches the
spectral supremum and gives the matching lower rate.

The branches have equal first derivatives at $|\theta|=2\vartheta$,
but the second derivative drops by $\gamma(\Delta-\Delta^{-1})$.
This is a current-bias binding transition on the infinite line; a fixed
finite ring has an analytic SCGF. Since Eq.~\eqref{eq:two-particle-full-SCGF}
is differentiable for every real field and steep at infinity, its full
velocity rate is
\begin{equation}
 I_2(v)=\begin{cases}
 \begin{aligned}
 &v\operatorname{arsinh}(\Delta v/\gamma)\\
 &\quad-\sqrt{v^2+(\gamma/\Delta)^2}+\gamma/\Delta,
 \end{aligned}&|v|\le v_c,\\[6pt]
 \begin{aligned}
 &2v\operatorname{arsinh}\left[v/(2\gamma)\right]\\
 &\quad-2\sqrt{v^2+4\gamma^2}+2\gamma(\Delta+\Delta^{-1}),
 \end{aligned}&|v|\ge v_c.
 \end{cases}\label{eq:two-particle-full-rate}
\end{equation}
where $v_c=2\gamma\sqrt{\Delta^2-1}$.
At $c_\theta\ge\Delta$ the positive band-edge wave is
$\psi_g=1+(1-\Delta/c_\theta)g$. The generalized Doob gap rates are
$2\gamma c_\theta\psi_{g+1}/\psi_g$ and
$2\gamma c_\theta\psi_{g-1}/\psi_g$, with no downward jump at zero.
The invariant weight $\psi_g^2$ is not summable. At equality the gap is
a null-recurrent reflecting random walk. Above threshold,
$\sum_{g\ge0}\left[2\gamma c_\theta\psi_g\psi_{g+1}\right]^{-1}<\infty$;
the birth--death resistance criterion therefore makes the gap transient.

\subsection{What survives at fixed density on a ring}

Nachtergaele and Starr's Theorem~6.1 applies to every fixed
$\Delta>1$, with $1\le Q\le N-1$
\cite{NachtergaeleStarr2001}. Their Hamiltonian is
$H_{\mathrm{drop}}/(2\gamma\Delta)$. After restoring this factor, the
first $N$ energies lie within
$O(e^{-Q\vartheta}+e^{-(N-Q)\vartheta})$ of
$2\gamma\sinh\vartheta$, with constants depending on $\Delta$ and
$\gamma$. This gives Eq.~\eqref{eq:droplet-ring-gap-bound}. When
$\min(Q,N-Q)\to\infty$, the asymptotic separation from higher visible
excitations is at least $2\gamma(\Delta-1)$. The Doob similarity preserves
energy differences. The previously derived positive lower bound for
nonuniform hidden modes then separates the lowest droplet relaxation
modes from hidden relaxation modes for sufficiently large rings.

These results leave the exact finite-density gap prefactor unsettled.
For example, inserting $D_Q$ into a nearest-neighbor random-walk formula
on a ring would require control of finite-ring binding corrections
relative to the exponentially small band splitting. An estimate on the
whole band width does not establish that control.

\subsection{High-precision finite-ring gaps}
\label{app:droplet-ring-numerics}
Particle--hole symmetry and the infinite-line dispersion motivate the
asymptotic estimate
\begin{equation}
 g^{\mathrm{pred}}_{N,Q}
 =2(D_Q+D_{N-Q})\left[1-\cos(2\pi/N)\right].
 \label{eq:droplet-gap-conjecture}
\end{equation}
It combines the mobilities of a bound particle domain and its
complementary vacancy domain. The conjecture is that
$g_{N,Q}/g^{\mathrm{pred}}_{N,Q}\to1$ at fixed nonzero density and fixed
$\Delta>1$. This requires more control than the band-width theorem.
Table~\ref{tab:droplet-gaps} tests it at half filling for the matched
value $\Delta=2$. The relative correction falls below $0.3\%$ at
$N=18$, after being of order unity on the smallest rings.

We construct the exchange matrix separately in every translation sector.
Double-precision spectra identify the first excited sector for each
listed size. Within the ground and first excited sectors, arbitrary-precision
sparse residuals refine the eigenvalues, using a fixed double-precision
bordered inverse as a preconditioner. Repeating at 50 and 80 decimal
digits leaves the quoted digits unchanged. Full configuration matrices
at $N=6,8$ independently check the complete sector decomposition.
The 80-digit residuals and the precision change are retained with the
calculation. These are spectral calculations, with no trajectory sampling.

\begin{table}[!tbp]
\caption{Finite-ring gaps at $\Delta=2$, $\gamma=1$, and half filling.
The comparison is with the conjectured asymptotic expression
$g^{\mathrm{pred}}_{N,Q}=2(D_Q+D_{N-Q})\left[1-\cos(2\pi/N)\right]$.
All momentum sectors identify $k=\pm2\pi/N$ as the first excited sector
for these sizes. The numerical values are stable under residual refinement at
50- and 80-digit working precision. The last column measures the finite-ring correction to the proposed
asymptotic gap; its decay supports the stated large-$N$ conjecture.}

\label{tab:droplet-gaps}
\begin{ruledtabular}
\begin{tabular}{rrcc}
$N$ & $Q$ & $g_{N,Q}$ & $g_{N,Q}/g^{\mathrm{pred}}_{N,Q}-1$\\

\hline
6 & 3 & $0.337426672758$ & $1.53070004568$\\
8 & 4 & $0.0448971241128$ & $1.14603718080$\\
10 & 5 & $0.00493743710659$ & $0.350806517968$\\
12 & 6 & $0.000753949403919$ & $0.0973732178870$\\
14 & 7 & $0.000140128060361$ & $0.0297589633432$\\
16 & 8 & $0.0000282855656782$ & $0.00923705432112$\\
18 & 9 & $0.00000596655495680$ & $0.00283735537764$
\end{tabular}
\end{ruledtabular}
\end{table}

\section{Hidden generators without translation symmetry}
\label{app:general-hidden}

The full long-time reduction extends beyond circulant hidden cycles.
The scalar-gate family has a common hidden generator multiplied by a
visible kinetic gate. It is a sufficient family within the positive-eigenvector
criterion of Appendix~\ref{app:common-perron-line}. This extension distinguishes exact spectral closure
from the stronger finite-time Poisson identity.

\subsection{Positive product lift and spectral separation}
\label{app:hidden-positive-lift}

Let $L_0$ be a finite irreducible visible generator reversible with respect
to $\pi_0$. On each support $e$, choose a finite irreducible hidden
generator $\mathsf K_e$, with bidirected rates $k^{(e)}_{ab}$ and stationary
distribution $\rho_e>0$. Its physical rates are
$g_e(z)k^{(e)}_{ab}$, where $g_e(z)\ge g_{e,\min}>0$.
The full stationary distribution is
$\pi_0(z)\prod_e\rho_e(h_e)$.
Visible entropy cancels its endpoint contribution; the remaining
total entropy has hidden increments
\begin{equation}
 \xi^{(e)}_{ab}
 =\log\frac{\rho_{e,a}k^{(e)}_{ab}}
               {\rho_{e,b}k^{(e)}_{ba}}.
 \label{eq:hidden-general-mark}
\end{equation}
The gate cancels from this ratio.

If distinct reaction channels connect the same two states, each channel and its reverse must retain its own entropy increment. The tilted off-diagonal entry is then the sum
$\sum_\nu k_{ab}^{(e),\nu}\exp\left[-s\xi_{ab}^{(e),\nu}\right]$.
Summing rates before taking their entropy ratios generally changes the observable. For example, opposing channels of a two-state pump can sustain entropy production even when the aggregated state process is reversible~\cite{Kaneko2026KineticLocking}.

Let $\mathsf K_{e,s}$ have off-diagonal entries
$k^{(e)}_{ab}e^{-s\xi^{(e)}_{ab}}$ and the original escape diagonal.
For real $s$, denote its Perron eigenvalue by $\lambda_{e,0}(s)$ and
its positive right and left vectors by $v_{e,s}$ and $\ell_{e,s}$,
normalized by $\sum_a\ell_{e,s}(a)v_{e,s}(a)=1$.
The full tilted generator is
$L_s=L_0+\sum_e g_e\mathsf K_{e,s}$.
The positive embedding
\begin{equation}
 (\iota_s f)(z,h)=f(z)\prod_e v_{e,s}(h_e)
 \label{eq:hidden-general-lift}
\end{equation}
obeys the exact intertwining
\begin{equation}
 L_s\iota_s=\iota_s\overline L_s,\qquad
 \overline L_s=L_0+\sum_e g_e(z)\lambda_{e,0}(s).
 \label{eq:hidden-general-reduction}
\end{equation}
The reduced operator is self-adjoint in $L^2(\pi_0)$, whose inner
product is $\sum_z\pi_0(z)f(z)^*g(z)$.
Its positive Perron eigenfunction $\varphi_s$ lifts to a positive
full eigenvector and hence determines the full SCGF.
For identical hidden generators and $g_e=1-\eta_eP_e$, the Ising
couplings are $-\lambda_0(s)\eta_e$.

The driven visible rates and full stationary probability are
\begin{align}
 k_0^{\rm D}(z,z')&=k_0(z,z')
             \frac{\varphi_s(z')}{\varphi_s(z)},\nonumber\\
 \pi_s^{\rm D}(z,h)&=
 \frac{\pi_0(z)\varphi_s(z)^2}
      {\sum_{z'}\pi_0(z')\varphi_s(z')^2}
 \prod_e\ell_{e,s}(h_e)v_{e,s}(h_e).
 \label{eq:hidden-general-driven}
\end{align}
Thus the visible driven process remains autonomous. Hidden driven rates
are $g_e k^{(e)}_{ab}e^{-s\xi^{(e)}_{ab}}
v_{e,s}(b)/v_{e,s}(a)$.
With $R_e=\operatorname{diag}\rho_e$, entropy reversal gives
$\mathsf K_{e,s}^{\mathsf T}=R_e\mathsf K_{e,1-s}R_e^{-1}$.
At the midpoint the whole operator has a weighted symmetric
representation, and the full driven process is reversible.

Triangularizing each hidden matrix separately proves a stronger
spectral statement. In their tensor-product basis the full generator
is block-upper-triangular, with visible diagonal blocks
$L_0+\sum_e g_e\lambda_{e,j_e}(s)$.
Its spectrum, with algebraic multiplicity, is their union, even if a
hidden matrix is defective. Let $E_0$ be the lowest energy of
$-\overline L_s$. Conjugating $L_0$ by $\sqrt{\pi_0}$ and taking the
Hermitian part shows that every block energy satisfies
\begin{equation}
 \operatorname{Re}E\ge E_0+
 \sum_{e:j_e\ne0}g_{e,\min}
 \operatorname{Re}\left[\lambda_{e,0}-\lambda_{e,j_e}\right].
 \label{eq:hidden-general-gap}
\end{equation}
Every indicated difference is positive. Identical finite hidden
generators and a uniform lower gate bound therefore separate all
non-Perron hidden sectors uniformly in size. Below that threshold
the full spectrum and algebraic multiplicities equal those of the
self-adjoint visible block, so hidden-induced Jordan blocks cannot
occur there. Higher sectors can retain Jordan structure away from
the midpoint.

For example, in fixed inverse-time units the bidirected generator
\begin{equation}
 \mathsf K=\frac15
 \begin{pmatrix}-9&6&3\\4&-11&7\\5&5&-10\end{pmatrix}
 \label{eq:hidden-Jordan-example}
\end{equation}
has a uniform stationary state and eigenvalues $0,-3,-3$.
The eigenspace at $-3$ is one-dimensional, so its relaxation includes
a term proportional to $t e^{-3t}$. With constant gates this hidden
Jordan block remains in the full process. Detailed balance guarantees
diagonalizability through a symmetric similarity transform, but
diagonalizability alone does not imply detailed balance.

For a fixed visible history, however, the exact hidden factor is
\begin{equation}
 \rho_e^{\mathsf T}
 \exp\!\left[\mathsf K_{e,s}\int_0^T g_e(z_t)\,dt\right]\mathbf1.
 \label{eq:hidden-general-finite-time}
\end{equation}
It need not reduce to a single exponential of the integrated gate.
The original circulant example has the additional constant-vector
property that makes its finite-time contraction scalar. The extension
above asserts exact SCGF and driven-process closure, not that stronger
identity for every hidden generator.

\subsection{An explicit cycle with nonuniform residence times}
\label{app:hidden-noncirculant}

Choose any positive probability vector
$\rho=(\rho_0,\rho_1,\rho_2)$ and set
\begin{equation}
 k_{h,h+1}=\frac{r}{3\rho_h}e^{A/2},\qquad
 k_{h,h-1}=\frac{r}{3\rho_h}e^{-A/2}.
 \label{eq:hidden-residence-rates}
\end{equation}
Stationary fluxes are $re^{\pm A/2}/3$, so the stationary distribution
is $\rho$, the increments in Eq.~\eqref{eq:hidden-general-mark} are
$\pm A$, and $\sigma,a_h$ retain Eq.~\eqref{eq:means}.
Nonuniform $\rho$ breaks hidden translation symmetry.

Put $R=\operatorname{diag}\rho$ and $d=2\cosh(A/2)+1$.
The midpoint hidden cost $\kappa_\rho=-\lambda_0(1/2)>0$, with $A>0$ as
in the original model, is the
lowest eigenvalue of 
\begin{equation}
 \mathcal K_\rho=\frac r3 R^{-1/2}
       (dI-\mathbf1\mathbf1^{\mathsf T})R^{-1/2}.
 \label{eq:hidden-residence-matrix}
\end{equation}
Here $\mathbf1=(1,1,1)^{\mathsf T}$ and
$\mathbf1\mathbf1^{\mathsf T}$ is the three-by-three all-ones matrix.
The rank-one determinant formula reduces its eigenvalue problem to
\begin{equation}
 \sum_{h=0}^2\frac1{d-q_\rho\rho_h}=1,\quad
 q_\rho=\frac{3\kappa_\rho}{r},\quad
 0<q_\rho<\frac d{\max_h\rho_h}.
 \label{eq:hidden-residence-secular}
\end{equation}
The left side increases from $3/d<1$ to infinity, establishing
a unique root in this interval. The normalized trial vector
$\sqrt\rho$ gives
\begin{equation}
 \kappa_\rho\le2r\left[\cosh(A/2)-1\right]=\kappa_{\rm unif},
 \label{eq:hidden-residence-bound}
\end{equation}
with equality precisely for uniform $\rho$.
Indeed, equality would require $\sqrt\rho$ to be an eigenvector
of Eq.~\eqref{eq:hidden-residence-matrix}, which forces all $\rho_h$
to coincide.
For $\rho_h=(1+\epsilon w_h)/3$, $\sum_h w_h=0$,
expanding Eq.~\eqref{eq:hidden-residence-secular} gives
\begin{equation}
 \kappa_\rho=\kappa_{\rm unif}
 -\frac{\kappa_{\rm unif}^2}{9r}\epsilon^2\sum_h w_h^2
 +O(\epsilon^3).
 \label{eq:hidden-residence-expansion}
\end{equation}
All parameters except $\epsilon$ are fixed in this analytic expansion.

Fixing any one nonuniform $\rho$ and varying the visible gates therefore
retains identical full stationary states and mean diagnostics across
the compared models. Their Ising coupling is $\eta\kappa_\rho$.
When $\kappa_\rho>\gamma$, varying $0\leq\eta<1$ again gives opposite
phases with identical ordinary diagnostics. The matched-phase mechanism
therefore survives without hidden translation symmetry.
Varying $\rho$ itself preserves the mean entropy, activity and cycle
affinity but changes the hidden stationary distribution, which is a
different observational comparison.

\section{Quenched gates and the entropy source}
\label{app:random-gates}
Periodic and random gates share a geometric-mean phase criterion, but
have different critical dynamics. For a positive Ising chain with bonds
$J_i$ and fields $h_i$, the zero-energy boundary Majorana recurrence is
$a_{i+1}=(h_i/J_i)a_i$. Its squared normalization gives the surface
magnetization with a remote boundary spin fixed,
\begin{equation}
 m_{\rm surf}(L)=\left[1+\sum_{\ell=1}^{L-1}
 \prod_{i=1}^{\ell}(h_i/J_i)^2\right]^{-1/2}.
 \label{eq:surface-product}
\end{equation}
This exact boundary formula~\cite{IgloiRieger1998RandomWalks} concerns
a symmetry-breaking boundary condition, not a one-point mean on a
symmetric finite ring. For a finite periodic cell and $h_i=\gamma$,
its convergence gives Eq.~\eqref{eq:periodic-criterion}; the dual
recurrence identifies the disordered side. For bounded positive iid
gates, the logarithm of its product has drift
\begin{equation}
 \mu_s=\log\frac{\gamma}{J_0(s)}-\E_{\rm dis}\left[\log\eta\right].
 \label{eq:random-drift}
\end{equation}
The zero-drift case is the random critical boundary.

On the disordered side $\mu_s>0$, arbitrarily long intervals with
locally strong bonds can relax slowly. When the gate distribution has
positive probability of $J_0(s)\eta>\gamma$, the Griffiths dynamical
exponent is the positive solution of the established relation
\begin{equation}
 \E_{\rm dis}\left[\left(\frac{J_0(s)\eta}{\gamma}\right)^{1/z_{\rm G}}\right]=1
 \label{eq:Griffiths-exponent}
\end{equation}
\cite{IgloiJuhaszRieger1999Griffiths}.
Put $v_\eta=\Var_{\rm dis}(\log\eta)>0$. Expanding the log moment-generating
function at small $1/z_{\rm G}$ gives
$z_{\rm G}=v_\eta/(2\mu_s)+O(1)$. At a tuned midpoint,
$J_0(1/2)\exp(\E_{\rm dis}\left[\log\eta\right])=\gamma$, so
\begin{align}
 \mu_s&=\frac{A^2(s-1/2)^2}{2\left[\cosh(A/2)-1\right]}
       +O((s-1/2)^4),\nonumber\\
 z_{\rm G}(s)&\sim
 \frac{v_\eta\left[\cosh(A/2)-1\right]}{A^2(s-1/2)^2}.
 \label{eq:entropy-Griffiths}
\end{align}
At the critical point itself, the strong-disorder theory has activated
rather than finite-$z_{\rm G}$ scaling~\cite{Fisher1995}.

Equation~\eqref{eq:random-stochastic-correlation} identifies the
random-Ising excitation energies with physical relaxation rates.
The size-independent separation of hidden sectors retains the
sufficiently low microscopic modes.
Take the long-time limit for each fixed gate sample before its
thermodynamic limit. Averaging $Z_{N,T}$ before taking its logarithm
defines a different, annealed ensemble. Activated scaling describes
typical relaxation scales or fixed disorder quantiles; disorder averages
can probe other tails. Specially correlated disorder need not share
these exponents~\cite{ShiraiTanaka2021}. Neither the activated gap nor
Eq.~\eqref{eq:entropy-Griffiths} determines a disordered entropy-cumulant
exponent without the corresponding energy-operator matrix elements.

\subsection{Exact energy response of a quenched sample}
\label{app:random-energy-response}
The activated gap does not by itself fix an entropy susceptibility.
The latter also contains the matrix elements of the weighted alignment
operator. For one periodic sample, define
\begin{equation}
 \begin{gathered}
 \mathcal B_\eta=\sum_i\eta_iZ_iZ_{i+1},\\
 \mathscr E_N(x)=\min\operatorname{spec}
 \left[-\gamma\sum_iX_i-x\mathcal B_\eta\right],
 \end{gathered}
 \label{eq:random-bond-energy}
\end{equation}
where $x>0$ scales all bonds and has units of inverse time. The gates
are held fixed in every derivative. At $x_*=J_0(1/2)$, put
$\chi_\eta(N)=-\mathscr E_N''(x_*)\geq0$. The exact midpoint entropy
cumulant rates for that sample are
\begin{align}
 \lim_{T\to\infty}\frac{\kappa_2}{T}
 &=2rA^2\left[N+\mathscr E_N'(x_*)\right],\nonumber\\
 \lim_{T\to\infty}\frac{\kappa_4}{T}
 &=2rA^4\left[N+\mathscr E_N'(x_*)\right]
       +12r^2A^4\chi_\eta(N).
 \label{eq:random-entropy-response}
\end{align}
These identities retain the scalar trajectory cost. They follow from
$\mathscr E_N'=-\langle\mathcal B_\eta\rangle$ and the chain rule for
$J_0(s)=2rd_s$. In particular the required averaged response is
$\E_{\rm dis}\left[\chi_\eta(N)/N\right]$, with a quenched disorder average.

An $N$-dimensional singular-value problem determines this response  through an exact finite-size formula. Numerical evaluation of its
singular values introduces the separately checked arithmetic error. Let $S_{\rm AP}$ translate a one-particle coordinate, with
$(S_{\rm AP})_{i,i+1}=1$ for $i<N$ and
$(S_{\rm AP})_{N,1}=-1$, and set
\begin{equation}
 M(x)=\gamma I-x\operatorname{diag}(\eta_i)S_{\rm AP}.
 \label{eq:random-Majorana-matrix}
\end{equation}
The antiperiodic sign is required by the positive, even-parity spin
ground state. Since
$\det M=\gamma^N+x^N\prod_i\eta_i>0$, its vacuum parity cannot change
along the positive bond scale. If $M=U\operatorname{diag}(\varsigma_a)V^{\mathsf T}$
is a real singular-value decomposition, then
$\mathscr E_N=-\sum_a\varsigma_a$. Write
$\mathsf F=U^{\mathsf T}M'V$. Its first two derivatives are
\begin{equation}
 \mathscr E_N'=-\operatorname{tr}\mathsf F,\qquad
 \chi_\eta(N)=\sum_{a<b}
 \frac{(\mathsf F_{ab}-\mathsf F_{ba})^2}{\varsigma_a+\varsigma_b}.
 \label{eq:random-curvature-SVD}
\end{equation}
To obtain the second expression, use the symmetric dilation
$\left(\begin{smallmatrix}0&M\\M^{\mathsf T}&0\end{smallmatrix}\right)$,
whose eigenvalues are $\pm\varsigma_a$. Write $\boldsymbol u_a,\boldsymbol v_a$ for columns of $U,V$.
The normalized dilation eigenvectors are
$|a,\pm\rangle=(\boldsymbol u_a,\pm\boldsymbol v_a)^{\mathsf T}/\sqrt2$.
Second-order perturbation of the sum of positive eigenvalues pairs
and cancels contributions between two positive modes. The positive--negative matrix element is
$(\mathsf F_{ba}-\mathsf F_{ab})/2$, which gives
Eq.~\eqref{eq:random-curvature-SVD}. The expression remains valid at
degenerate positive singular values and has no small difference
denominators. Equivalently,
$\mathsf F_{ab}-\mathsf F_{ba}=-(\gamma/x)\left[(U^{\mathsf T}V)_{ab}-(U^{\mathsf T}V)_{ba}\right]$.

The response involves even-parity excitations of energy
$\Omega_{ab}=2(\varsigma_a+\varsigma_b)$, with weight
$W_{ab}=(\mathsf F_{ab}-\mathsf F_{ba})^2$. The stationary driven correlation is
\begin{equation}
 \operatorname{Cov}_s\left[\mathcal B_\eta(t),\mathcal B_\eta(0)\right]
 =\sum_{a<b}W_{ab}e^{-\Omega_{ab}|t|},
 \label{eq:random-energy-correlation}
\end{equation}
where $M$ and its derivatives are evaluated at the selected bond scale.
Consequently its integrated occupation-time variance is
\begin{equation}
 2\sum_{a<b}W_{ab}\left[
 \frac{T}{\Omega_{ab}}-
 \frac{1-e^{-\Omega_{ab}T}}{\Omega_{ab}^2}\right].
 \label{eq:random-integrated-response}
\end{equation}
The conditional counting argument then fixes the fourth cumulant at
every finite observation time under stationary midpoint preparation.
An exceptionally small odd-parity gap is absent from these expressions.
Even a small pair energy contributes little when its matrix element is
suppressed. Renormalized energy operators are therefore needed in a
random critical theory, as earlier analyses of local and boundary energy
correlations emphasize~\cite{IgloiJuhaszRieger1999Griffiths,RefaelFisher2004Energy}.

The weak-disorder random-mass continuum theory has a smooth but
nonanalytic averaged ground energy~\cite{BunderMcKenzie1999}. At fixed
nonzero disorder, its energy curvature stays finite at criticality.
This permits a finite fourth-cumulant density despite activated
relaxation. The continuum result does not establish interchange of
disorder averages, derivatives and thermodynamic limits for every
bounded lattice distribution. Equation~\eqref{eq:random-curvature-SVD}
makes that stronger lattice question directly testable without numerical
differentiation of nearby extensive energies.

\subsection{Quenched sampling and checks of the energy response}
\label{app:disorder-numerics}
We evaluate  the lattice susceptibility
$\chi_\eta(N)=-\mathscr E_N''(x_*)$ in
Eq.~\eqref{eq:random-curvature-SVD}, with every gate held fixed during
differentiation. We choose time units with $\gamma=1$ and sample independent gates from 
\begin{equation}
 \log(4\eta_i)\overset{\mathrm{iid}}\sim\mathrm{Unif}\left[-w,w\right],\quad
 w\in\{0.3,0.8\},\quad x_*=4.
 \label{eq:disorder-numerical-ensemble}
\end{equation}
All gates lie strictly between zero and one. Since
$\E_{\rm dis}\left[\log\eta_i\right]=-\log4$, both ensembles satisfy the exact
geometric-mean critical criterion. We retain each sample's log drift;
subtracting its sample mean would introduce correlations and change
the ensemble. For either width, any fixed affinity $A>0$ can be used
with $r=2/\left[\cosh(A/2)-1\right]$. The ordinary entropy-production rate,
hidden activity and affinity are then identical, as is the entire
unbiased visible path process. Their higher ordinary fluctuation
statistics need not coincide.

For each width we use $512$ independent samples at each of
$N=64,128,256$, $256$ at $N=512$ and $64$ at $N=1024$, for $3712$
production samples in total. A separate $512$-sample pilot is not
pooled into these averages. No sample is removed or trimmed.
The plotted response is $\chi_\eta(N)/N$, in the chosen time units.
The exact clean comparison is
$\eta_0^2 S_N/(\gamma N)$ at $\eta_0=1/4$.

Error bars denote one empirical standard error of the mean. We also
inspect medians, upper quantiles, maxima, equal-sized independent
batches and percentile-bootstrap intervals from $5000$ resamples.
At $N=1024$ the $95\%$ bootstrap intervals are
$\left[0.08539,0.08775\right]$ for $w=0.3$ and $\left[0.04962,0.05056\right]$ for $w=0.8$.
The two $w=0.3$ batches at that size differ by $2.27$ combined standard
errors; the other batch differences are below $1.35$.
The largest $\lceil0.01n_{\rm samp}\rceil$ samples account for
between $1.33\%$ and $1.85\%$ of the observed response sums. At the
smallest sample count this selects one of $64$ samples, so it should
not be interpreted as an exactly resolved upper one-percent tail.
These diagnostics quantify sampling variation in the observed
ensembles. They do not bound contributions from rare realizations
absent from them, and no universal finite-size correction is fitted.

The response is evaluated by a real singular-value decomposition,
without finite differences of extensive energies. Three distinct
checks test the calculation. First, the clean expression agrees
through $N=256$, and direct $2^N$-dimensional spin Hamiltonians at
$N=3,\ldots,7$ agree in energy, first derivative and curvature.
Second, the symmetric-dilation formula at $70$ and $100$ decimal
digits agrees to relative errors below $2\times10^{-65}$ for two
moderately disordered fixtures at $N=8,12$. These are independent
precision checks of small samples, not high-precision calculations
of every production realization. Third, the full spin-space
correlation function at $N=4,6,8$ agrees with
Eq.~\eqref{eq:random-energy-correlation} at
$t=0,0.25,2,15$ to absolute error below $10^{-15}$.

For every production sample, we compare the two numerator expressions
following Eq.~\eqref{eq:random-curvature-SVD} and monitor the
singular-value residual. Their maximum relative discrepancies are
$8.3\times10^{-16}$ and $3.4\times10^{-15}$, respectively.
Conditioning is assessed through the sum of the two smallest singular
values, since the response creates pairs. This sum remains at least
$3.6\times10^8$ times the double-precision spacing $\epsilon_{\rm mach}=2^{-52}\simeq2.22\times10^{-16}$  multiplied by the largest
singular value. Comparisons of two singular-value algorithms on
benchmark samples through $N=1024$ agree in curvature to relative error
below $10^{-15}$.

The observed plateaus therefore concern the finite-size random
lattice response, with arithmetic and sampling errors checked
separately. The weak-disorder continuum result~\cite{BunderMcKenzie1999}
explains how the mean energy can be smooth at a critical point. Establishing the lattice limit requires control of rare realizations
when taking disorder averages and derivatives. The present data support a finite limiting
fourth-cumulant density for Eq.~\eqref{eq:disorder-numerical-ensemble};
they do not extend that conclusion to every bounded disorder
distribution.

Numerical linear algebra uses SciPy~\cite{SciPy2020};
the arbitrary-precision checks use mpmath~\cite{Mpmath}.

\subsection{A resummed continuum response and its error bound}
\label{app:disorder-borel}
The continuum function in Eq.~\eqref{eq:disorder-borel} follows from
the low-energy density of states of the weak-disorder random-mass
theory. Up to analytic backgrounds and an overall dimensional factor,
its digamma representation is
\begin{equation}
 \Phi_{\rm dis}(\upsilon)
 =\log(2\upsilon)+\psi_{\rm dg}\bigl(1/(2\upsilon)\bigr)+\upsilon,
 \qquad \upsilon>0.
 \label{eq:disorder-digamma}
\end{equation}
Here $\psi_{\rm dg}(x)=\mathrm d\log\Gamma(x)/\mathrm dx$ is the
digamma function, distinct from the trajectory SCGF $\psi_N$.
The ground-energy construction and its Bernoulli expansion are due to
Bunder and McKenzie~\cite{BunderMcKenzie1999}. The added linear term
removes an analytic contribution. Matching the neighboring Griffiths
exponent fixes the leading continuum coordinate,
\begin{equation}
 \upsilon\simeq\frac{\mu_s}{v_\eta}
 =\frac{A^2(s-1/2)^2}{2v_\eta\left[\cosh(A/2)-1\right]}
 +O((s-1/2)^4).
 \label{eq:disorder-continuum-coordinate}
\end{equation}
The first relation is a scaling-field identification in the weak-disorder
limit. It is not an exact finite-disorder relation between microscopic
bonds and a bare continuum mass.

The integral representations of the digamma function~\cite{DLMF2026}
give Eq.~\eqref{eq:disorder-borel} and the complementary Binet form
\begin{equation}
 \Phi_{\rm dis}(\upsilon)=-8\upsilon^2\int_0^\infty
 \frac{t\dd t}{(1+4\upsilon^2t^2)(e^{2\pi t}-1)}.
 \label{eq:disorder-binet}
\end{equation}
Expanding the kernel of the first integral defines the asymptotic series
\begin{equation}
 \Phi_{\rm dis}(\upsilon)
 \sim-\sum_{n\geq1}\frac{B_{2n}}{2n}(2\upsilon)^{2n},
 \label{eq:disorder-bernoulli}
\end{equation}
where the Bernoulli numbers are defined by
$t/(e^t-1)=\sum_{n\geq0}B_n t^n/n!$.
The coefficients grow factorially. Equation~\eqref{eq:disorder-borel}
is their Borel--Laplace sum, with the convention that the Laplace integral
has no additional $1/(2\upsilon)$ prefactor. Dividing the coefficient of
$(2\upsilon)^{2n}$ by $(2n-1)!$ reconstructs its meromorphic kernel.
The nearest poles are $t=\pm2\pi i$, away from the positive integration
ray, so the sum is unambiguous for real $\upsilon>0$.

The second integral controls truncation without assuming convergence.
Expanding its denominator through $M$ terms gives
\begin{align}
 \Phi_{\rm dis}-\Phi_{{\rm dis},M}
 &=2(-1)^{M+1}(2\upsilon)^{2M+2}\nonumber\\
 &\quad\times\int_0^\infty
 \frac{t^{2M+1}\dd t}{(1+4\upsilon^2t^2)(e^{2\pi t}-1)},
 \nonumber\\[-2pt]
 \left|\Phi_{\rm dis}-\Phi_{{\rm dis},M}\right|
 &\leq\frac{|B_{2M+2}|}{2M+2}(2\upsilon)^{2M+2},
 \label{eq:disorder-remainder}
\end{align}
where $\Phi_{{\rm dis},M}$ denotes the sum in
Eq.~\eqref{eq:disorder-bernoulli} through $n=M$.
The remainder has the sign of the first omitted term. Stirling's formula
places the least term at $M\simeq\pi/(2\upsilon)$, with a bound of order
$\sqrt{\upsilon}\,e^{-\pi/\upsilon}$. Under the quadratic source this is
an error smaller than every power of $s-1/2$, of order
$\exp\left[-\mathrm{const}/(s-1/2)^2\right]$ up to algebraic factors.
This exponential scale quantifies the optimal truncation error of the
continuum response. Interpreting it as a microscopic transition rate
would require a separate dynamical derivation.

The leading term $-\upsilon^2/3$ contributes a finite fourth entropy
derivative under Eq.~\eqref{eq:disorder-continuum-coordinate}. Higher
terms determine the large-order response, with the bound above controlling
what any finite truncation omits. Analytic backgrounds and nonlinear
scaling fields remain necessary for actual lattice amplitudes.

\section{Exact finite-rate memory after entropy selection}
\label{app:selected-memory}

This appendix derives the visible correlation and its memory kernel for
the two-cycle construction of Sec.~\ref{sec:selected-memory}. The
calculation concerns the stationary midpoint driven process, which
describes the interior of a long entropy-selected trajectory. It does
not identify that stationary process with every finite-duration
conditioned ensemble, whose endpoints can matter
\cite{ChetriteTouchetteControl2015}.

\subsection{When selection preserves visible autonomy}
\label{app:common-perron-line}

The factorization $g_e(z)\mathsf K_e$ is sufficient for the autonomous
selected dynamics, but the relevant condition is weaker and has a
direct eigenvector formulation. Let $z$ range over a finite irreducible
visible graph, and let its generator $L_0$ act by jumps that leave the
hidden state $h$ unchanged. Suppose the full tilted operator has the
form
\begin{equation}
 (L_s f)(z,h)=(L_0 f)(z,h)
                 +\left[\mathsf K_s(z)f(z,\cdot)\right](h),
 \label{eq:common-line-operator}
\end{equation}
where every $\mathsf K_s(z)$ is an irreducible finite tilted hidden
generator, with nonnegative off-diagonal entries. This form applies to
the entropy tilt when the ordinary stationary measure is
$\pi_0(z)\rho(h)$, the visible generator is reversible with respect to
$\pi_0$, and the endpoint contribution is included in the stationary
entropy convention.

\begin{proposition}[Visible autonomy] The full driven process has visible jump rates
independent of $h$ if and only if the matrices $\mathsf K_s(z)$ have
a common positive right eigenvector $b_s$, up to normalization.
That vector is necessarily their Perron vector. When the criterion
holds, the exact reduced tilted operator is
\begin{equation}
 \overline L_s=L_0+\lambda_s(z),\qquad
 \mathsf K_s(z)b_s=\lambda_s(z)b_s.
 \label{eq:common-line-reduction}
\end{equation}
Thus a common positive eigenline, rather than proportionality of
the hidden matrices, is the exact structural requirement for autonomous
visible rates.

\end{proposition}
\begin{proof} Let $\Phi_s(z,h)>0$ be the full right Perron vector.
For an allowed visible edge, the driven rate is
$k_0(z,z')\Phi_s(z',h)/\Phi_s(z,h)$. Independence of $h$, together
with connectivity of the visible graph, implies
$\Phi_s(z,h)=a_s(z)b_s(h)$. Substitution into the full eigenvalue
equation gives
\begin{equation}
 \mathsf K_s(z)b_s=
 \left[\psi(s)-\frac{(L_0a_s)(z)}{a_s(z)}\right]b_s.
\end{equation}
Positivity identifies this eigenvector as the Perron vector of every
hidden matrix. Conversely, given a common positive $b_s$, take the
positive Perron vector $a_s$ of $L_0+\lambda_s(z)$. Their product is
a positive eigenvector of $L_s$, hence its Perron vector. Its visible
Doob rates depend on $z,z'$ alone. \end{proof}

The criterion above concerns autonomy for every hidden preparation,
often called strong lumpability. At a general tilt, a failure of this
property should not by itself be identified with a proof that every
stationary visible projection is non-Markovian. At the entropy
midpoint, reversibility closes that distinction, consistently with
the reversible lumpability result of Burke and Rosenblatt
\cite{BurkeRosenblatt1958}. Let $\mathcal P$ be
conditional expectation onto visible functions in the full driven
stationary measure and let $\mathcal Q=I-\mathcal P$. The full driven
generator $L^{\rm D}$ is self-adjoint in that measure. If its
stationary visible projection were Markov, the compressed conditional
propagators $\mathcal P e^{tL^{\rm D}}\mathcal P$ would form a
semigroup on visible functions. Their first two derivatives at zero
would then obey
\begin{equation}
 0=\mathcal P(L^{\rm D})^2\mathcal P
       -(\mathcal P L^{\rm D}\mathcal P)^2
   =(\mathcal Q L^{\rm D}\mathcal P)^*
     (\mathcal Q L^{\rm D}\mathcal P).
 \label{eq:common-line-reversible-compression}
\end{equation}
Consequently $L^{\rm D}$ leaves the space of visible functions
invariant, which is precisely autonomy of the visible rates.
Conversely, that invariance gives a Markov projection. Therefore,
under the finite midpoint assumptions, the common positive hidden
eigenline is also necessary and sufficient for the stationary visible
projection to be Markov. This last statement uses reversibility;
the general-$s$ autonomy criterion does not require it.

\subsection{Stationarity and the two symmetry sectors}

Each oriented triangle separately has equal incoming and outgoing
total rates at every vertex. Their sum therefore has the uniform
hidden stationary measure for both values of $z$. The visible flips
are also symmetric, so the full stationary measure is $1/10$.
Each triangle has three edges. Summing its forward and reverse
stationary flows gives $6r(1\pm\epsilon z)\cosh(A/2)/5$;
their difference, multiplied by $A$, gives
$6r(1\pm\epsilon z)A\sinh(A/2)/5$. Adding the triangles proves
Eq.~\eqref{eq:memory-matched-means}. The full stationary entropy is
$A$ times the net number of forward hidden jumps, because the uniform
stationary measure has no endpoint contribution and visible jumps have
zero log rate ratio.

Write $c=\cosh(A/2)$ and $d=2c-1$. At the midpoint the off-diagonal
Hamiltonian entries on a hidden edge are minus its turnover
prefactor, while the diagonal entries retain the original escape
rates. Reflection within either hidden triangle is now a symmetry.
In the reflection-even hidden subspace choose $|0\rangle$ together with
the orthonormal vectors
\begin{equation}
 \begin{aligned}
 |S\rangle&=\frac{|1\rangle+|2\rangle+|3\rangle+|4\rangle}{2},\\
 |D\rangle&=\frac{|1\rangle+|2\rangle-|3\rangle-|4\rangle}{2}.
 \end{aligned}
 \label{eq:memory-hidden-basis}
\end{equation}
The visible vectors $|\pm\rangle=(|z=+1\rangle\pm|z=-1\rangle)/\sqrt2$
diagonalize the visible flip. The positive ground vector is even under
the combined visible flip and interchange of the hidden branches.
In the ordered basis $|+\rangle|0\rangle$, $|+\rangle|S\rangle$,
$|-\rangle|D\rangle$, its Hamiltonian is exactly
\begin{equation}
 \mathsf H_{\rm e}
 =r\begin{pmatrix}
 4c&-2&-2\epsilon\\
 -2&d&d\epsilon\\
 -2\epsilon&d\epsilon&d+2\gamma/r
 \end{pmatrix}.
 \label{eq:memory-even-block}
\end{equation}
Let $E_\epsilon$ be its lowest eigenvalue and
$\boldsymbol u=(u_0,u_S,u_D)^{\mathsf T}$ its real normalized ground
vector, oriented so that $u_0,u_S>0$.
The full physical vector is
\begin{equation}
 \Phi_\epsilon(z,h)=\frac1{\sqrt2}
 \left[u_0\delta_{h0}+u_S S_h+z u_D D_h\right].
 \label{eq:memory-positive-vector}
\end{equation}
Strict positivity implies $u_S>|u_D|$ and $u_D^2<1/2$.
Summing $\Phi_\epsilon^2$ over $h$ gives $1/2$ for either visible state.

Multiplication by $z$ maps the ground vector into the opposite
combined-parity sector. In the corresponding ordered basis
$|-\rangle|0\rangle$, $|-\rangle|S\rangle$, $|+\rangle|D\rangle$,
the matrix is
\begin{equation}
 \mathsf H_{\rm o}
 =\mathsf H_{\rm e}+2\gamma\operatorname{diag}(1,1,-1).
 \label{eq:memory-odd-block}
\end{equation}
The diagonal similarity relating the symmetric Hamiltonian and its
Doob generator then gives
\begin{equation}
 C_\epsilon(t)=\boldsymbol u^{\mathsf T}
 e^{-t\mathsf G_\epsilon}\boldsymbol u,
 \qquad
 \mathsf G_\epsilon=\mathsf H_{\rm o}-E_\epsilon I.
 \label{eq:memory-exact-correlation}
\end{equation}
The full finite chain is irreducible, so $\mathsf G_\epsilon$ is positive
definite. If $\Lambda_j$ and $\boldsymbol v_j$ are the eigenvalues and
orthonormal eigenvectors of $\mathsf H_{\rm o}$, Eq.~\eqref{eq:memory-exact-correlation}
is a sum of at most three exponentials with rates
$\Lambda_j-E_\epsilon>0$ and weights
$|\boldsymbol v_j^{\mathsf T}\boldsymbol u|^2$. The cubic eigenproblem
specifies the full time dependence without a perturbative approximation.

\subsection{A strict test for visible memory}

Equation~\eqref{eq:memory-odd-block} and the ground-state equation give
$\mathsf G_\epsilon\boldsymbol u=
2\gamma\operatorname{diag}(1,1,-1)\boldsymbol u$.
The first two derivatives of Eq.~\eqref{eq:memory-exact-correlation}
give $-C_\epsilon'(0^+)=2\gamma(1-2u_D^2)$ and the curvature excess
$16\gamma^2u_D^2(1-u_D^2)$. In the physical rate notation of
Eq.~\eqref{eq:memory-witness}, these identities read
\begin{equation}
 \begin{gathered}
 \langle w_\epsilon\rangle=\gamma(1-2u_D^2),\quad
 \langle w_\epsilon^2\rangle=\gamma^2,\\
 C_\epsilon''(0^+)-\left[C_\epsilon'(0^+)\right]^2
 =4\operatorname{Var}(w_\epsilon).
 \end{gathered}
 \label{eq:memory-rate-variance}
\end{equation}
The second identity follows by cancelling the ratio in
$\Phi_\epsilon(z,h)^2w_\epsilon(z,h)^2$ and then relabelling $z$.
Thus fluctuations of the selected visible flip rate across hidden
states are directly measurable through the correlation curvature.

To prove strict positivity for every $\epsilon\ne0$, suppose instead
that $u_D=0$. The first two components of the eigenvalue equation would
make $(u_0,u_S)$ the positive ground vector of
$r\bigl(\begin{smallmatrix}4c&-2\\-2&d\end{smallmatrix}\bigr)$.
Its eigenvalue is
\begin{equation}
 E_{\rm h}=\frac r2
 \left[6c-1-\sqrt{(2c+1)^2+16}\right]>0
 \qquad(c>1).
 \label{eq:memory-base-cost}
\end{equation}
The third component would require $-2u_0+d u_S=0$, whereas the
second requires $r(-2u_0+d u_S)=E_{\rm h}u_S>0$.
This contradiction proves $u_D\ne0$. Together with $u_D^2<1/2$, it
proves the strict inequality. At $A=0$, entropy is identically zero
and this argument appropriately ceases to apply. At $\epsilon=0$,
the hidden process is independent of $z$ and
$C_0(t)=e^{-2\gamma t}$. The endpoints $|\epsilon|=1$, where some
rates vanish, are outside the assumed irreducible family.

The positive spectral representation also bounds the full correlation.
Its decay rates have mean $a_\epsilon$ and strictly positive variance
given by Eq.~\eqref{eq:memory-witness}. Strict convexity of
$x\mapsto e^{-tx}$ therefore gives
\begin{equation}
 C_\epsilon(t)>e^{-a_\epsilon t}>e^{-2\gamma t},
 \qquad t>0,\quad\epsilon\ne0.
 \label{eq:memory-persistence-bound}
\end{equation}
The first inequality distinguishes memory from a reduction of the
mean selected flip rate, and the second compares with the ordinary
visible process.

\subsection{The exact memory kernel and its controlled onset}

The scalar resolvent of Eq.~\eqref{eq:memory-exact-correlation} can be
evaluated by projecting onto $\boldsymbol u$ and its orthogonal
complement. Define
\begin{equation}
 \begin{gathered}
 \mathsf P=\boldsymbol u\boldsymbol u^{\mathsf T},\quad
 \mathsf Q=I-\mathsf P,\\
 a_\epsilon=\boldsymbol u^{\mathsf T}\mathsf G_\epsilon\boldsymbol u,\quad
 \boldsymbol b=\mathsf Q\mathsf G_\epsilon\boldsymbol u.
 \end{gathered}
\end{equation}
Let $\mathsf G_\perp$ denote the restriction of
$\mathsf Q\mathsf G_\epsilon\mathsf Q$ to the two-dimensional space
orthogonal to $\boldsymbol u$. Block elimination, or the Schur
complement of this restricted matrix, gives for $p>0$
\begin{equation}
 \widetilde C_\epsilon(p)
 =\frac{1}{p+a_\epsilon-
 \boldsymbol b^{\mathsf T}(pI+\mathsf G_\perp)^{-1}\boldsymbol b}\,.
 \label{eq:memory-resolvent}
\end{equation}
Here the components of $\boldsymbol b$ are understood in any
orthonormal basis of that complement. Inverting the Laplace transform
proves Eq.~\eqref{eq:memory-main-kernel}, with
\begin{equation}
 \mathfrak M_\epsilon(t)
 =\boldsymbol b^{\mathsf T}e^{-t\mathsf G_\perp}\boldsymbol b,
 \quad
 \mathfrak M_\epsilon(0)
 =16\gamma^2u_D^2(1-u_D^2).
 \label{eq:memory-exact-kernel}
\end{equation}
The restricted matrix is positive definite. Its spectral decomposition
therefore proves the two-exponential form and nonnegative weights.
The positive convolution term in Eq.~\eqref{eq:memory-main-kernel}
follows from eliminating amplitudes in dissipative evolution;
it is consistent with the minus sign of the self-energy in
Eq.~\eqref{eq:memory-resolvent}. This exact correlation projection
uses the geometry underlying Mori's memory construction
\cite{Mori1965}; it does not postulate a stochastic equation for $z$.

For a controlled small-$\epsilon$ expansion, define
\begin{equation}
 \begin{gathered}
 \Delta_{\rm h}=rd-E_{\rm h}>0,\\
 b_{\rm h}=\left[1+\left(\frac{\Delta_{\rm h}}{2r}\right)^2\right]^{-1/2},\\
 v_{\rm h}=E_{\rm h}b_{\rm h},\qquad
 \chi_{\rm h}=\frac{v_{\rm h}}{\Delta_{\rm h}+2\gamma}.
 \end{gathered}
 \label{eq:memory-small-parameters}
\end{equation}
The unperturbed hidden ground vector has components
$b_{\rm h}(\Delta_{\rm h}/2r,1)$. The perturbation couples it to the
branch-odd mode with matrix element $v_{\rm h}$; the energy denominator
includes both hidden relaxation and the visible flip, giving
\begin{align}
 E_\epsilon&=E_{\rm h}
 -\epsilon^2\frac{v_{\rm h}^2}{\Delta_{\rm h}+2\gamma}
 +O(\epsilon^4),\nonumber\\
 u_D&=-\epsilon\chi_{\rm h}+O(\epsilon^3),\nonumber\\
 a_\epsilon&=2\gamma-4\gamma\epsilon^2\chi_{\rm h}^2
 +O(\epsilon^4),\nonumber\\
 \mathfrak M_\epsilon(t)
 &=16\gamma^2\epsilon^2\chi_{\rm h}^2e^{-\Delta_{\rm h}t}
 +O(\epsilon^4).
 \label{eq:memory-small-expansion}
\end{align}
The scalar quantities are even in $\epsilon$ by branch interchange.
Finite-dimensional analytic perturbation theory controls the
remainders at fixed $r,A,\gamma>0$; the kernel remainder is uniform
on compact time intervals. These expressions remain regular when
$\Delta_{\rm h}=2\gamma$, where expansions of individual correlation
poles would need degenerate perturbation theory. They do not assume
a uniform many-body gap or establish a critical scaling limit.

For a directly visible example, take $r=\gamma=1$,
$A=2\operatorname{arcosh}2$ and $\epsilon=1/2$.
The three decay rates are approximately $0.315916$, $2.954914$,
and $9.087043$, with weights $0.559711$, $0.438938$, and $0.001351$.
Consequently $C_\epsilon(2)=0.298743$, while the exponential with
the same initial slope is $0.051188$. The ordinary visible
correlation remains $e^{-2t}$, which equals $0.018316$ at $t=2$.
This contrast separates memory from a change of a single effective
flip rate.

\bibliography{references}
\end{document}